\documentclass[11pt]{article}
\usepackage{amsfonts}
\usepackage{amssymb}
\usepackage{amstext}
\usepackage{amsmath}
\usepackage{xspace}
\usepackage{theorem}
\usepackage{thmtools}
\usepackage{graphicx}
\usepackage{graphics}
\usepackage{colordvi}
\usepackage{color}
\usepackage{subcaption}

\usepackage{comment}

\usepackage{paralist}
\usepackage{bm}

\advance \oddsidemargin by -0.8in

\newenvironment{proofof}[1]{\noindent{\bf Proof of #1.}}%
        {\hspace*{\fill}$\Box$\par\vspace{4mm}}

\newcommand{\height}{\operatorname{height}}

\newcommand{\leafBST}{{\sc LeafBST}\xspace}
\newcommand{\recBST}{{\sc RecursiveBST}\xspace}
\newcommand{\topset}{\operatorname{top}}

\newcommand{\ceil}[1]{\ensuremath{\left\lceil#1\right\rceil}}
\newcommand{\floor}[1]{\ensuremath{\left\lfloor#1\right\rfloor}}

\newcommand{\paren}[1]{\left ( #1 \right ) }

\newcommand{\NP}{\mbox{\sf NP}}

\newcommand{\DTIME}{\mbox{\sf DTIME}}

\newcommand{\opt}{\mbox{\sf OPT}}

\newcommand{\minsat}{\mbox{\sf Min-Sat}\xspace}

\newcommand{\set}[1]{\left\{ #1 \right\}}

\newcommand{\be}{\begin{enumerate}}
\newcommand{\ee}{\end{enumerate}}
\newcommand{\bd}{\begin{description}}
\newcommand{\ed}{\end{description}}
\newcommand{\bi}{\begin{itemize}}
\newcommand{\ei}{\end{itemize}}

\definecolor{ForestGreen}{rgb}{0.1333,0.5451,0.1333}
\definecolor{DarkRed}{rgb}{0.8,0,0}
\definecolor{Red}{rgb}{1,0,0}
\usepackage[linktocpage=true,
pagebackref=true,colorlinks,
linkcolor=DarkRed,citecolor=ForestGreen,
bookmarks,bookmarksopen,bookmarksnumbered]
{hyperref}
\usepackage{cleveref}

\declaretheorem[numberwithin=section]{theorem}
\declaretheorem[numberlike=theorem]{lemma}

\declaretheorem[numberlike=theorem]{question}
\declaretheorem[numberlike=theorem]{corollary}
\declaretheorem[numberlike=theorem]{definition}
\declaretheorem[numberlike=theorem]{claim}
\declaretheorem[numberlike=theorem]{observation}
\declaretheorem[numberlike=theorem]{problem}

\newenvironment{proof}{\par \smallskip{\bf Proof:}}{\hfill\stopproof}
\def\stopproof{\square}
\def\square{\vbox{\hrule height.2pt\hbox{\vrule width.2pt height5pt \kern5pt
\vrule width.2pt} \hrule height.2pt}}

\newenvironment{prog}[1]{
\begin{minipage}{5.8 in}
{\sc\bf #1}
\begin{enumerate}}
{
\end{enumerate}
\end{minipage}
}

\newcommand{\program}[2]{\vspace{2mm}\fbox{\vspace{2mm}\begin{prog}{#1} #2 \end{prog}\vspace{2mm}}\vspace{2mm}}

\renewcommand{\phi}{\varphi}

\newcommand{\half}{\ensuremath{\frac{1}{2}}}

\newcommand{\poly}{\operatorname{poly}}

\newcommand{\Z}{\ensuremath{\mathbb Z}}

\newcommand{\prob}[1]{\text{\bf Pr}\left [#1\right]}

\newcommand{\pset}{{\mathcal P}}

\newcommand{\cset}{{\mathcal C}} 
 
\newcommand{\lset}{{\mathcal L}} 
\newcommand{\rset}{{\mathcal R}}

\newcommand{\tset}{{\mathcal T}}
\newcommand{\iset}{{\mathcal{I}}}
\newcommand{\bset}{{\mathcal{B}}}

\newcommand{\xset}{{\mathcal{X}}}

\newcommand{\hX}{\hat X}

\newcommand{\hF}{\hat F}
\newcommand{\tX}{\tilde X}
\newcommand{\tcset}{\tilde{\mathcal{C}}}
\newcommand{\tlset}{\tilde{\mathcal{L}}}
\newcommand{\hlset}{\hat{\mathcal{L}}}
\newcommand{\tsigma}{\tilde{\sigma}}
\newcommand{\hsigma}{\hat{\sigma}}

\newcommand{\cGB}{\mathsf{cGB}}
\newcommand{\GB}{\mathsf{GB}}
\newcommand{\XXWB}{\mathsf{GB}}
\newcommand{\XWB}{\mathsf{cGB}}
\newcommand{\WB}{\mathsf{WB}}
\newcommand{\BRS}{\ensuremath{\mathsf{BRS}}}
\newcommand{\ESBRS}{\mbox{\sf{ES-BRS}}}
\newcommand{\cost}{\mathsf{cost}}
\newcommand{\size}{N}
\newcommand{\cro}{\operatorname{Cr}}
\newcommand{\rect}{\Box}
\newcommand{\width}{\operatorname{width}}
\newcommand{\event}{{\cal{E}}}
\newcommand{\patrascu}{P\v{a}tra\c{s}cu\xspace}
\newcommand{\WBone}{WB-1 bound\xspace}

\global\long\def\alt{\mathrm{alt}}
\global\long\def\WBtwo{\mathsf{WB}^{(2)}}
\global\long\def\funnel{\mathrm{funnel}}
\global\long\def\Xhat{\hat{X}}
\global\long\def\Xstar{X^{*}}

\def \ShowComment{true}

\ifdefined\ShowComment

\def\julia#1{\marginpar{$\leftarrow$\fbox{J}}\footnote{$\Rightarrow$~{\sf #1 --Julia}}}
\def\parinya#1{\marginpar{$\leftarrow$\fbox{P}}\footnote{$\Rightarrow$~{\sf #1 --Parinya}}}
\def\thatchaphol#1{\marginpar{$\leftarrow$\fbox{T}}\footnote{$\Rightarrow$~{\sf #1 --Thatchaphol}}}
\def\julia#1{\marginpar{$\leftarrow$\fbox{J}}\footnote{$\Rightarrow$~{\sf #1 --Julia}}}
\newcommand{\mynote}[2][red]{\textcolor{red}{\sc\bf{[#2]}}}

\else 
\def\julia#1{} 
\def\parinya#1{} 
\def\thatchaphol#1{} 
\def\julia#1{} 
\newcommand{\mynote}[2][red]{}

\fi

\title{Pinning Down the Strong Wilber 1 Bound for Binary Search Trees}

\author{Parinya Chalermsook\thanks{Aalto University, Finland. Email: {\tt chalermsook@gmail.com}. Supported by European Research Council (ERC) under the European Union’s Horizon 2020 research and innovation programme (grant agreement No. 759557) and by Academy of Finland Research Fellows, under grant No. 310415 }\and Julia Chuzhoy\thanks{Toyota Technological Institute at Chicago. Email: {\tt cjulia@ttic.edu}. Part of the work was done while the author was a Weston visiting professor at the Department of Computer Science and Applied Mathematics, Weizmann Institute of Science. Supported in part by NSF grant CCF-1616584.}  \and Thatchaphol Saranurak \thanks{Toyota Technological Institute at Chicago. Email: {\tt saranurak@ttic.edu}.} 
}

\begin{document}

\maketitle

\pagenumbering{gobble}

\begin{abstract} 
The dynamic optimality conjecture, postulating the existence  of an $O(1)$-competitive online algorithm for binary search trees (BSTs),  is among the most fundamental open problems in dynamic data structures. Despite extensive work and some notable progress, including, for example, the Tango Trees (Demaine et al., FOCS 2004), that give the best currently known $O(\log \log n)$-competitive algorithm, the conjecture remains widely open.  
One of the main hurdles towards settling the conjecture is that we currently do not have approximation algorithms achieving better than an $O(\log \log n)$-approximation, even in the offline setting. All known non-trivial algorithms for BST's so far rely on comparing the algorithm's cost with the so-called Wilber's first bound (WB-1). Therefore, establishing the worst-case relationship between this bound and the optimal solution cost appears crucial for further progress, and it is an interesting open question in its own right.

Our contribution is two-fold. First, we show that the gap between the WB-1 bound and the optimal solution value can be as large as $\Omega(\log \log n/ \log \log \log n)$; in fact, the gap holds even for several stronger variants of the bound.  
Second, we provide a simple algorithm, that, given an integer $D>0$, obtains an $O(D)$-approximation in time $\exp\left(O\left (n^{1/2^{\Omega(D)}}\log n\right )\right )$. In particular, this gives a constant-factor approximation sub-exponential time algorithm. %
Moreover, we obtain a simpler and cleaner efficient $O(\log \log n)$-approximation algorithm that can be used in an online setting. %
Finally, we suggest a new bound, that we call {\em Guillotine Bound}, that is stronger than WB, while maintaining its algorithm-friendly nature, that we hope will lead to better algorithms.  
All our results use the geometric interpretation of the problem, leading to cleaner and simpler analysis.
\end{abstract}

\newpage
\tableofcontents{}
\newpage
\pagenumbering{arabic}
\newcommand{\temptext}[1]{\textcolor{blue}{#1} }

\section{Introduction}

Binary search trees (BST's) are a fundamental data structure that has been extensively studied for many decades.
Informally, suppose we are given as input an {\bf online} access sequence $X=\{x_{1},\dots,x_{m}\}$ of keys from $\set{1,\ldots,n}$, and our goal is to maintain a binary search tree $T$ over the set $\set{1,\ldots,n}$ of keys. 
The algorithm is allowed to modify the tree $T$ after each access; the tree obtained after the $i$th access is denoted by $T_{i+1}$. 
Each such modification involves a sequence of \emph{rotation} operations that transform the current tree $T_i$ into a new tree $T_{i+1}$. 
The cost of the transformation is the total number of rotations performed plus the depth of the key $x_i$ in the tree $T_i$. 
The total cost of the algorithm is the total cost of all transformations performed as the sequence $X$ is processed. 
We denote by $\opt(X)$ the smallest cost of any algorithm for maintaining a BST for the access sequence $X$, when  the whole sequence $X$ is known to the algorithm in advance.

Several algorithms for BST's, whose costs are guaranteed to be $O(m\log n)$ for any access sequence, such as AVL-trees \cite{avl} and red-black trees \cite{Bayer1972}, are known since the 60's. Moreover, it is well known that there are length-$m$ access sequences $X$ on $n$ keys, for which $\opt(X)=\Omega(m\log n)$.
However, such optimal worst-case  guarantees are often unsatisfactory from both practical and theoretical perspectives, as one can often obtain better results for ``structured'' inputs. 
Arguably, a better notion of the algorithm's performance to consider is {\em instance optimality}, where the algorithm's performance is compared to the optimal cost $\opt(X)$ for the specific input access sequence $X$. This notion is naturally captured by the algorithm's {\em competitive ratio}:  
We say that an algorithm for BST's is \emph{$\alpha$-competitive}, if, for every online input access sequence
$X$, the cost of the algorithm's execution on $X$ is at most $\alpha \cdot \opt(X)$.
Since for every length-$m$ access sequence $X$, $\opt(X)\geq m$, the above-mentioned algorithms that provide worst-case $O(m\log n)$-cost guarantees are also $O(\log n)$-competitive. 
However, there are many known important special cases, in which the value of the optimal solution is $O(m)$, and for which the existence of an $O(1)$-competitive algorithm would lead to a much better performance, including some interesting applications, such as, for example, adaptive sorting
\cite{Tar85,chaudhuri,split_Luc91,deque_Sun92,deque_Elm04,deque_Pet08,DS09,finger1,finger2,BoseDIL14,FOCS15,our_wads}. 
A striking conjecture of Sleator and Tarjan~\cite{ST85} from 1985, called the \emph{dynamic optimality conjecture}, asserts that the \emph{Splay Trees} provide an $O(1)$-competitive algorithm for BST's. 
This conjecture has sparked a long line of research, but despite the continuing effort, and the seeming simplicity of BST's, it remains widely open.
In a breakthrough result, Demaine, Harmon, Iacono and Patrascu \cite{tango} proposed Tango Trees algorithm, that achieves an $O(\log\log n)$-competitive ratio, and has remained the best known algorithm for the problem, for over 15 years.   
A natural avenue for overcoming this barrier is to first consider the ``easier'' task of designing (offline) approximation algorithms, whose approximation factor is below $O(\log \log n)$. Designing better approximation algorithms is often a precursor to obtaining better online algorithms, and it is a natural stepping stone towards this goal. 

The main obstacle towards designing better algorithms, both in the online and the offline settings, is  obtaining tight lower bounds on the value $\opt(X)$, that can be used in algorithm design.  If the input access sequence $X$ has length
$m$, and it contains $n$ keys, then it is easy to see that $\opt(X)\ge m$,  and, by using any balanced BST's, such as AVL-trees, one can show that $\opt(X)=O(m\log n)$. This trivially implies an
$O(\log n)$-approximation for both offline and online settings. However, in order to obtain better approximation, these simple bounds do not seem sufficient. Wilber \cite{wilber} proposed two new bounds, that we refer to as the first Wilber Bound (WB-1) and the second Wilber
Bound (WB-2). He proved that, for every input sequence $X$, the values of both these bounds on $X$ are at most $\opt(X)$. %
The breakthrough result of 
Demaine, Harmon, Iacono and \patrascu \cite{tango}, that gives
an $O(\log\log n)$-competitive online algorithm, relies on WB-1. In particular, they show that the cost of the solution produced by their algorithm is within an $O(\log\log n)$-factor from the
WB-1 bound on the given input sequence $X$, and hence from $\opt(X)$. This in turn implies that, for every input sequence $X$, the value of the WB-1 bound is within an $O(\log\log n)$ factor from $\opt(X)$. Follow-up work \cite{multisplay,chain_splay} improved several
aspects of Tango Trees, but it did not improve the approximation factor. Additional lower bounds on $\opt$, that subsume both the WB-1 and the WB-2 bounds, were suggested in \cite{DHIKP09,DSW}, but unfortunately it is not clear how to exploit them in algorithm design. %
To this day, the only method we have for designing non-trivial online or offline approximation algorithms for BST's is by relying on the WB-1 bound, and this seems to be the most promising approach for obtaining better algorithms.
In order to make further progress on both online and offline approximation algorithms for BST, it therefore appears crucial that we better understand the relationship between the WB-1 bound and the optimal solution cost. 

Informally, the WB-1 bound relies on recursive partitioning of the input key sequence, that  can be represented by a partitioning tree. The standard WB-1 bound (that we refer to as the \emph{weak} WB-1 bound)  only considers a single such partitioning tree. It is well-known (see e.g. \cite{tango,multisplay,in_pursuit}), that the gap between $\opt(X)$
and the weak WB-1 bound for an access sequence $X$ may be as large as  $\Omega(\log\log n)$.  
However, the ``bad'' access sequence $X$ used to obtain this gap is highly dependent on the fixed partitioning tree $T$. It is then natural to consider a stronger variant of WB-1, that we refer to as \emph{strong} WB-1 and denote by $\WB(X)$, that maximizes the weak WB-1 bound over all such partitioning trees.
As suggested by  Iacono \cite{in_pursuit}, and by Kozma \cite{Kozma16_thesis}, this gives a promising approach for improving the $O(\log\log n)$-approximation factor. 

In this paper, we show that, even for this strong variant of Wilber Bound, the gap between $\opt(X)$ and $\WB(X)$ may be as large as $\Omega(\log \log n/ \log \log \log n)$.  This negative result extends to an even stronger variant of the Wilber Bound, where both vertical and horizontal cuts are allowed (in the geometric view of the problem that we describe below). We then propose  an even stronger variant of WB-1, that we call the {\em Guillotine Bound}, to which our negative results do not extend. This new bound seems to maintain the  algorithm-friendly nature of WB-1, and in particular it naturally fits into the algorithmic framework that we propose. We hope that this bound can lead to improved algorithms, both in the offline and the online settings.

Our second set of results is algorithmic.
We show an (offline) algorithm that, given an input sequence $X$ and a positive integer $D$, obtains an $O(D)$-approximation, in time $\poly(m)\cdot \exp\left(n^{1/2^{\Omega(D)}}\log n\right )$. 
When $D$ is constant, the algorithm obtains an $O(1)$-approximation in sub-exponential time. When $D$ is $\Theta(\log \log n)$, it matches the best current efficient $O(\log \log n)$-approximation algorithm. In the latter case, we can also adapt the algorithm to the online setting, obtaining an $O(\log\log n)$-competitive online algorithm.

All our results use the geometric interpretation of the problem, introduced by Demaine et~al.~\cite{DHIKP09}, leading to clean divide-and-conquer-style arguments that avoid, for example, the notion of pointers and rotations.
We feel that this approach, in addition to providing a cleaner and simpler view of the problem, is more natural to work with in the context of approximation algorithms. The area of approximation algorithms offers a wealth of powerful techniques, that appear to be more suitable to the geometric interpretation of the problem. The new Guillotine Bound that we introduce also fits in naturally with our algorithmic techniques, and we hope that this approach will lead to better approximation algorithms, and eventually online algorithms.

We note that several other lower bounds on $\opt(X)$ are known, that subsume the WB-1 bound, including the second Wilber Bound (WB-2)~\cite{wilber}, the Independent Rectangles Bound~\cite{DHIKP09}, and the Cut Bound~\cite{Harmon}; the latter two are known to be equivalent~\cite{Harmon}. Unfortunately, so far it was unclear how to use these bounds in algorithm design.

\paragraph{Independent Work.}
Independently from our work, Lecomte and Weinstein \cite{LW19} showed that Wilber's Funnel bound (that we call WB-2 bound and discuss below) dominates the WB-1 bound, and moreover, they show an access sequence $X$ for which the two bounds have a gap of $\Omega(\log\log n)$. In particular, their result implies that the gap between $\WB(X)$ and $\opt(X)$ is $\Omega(\log\log n)$ for that access sequence. We note that the access sequence $X$ that is used in our negative results also provides a gap of $\Omega(\log\log n/\log\log\log n)$ between the WB-2 and the WB-1 bounds, although we only realized this after hearing the statement of the results of \cite{LW19}. Additionally, Lecomte and Weinstein show that the WB-2 bound is invariant under rotations, and use this to show that, when the WB-2 bound is constant, then the Independent Rectangle bound of \cite{DHIKP09} is linear, thus making progress on establishing the relationship between the two bounds.

We now provide a more detailed description of our results.

\subsection*{Our Results and Techniques}

We use the geometric interpretation of the
problem, introduced by Demaine et~al.~\cite{DHIKP09}, that we refer to as the \minsat problem. 
Let $P$ be any set of points in the plane. We say that two points $p,q\in P$ are \emph{collinear} iff either their $x$-coordinates are equal, or their $y$-coordinates are equal. If $p$ and $q$ are non-collinear, then we let $\rect_{p,q}$ be the smallest closed rectangle containing both $p$ and $q$; notice that $p$ and $q$ must be diagonally opposite corners of this rectangle. We say that the pair $(p,q)$ of points is \emph{satisfied} in $P$ iff there is some additional point $r\neq p,q$ in $P$ that lies in $\rect_{p,q}$.
Lastly, we say that the set $P$ of points is satisfied iff for every pair $p,q\in P$ of distinct points, either $p$ and $q$ are collinear, or they are satisfied in $P$.

In the \minsat problem, the input is a set $P$ of points in the plane with integral $x$- and $y$-coordinates; we assume that all $x$-coordinates are between $1$ and $n$, and all $y$-coordinates are between $1$ and $m$ and distinct from each other, and that $|P|=m$. The goal is to find a minimum-cardinality set $Y$ of points, such that the set $P\cup Y$ of points is satisfied.

An access sequence $X$ over keys $\set{1,\ldots,n}$ can be represented by a set $P$ of points in the plane as follows:
if a key $x$ is accessed at time $y$, then add the point $(x,y)$ to $P$.
Demaine et~al.~\cite{DHIKP09} showed that, for every access sequence $X$, if we denote by $P$ the corresponding set of points in the plane, then the value of the optimal solution to the \minsat problem on $P$ is $\Theta(\opt(X))$. Therefore, in order to obtain an $O(\alpha)$-approximation algorithm for BST's, it is sufficient to obtain an $\alpha$-approximation algorithm for the \minsat problem. In the online version of the \minsat problem, at every time step $t$, we discover the unique input point whose $y$-coordinate is $t$, and we need to decide which points with $y$-coordinate $t$ to add to the solution. Demaine et~al.~\cite{DHIKP09} also showed that an $\alpha$-competitive online algorithm for \minsat implies an $O(\alpha)$-competitive online algorithm for BST's.
For convenience, we do not distinguish between the input access sequence $X$ and the corresponding set of points in the plane, that we also denote by $X$.

\subsection{Negative Results for WB-1 and Its Extensions}

We say that an input access sequence $X$ is a \emph{permutation} if each key in $\set{1,\ldots,n}$ is accessed exactly
once. Equivalently, in the geometric view, every column with an integral $x$-coordinate contains exactly one input point.

Informally, the WB-1 bound for an input sequence $X$ is defined as follows. Let $B$ be the bounding box containing all points of $X$, and consider any vertical line $L$ drawn across $B$, that partitions it into two vertical strips, separating the points of $X$ into two subsets $X_{1}$
and $X_{2}$. Assume that the points of $X$ are ordered by their $y$-coordinates from smallest to largest. We say that a pair $(x,x')\in X$ of points \emph{cross} the line $L$, iff $x$ and $x'$ are consecutive points of $X$, and they lie on different sides of $L$. Let $C(L)$ be the number of all pairs of points in $X$ that cross $L$. We then continue this process recursively with $X_1$ and $X_2$, with the final value of the WB-1 bound being the sum of the two resulting bounds obtained for $X_1$ and $X_2$, and $C(L)$.
This recursive partitioning process can be represented by a binary tree $T$ that
we call a \emph{partitioning tree} (we note that the partitioning tree is not related to the BST tree that the BST algorithm maintains). 
Every vertex $v$ of the partitioning tree is associated with a vertical strip $S(v)$, where for the root vertex $r$, $S(r)=B$. If the partitioning algorithm uses a vertical line $L$ to partition the strip $S(v)$ into two sub-strips $S_1$ and $S_2$, then vertex $v$ has two children, whose corresponding strips are $S_1$ and $S_2$. Note that every sequence of vertical lines used in the recursive partitioning procedure corresponds to a unique partitioning tree and vice versa. 
Given a set  $X$ of points and a partitioning tree $T$, we denote  by $\WB_{T}(X)$ the WB-1 bound obtained for $X$ while following the partitioning scheme defined by $T$. Wilber \cite{wilber}
showed that, for every partitioning tree $T$, $\opt(X)\geq \Omega(\WB_{T}(X))$ holds. Moreover,  Demaine et al. \cite{tango} showed that, if $T$ is a balanced tree, then $\opt(X)\le O(\log\log n)\cdot \WB_{T}(X)$. These two bounds are used to obtain the $O(\log\log n)$-competitive algorithm of \cite{tango}. We call this variant of WB-1, that is defined with respect to a fixed tree $T$, the \emph{weak} WB-1 bound.

Unfortunately, it is well-known (see e.g. \cite{tango,multisplay,in_pursuit}), that the gap between $\opt(X)$
and the weak WB-1 bound on an input $X$ may be as large as  $\Omega(\log\log n)$. In other words, for any fixed partitioning tree $T$, there exists an input $X$ (that depends on $T$),  with $\WB_{T}(X)\leq O(\opt(X)/\log\log n)$ (see Appendix \ref{sec: Iacono} for details). 
However, the construction of this ``bad'' input $X$  depends on the fixed partitioning tree $T$. We consider a stronger variant of WB-1, that we refer to as \emph{strong} WB-1 bound and denote by $\WB(X)$, that maximizes the weak WB-1 bound over all such partitioning trees, that is, $\WB(X)=\max_{T}\{\WB_{T}(X)\}$. Using this stronger bound as an alternative to weak WB-1 in order to obtain better approximation algorithms was suggested by  Iacono \cite{in_pursuit}, and by Kozma \cite{Kozma16_thesis}.

Our first result rules out this approach: we show that, even for the strong WB-1 bound, the gap between $\WB(X)$ and $\opt(X)$ may be as large as $\Omega(\log\log n/\log\log\log n)$, even if the input $X$ is a permutation. %

\begin{theorem}
\label{thm:intro_WB}
For every integer $n'$, there is an integer $n\geq n'$, and an access sequence $X$ on $n$ keys with $|X|=n$, such that $X$ is a permutation, $\opt(X)\geq \Omega(n\log\log n)$, but $\WB(X)\leq O(n\log\log\log n)$. In other words, for every partitioning tree $T$, $\frac{\opt(X)}{\WB_{T}(X)}\geq \Omega\left(\frac{\log\log n}{\log\log\log n}\right )$.
\end{theorem}

We note that it is well known  (see e.g. \cite{FOCS15}), that any $c$-approximation algorithm for permutation input can be turned into an $O(c)$-approximation algorithm for any input sequence. However, the known instances that achieve an $\Omega(\log\log n)$-gap between the weak WB-1 bound and $\opt$ are not permutations (see \Cref{sec: Iacono}). Therefore, our result is the first one providing a super-constant gap between the Wilber Bound and $\opt$ for permutations, even for the case of weak WB-1.

\paragraph{Further Extensions of the WB-1 Bound.}
We then turn to consider several generalizations of the WB-1 bound: namely, the \emph{consistent Guillotine Bound}, where the partitioning lines $L$ are allowed to be vertical or horizontal, and an even stronger \emph{Guillotine Bound}. While our negative results extend to the former bound, they do not hold for the latter bound. Further, the general algorithmic approach of Theorem \ref{thm:intro_alg} can be naturally adapted to work with the Guillotine Bound. We now discuss both bounds in more detail.

The Guillotine bound
$\GB(X)$ extends $\WB(X)$ by allowing both vertical and horizontal partitioning lines. Specifically, given the bounding box $B$, we let $L$ be any vertical or horizontal line crossing $B$, that separates  $X$ into two
subsets $X_{1}$ and $X_{2}$. We define the number of crossings of $L$ exactly as before, and then recurse on both sides of $L$ as before. This partitioning scheme can be represented by a binary tree
$T$, where every vertex of the tree is associated with a rectangular region of the plane. We denote the resulting bound obtained by using the partitioning tree $T$ by $\GB_{T}(X)$, and we define
$\GB(X)=\max_{T}\GB_{T}(X)$. 

The Consistent Guillotine bound restricts
the Guillotine bound by maximizing only over partitioning schemes that are
``consistent'' in the following sense: suppose that the current partition of the bounding box $B$, that we have obtained using previous partitioning lines, is  a collection $\{R_{1},\dots,R_{k}\}$ of rectangular regions. Once we choose a vertical or a horizontal line $L$, then for
{\bf every} rectangular region $R_i$ that intersects $L$, we must partition $R_i$ into two sub-regions using the line $L$, and then count the number of consecutive pairs of points in $X\cap R_i$ that cross the line $L$. In other words, we must partition all rectangles $R_{1},\dots,R_{k}$
consistently with respect to the line $L$. In contrast, in the Guillotine bound, we
are allowed to partition each area $R_{i}$ independently. From the definitions,
the value of the Guillotine bound is always at least as large as the value of the Consistent Guillotine
bound on any input sequence $X$, which is at least as large as $\WB(X)$.
We generalize our negative result to show a gap between $\opt$ and the
Consistent Guillotine bound:
\begin{theorem}
	For every integer $n'$, there is an integer $n\geq n'$, and an access sequence $X$ on $n$ keys with $|X|=n$, such that $X$ is a permutation, $\opt(X)=\Omega(n\log\log n)$,
	but $\XWB(X)=O(n\log\log\log n)$.
\end{theorem}

We note that our negative results do not extend to the general $\GB$. 
We leave open an interesting question of establishing the worst-case gap between the value of
$\opt$ and that of the Guillotine bound, and we hope that combining the Guillotine bound with our algorithmic approach that we discuss below will eventually lead to better online and offline approximation algorithms.

\paragraph{Separating the Two Wilber Bounds.} We note that the sequence $X$ given by Theorem~\ref{thm:intro_WB} not only provides a separation between WB-1 and $\opt$, but it also provides a separation between the WB-1 bound and the WB-2 bound (also called the \emph{funnel} bound). The latter can be defined in the geometric view as follows (see Section \ref{subsec: separating two wilbers} for a formal definition). Recall that, for a pair of points $x,y\in X$, $\rect_{x,y}$ is the smallest closed rectangle containing both $x$ nd $y$. For a point $x$ in the access sequence $X$, the \emph{funnel} of $x$ is the set of all points $y\in X$, for which $\rect_{x,y}$ does not contain any point of $X\setminus\set{x,y}$, and $\operatorname{alt}(x)$ is the number of alterations between the left of $x$ and the right of $x$ in the {funnel} of $x$. The second Wilber Bound for sequence $X$ is then defined as: $\WB^{(2)}(X)=|X|+\sum_{x\in X}\operatorname{alt}(x)$. 
We show that, for the sequence $X$ given by Theorem~\ref{thm:intro_WB}, $\WB^{(2)}(X) \geq \Omega(n \log \log n)$ holds, and therefore $\WB^{(2)}(X)/\WB(X) \geq \Omega(\log \log n/ \log \log \log n)$ for that sequence, implying that  the gap between $\WB(X)$ and $\WB^{(2)}(X)$ may be as large as $\Omega(\log \log n/ \log \log \log n)$. 
We note that we only realized that our results provide this stronger separation between the two Wilber bounds after hearing the statements of the results from the independent work of Lecomte and Weinstein~\cite{LW19}  mentioned above.

\subsection{Algorithmic Results}

We provide new simple approximation algorithms for the problem, that rely on its geometric interpretation, namely the \minsat problem.

\begin{theorem}
\label{thm:intro_alg} There  is
an offline algorithm for \minsat, that, given any integral parameter $D\ge1$, and an access sequence $X$ to $n$ keys of length
$m$, has cost at most $O(D\cdot\opt(X))$ and has running
time at most $\poly(m) \cdot \exp\left(O\left (n^{1/2^{O(D)}}\log n\right )\right )$.  
When
$D=O(\log\log n)$, the algorithm's running time is polynomial in $n$ and $m$, and it can be adapted to the online setting, achieving an $O(\log\log n)$-competitive ratio.
\end{theorem}

Our results show that the problem of obtaining a constant-factor approximation for \minsat cannot be NP-hard, unless $\mathsf{NP}\subseteq\mathsf{SUBEXP}$,
where $\mathsf{SUBEXP}=\bigcap_{\epsilon>0}\text{DTime}[2^{n^{\epsilon}}]$. 
This, in turn provides a positive evidence towards the
dynamic optimality conjecture, as one natural avenue to disproving it is to show that obtaining a constant-factor approximation for BST's is NP-hard. Our results rule out this possibility, unless $\mathsf{NP}\subseteq\mathsf{SUBEXP}$, and in particular the Exponential Time Hypothesis is false.
While the $O(\log\log n)$-approximation factor achieved by our algorithm in time $\poly(mn)$ is similar to that achieved by other known algorithms  \cite{tango,chain_splay,multisplay}, this is the first algorithm that relies solely on the geometric formulation of the problem, which is arguably cleaner, simpler, and better suited for exploiting the rich toolkit of algorithmic techniques developed in the areas of online and approximation algorithms. Our algorithmic framework can be naturally adapted to work with both vertical and horizontal cuts, and so a natural direction for improving our algorithms is to try to combine them with the $\GB$ bound.

\subsection{Technical Overview}
As already mentioned, both our positive and negative results use the geometric interpretation of the problem, namely the \minsat problem. This formulation allows one to ignore all the gory details of pointer moving and tree rotations, and instead to focus on a clean combinatorial problem that is easy to state and analyze. Given an input set of points $X$, we denote by $\opt(X)$ the value of the optimal solution to the \minsat problem on $X$.

\paragraph{Nevative results.}
For our main negative result -- the proof of Theorem \ref{thm:intro_WB}, we start with the standard bit-reversal sequence $X$ on $n$ keys of length $n$. It is well known that $\opt(X)=\Theta(n\log n)$ for such an instance $X$. Again, we view $X$ as a set of points in the plane. Next, we transform this set of points by introducing \emph{exponential spacing}: if we denote the original set of points by $X=(p_1,\ldots,p_n)$, where for all $1\leq i\leq n$, the $x$-coordinate of $p_i$ is $i$, then we map each such point $p_i$ to a new point $p'_i$, whose $y$-coordinate remains the same, and $x$-coordinate becomes $2^i$. Let $\tilde X$ denote the resulting set of points. It is well known that, if $T$ is a {\bf balanced} partitioning tree (that is, every vertical strip is always partitioned by a vertical line crossing it in the middle), then $\WB_T(\tilde X)\leq n$, that is, the gap between the WB-1 bound defined with respect to the tree $T$ and $\opt(X)$ is at least $\Omega(\log n)$. However, it is easy to see that there is another tree $T'$, for which $\WB_{T'}(\tilde X)=\Theta(n\log n)=\Theta(\opt(\tilde X))$. It would be convenient for us to denote by $N=2^n$ the number of columns with integral $x$-coordinates for instance $\tilde X$.

Our next step is to define {\bf circularly shifted} instances: for each $0\leq s\leq N-1$, we obtain an instance $X^s$ from $\tilde X$ by circularly shifting it to the right by $s$ units. In other words, we move the last $s$ columns (with integral $x$-coordinates) of $\tilde X$ to the beginning of the instance.

Our final instance is obtained by stacking the instances $X^0,X^1,\ldots,X^{N-1}$ on top of each other in this order. Intuitively, for each individual instance $X^s$, it is easy to find a partitioning tree $T^s$, such that $\WB_{T^s}(X^s)$ is close to $\opt(X^s)$. However, the trees $T^s$ for different values of $s$ look differently. In particular, we show that no single partitioning tree $T$ works well for all instances $X^s$ simultaneously. This observation is key for showing the $\Omega(\log\log n/\log\log\log n)$ gap between the strong $\WB$ bound of the resulting instance and the value of the optimal solution for it.

In order to extend this result to the $\cGB$ bound, we perform the exponential spacing both vertically and horizontally. For every pair $1\leq s,s'\leq N-1$ of integers, we then define a new instance $X^{s,s'}$, obtained by circularly shifting the exponentially spaced instance horizontally by $s$ units and vertically by $s'$ units. The final instance is obtained by combining all resulting instances $X^{s,s'}$ for all $1\leq s,s'\leq N-1$. We then turn this instance into a permutation using a straightforward transformation. 

\paragraph{Geometric decomposition of instances.}
We employ geometric decompositions of instances in both our negative results and in our algorithms. Assume that we are given a set $X$ of input points, with $|X|=n$, such that each point in $X$ has integral $x$- and $y$-coordinates between $1$ and $n$, and all points in $X$ have distinct $y$-coordinates. Let $B$ be a bounding box containing all points. Consider now a collection $\lset=\set{L_1,\ldots,L_{k-1}}$ of vertical lines, that partition the bounding box into vertical strips $S_1,\ldots,S_k$, for some integer $k>1$. We assume that all lines in $\lset$ have half-integral $x$-coordinates, so no point of $X$ lies on any such line. This partition of the bounding box naturally defines $k$ new instances $X_1,\ldots,X_k$, that we refer to as \emph{strip instances}, where instance $X_i$ contains all input points of $X$ that lie in strip $S_i$. Additionally, we define a \emph{contracted instance} $\tilde X$, obtained as follows. For each $1\leq i\leq k$, we collapse all columns lying in the strip $S_i$ into a single column, where the points of $X_i$ are now placed.

While it is easy to see that $\sum_{i=1}^k \opt(X_i) \leq \opt(X)$, by observing that a feasible solution to instance $X$ naturally defines feasible solutions to instances $X_1,\ldots,X_k$, we prove a somewhat stronger result, that $\opt(\tilde X)+\sum_{i=1}^k \opt(X_i) \leq \opt(X)$. This result plays a key role in our algorithm, that follows a divide-and-conquer approach.

We then relate the Wilber Bounds of the resulting instances, by showing that:

	\[\WB(X)\leq O\left (\WB(\tilde X)+\sum_{i=1}^k\WB(X_i)+|X|\right ).   \]

The latter result is perhaps somewhat surprising. One can think of the expression $\WB(\tilde X)+\sum_{i=1}^k\WB(X_i)$ as a Wilber bound obtained by first partitioning along the vertical lines in $\lset$, and then continuing the partition within each strip. However, $\WB(X)$ is defined to be the largest bound over all partitioning schemes, including those where we start by cutting inside the strip instances, and only cut along the lines in $\lset$ later. This result is a convenient tool for analyzing our lower-bound constructions.

\paragraph{Algorithms.}
Our algorithmic framework is based on a simple divide-and-conquer technique. 
Let $X$ be an input point set, such that $|X|=n$, and for every pair of points in $X$, their $x$- and $y$-coordinates are distinct; it is well known that it is sufficient to solve this special case in order to obtain an algorithm for the whole problem. Let $B$ be a bounding box containing all points of $X$. By an active column we mean a vertical line that contains at least one point of $X$. Suppose we choose an integer $k>1$, and consider a set $\lset$ of $k-1$ vertical lines, such that each of the resulting vertical strips $S_1,\ldots,S_k$ contains roughly the same number of active columns. The idea of the algorithm is then to solve each of the resulting strip instances $X_1,\ldots,X_k$, and the contracted instance $\tilde X$ recursively. From the decomposition result stated above, we are guaranteed that  $\opt(\tilde X)+\sum_{i=1}^k \opt(X_i) \leq \opt(X)$. For each $1\leq i\leq k$, let $Y_i$ be the resulting solution to instance $X_i$, and let $\tilde Y$ be the resulting solution to instance $\tilde X$. Unfortunately, it is not hard to see that, if we let $Y'=\tilde Y\cup \left(\bigcup_{i=1}^kY_i\right )$, then $Y'$ is not necessarily a feasible solution to instance $X$. However, we show that there is a collection $Z$ of $O(|X|)$ points, such that for {any} feasible solutions $Y_1,\ldots,Y_k,\tilde Y$ to instances $X_1,\ldots,X_k,\tilde X$ respectively, the set of points $Z\cup \tilde Y\cup \left(\bigcup_{i=1}^kY\right )$ is a feasible solution to instance $X$.
We now immediately obtain a simple recursive algorithm for the \minsat problem. We show that the approximation factor achieved by this algorithm is proportional to the number of the recursive levels, which is optimized if we choose the parameter $k$ -- the number of strips in the partition -- to be $\sqrt{n}$. A simple calculation shows that the algorithm achieves an approximation factor of $O(\log\log n)$, and it is easy to see that the algorithm is efficient.

Finally, to obtain a tradeoff between the running time and the approximation factor, notice that, if we terminate the recursion at recursive depth $D$ and solve each resulting problem optimally using an exponential-time algorithm, we obtain an $O(D)$-approximation. 
Each subproblem at recursion depth $D$ contains $n^{1/2^D}$ active columns. 
We present an algorithm whose running time is only exponential in the number of active columns to obtain the desired result: an exact algorithm, that, for any instance $X$, runs in time $c(X)^{O(c(X))} \poly (|X|)$, where $c(X)$ is the number of active columns.
By using this algorithm to solve the ``base cases'' of the recursion, we obtain an $O(D)$-approximation algorithm with running time $\poly(m)\cdot \exp\left(n^{1/2^{\Omega(D)}}\log n\right )$.

\subsection{Organization}
We start with preliminaries in Section \ref{sec: prelims}, and we provide our results for the geometric decompositions into strip instances and a compressed instance in Section \ref{sec: decomp}. We prove our negative results in Sections \ref{sec: negative} and \ref{sec: extension to extended}; Section \ref{sec: negative} contains the proof of Theorem \ref{thm:intro_WB}, while Section \ref{sec: extension to extended} discusses extensions of the WB-1 bound, and extends the result of Theorem  \ref{thm:intro_WB} to the $\cGB$ bound. Our algorithmic results appear in Section \ref{sec: alg}.

\section{Preliminaries}\label{sec: prelims}

All our results only use the geometric interpretation of the problem, that we refer to as the \minsat problem, and define below. For completeness, we include the formal definition of algorithms for BST's and formally state their equivalence to \minsat in \Cref{sec: equivalence}.

	\subsection{The \minsat Problem}
	\label{sec: problem def}

	For a point $p \in {\mathbb R}^2$ in the plane, we denote by $p.x$ and $p.y$ its $x$- and $y$-coordinates, respectively. 
	Given any pair $p$, $p'$ of points, we say that they are \emph{collinear} if $p.x=p'.x$ or $p.y=p'.y$. If $p$ and $p'$ are not collinear, then we let $\rect_{p,p'}$ be the smallest closed rectangle containing both $p$ and $p'$; note that $p$ and $p'$ must be diagonally opposite corners of the rectangle.

	\begin{definition}
	We say that a non-collinear pair $p,p'$ of points is \emph{satisfied by a point $p''$} if $p''$ is distinct from $p$ and $p'$ and $p''\in \rect_{p,p'}$.  
	We say that a set $S$ of points is \emph{satisfied} iff for every non-collinear pair $p,p'\in S$ of points, there is some point $p''\in S$ that satisfies this pair.
	\end{definition}

	We refer to horizontal and vertical lines as \emph{rows} and \emph{columns} respectively. 
	For a collection of points $X$, the \textit{active rows} of $X$ are the rows that contain at least one point in $X$. We define the notion of \textit{active columns} analogously. 
	We denote by $r(X)$ and $c(X)$ the number of active rows and active columns of the point set $X$, respectively.
	We say that a point set $X$ is a {\em semi-permutation} if every active row contains exactly one point of $X$.
	Note that, if $X$ is a semi-permutation, then $c(X) \le r(X)$. We say that $X$ is a {\em permutation} if it is a semi-permutation, and additionally, every active column contains exactly one point of $X$.  
	Clearly,  if $X$ is a permutation, then $c(X)=r(X)=|X|$. 
	We denote by  $B$ the smallest closed rectangle containing all points of $X$, and call $B$ the \emph{bounding box}.

	We are now ready to define the  \minsat problem. 
The input to the problem is a set $X$ of points that is a semi-permutation, and the goal is to compute a minimum-cardinality set $Y$ of points, such that $X\cup Y$ is satisfied.  
	We say that a set $Y$ of points is a \emph{feasible solution} for $X$ if $X\cup Y$ is satisfied. We denote by $\opt(X)$ the minimum value $|Y|$ of any feasible solution $Y$ for $X$.\footnote{We remark that in the original paper that introduced this problem~\cite{DHIKP09}, the value of the solution is defined as $|X \cup Y|$, while our solution value is $|Y|$. It is easy to see that
		for any semi-permutation $X$ and solution $Y$ for $X$, $|Y|\geq \Omega(|X|)$ must hold, so the two definitions are equivalent to within factor $2$.}
In the online version of the \minsat problem, at every time step $t$, we discover the unique input point from $X$ whose $y$-coordinate is $t$, and we need to decide which points with $y$-coordinate $t$ to add to the solution $Y$.
The \minsat problem is equivalent to the BST problem, in the following sense (see \Cref{sec: equivalence} for more details): 

\begin{theorem}[\cite{DHIKP09}] 
Any efficient $\alpha$-approximation algorithm for \minsat can be transformed into an efficient $O(\alpha)$-approximation algorithm for BST's, and similarly any online $\alpha$-competitive algorithm for \minsat can be transformed into an online $O(\alpha)$-competitive algorithm for BST's. 
\end{theorem}

	\subsection{Basic Geometric Properties }
	The following observation is well known (see, e.g.~Observation 2.1 from \cite{DHIKP09}). We include the proof Section \ref{subsec: proof of Obs: alligned point} of the Appendix for completeness. %
\begin{observation}
        \label{obs: aligned point} 
        Let $Z$ be  any satisfied point set. Then for every pair $p,q \in Z$ of distinct points, there is a point $r \in \Box_{p,q} \setminus \{p,q\}$ such that $r.x = p.x$ or $r.y = p.y$. 
\end{observation}

\paragraph{Collapsing Sets of Columns or Rows.}
Assume that we are given any set $X$ of points, and any collection $\cset$ of consecutive active columns for $X$. 
In order to collapse the set $\cset$ of columns, we replace $\cset$ with a single representative column $C$ (for concreteness, we use the column of $\cset$ with minimum $x$-coordinate). 
For every point $p\in X$ that lies on a column of $\cset$, we replace $p$ with a new point, lying on the column $C$, whose $y$-coordinate remains the same.
Formally, we replace point $p$ with point $(x, p.y)$, where $x$ is the $x$-coordinate of the column $C$.  
We denote by $X_{|\cset}$ the resulting new set of points. We can similarly define collapsing set of rows. The following useful observation is easy to verify; the proof appears in Section \ref{subsec: proof of obs: collapsing col} of Appendix. 

\begin{observation}\label{obs: collapsing columns}
Let $S$ be any set of points, and let $\cset$ be any collection of consecutive active columns (or rows) with respect to $S$. If $S$ is a satisfied set of points, then so is $S_{|\cset}$.
\end{observation}

\paragraph{Canonical Solutions.} We say that a solution $Y$ for input $X$ is \emph{canonical} iff every point $p\in Y$ lies on an active row and an active column of $X$. %
It is easy to see that \emph{any} solution can be transformed into a canonical solution, without increasing its cost. The proof of the following observation appears in Section \ref{subsec: proof of obs: canonical solutions} of Appendix. %

\begin{observation} \label{obs: canonical solutions} There is an efficient algorithm, that, given an instance $X$ of \minsat and
	 any feasible solution $Y$ for $X$, computes a feasible canonical solution $\hat{Y}$ for $X$ with $|\hat{Y}| \leq |Y|$.  %
\end{observation}

\subsection{Partitioning Trees}

We now turn to define partitioning trees, that are central to both defining the \WBone and to describing our algorithm.

Let $X$ be the a  set of points that is a semi-permutation.
We can assume without loss of generality that every column with an integral $x$-coordinate between $1$ and $c(X)$ inclusive contains at least one point of $X$.
Let $B$
be the bounding box of $X$. Assume that the set of active columns
is $\{C_1,\dots,C_a\}$, where $a=c(X)$, and that for all $1\leq i\leq a$, the $x$-coordinate of column $C_i$ is $i$.
Let $\lset$ be the set of all vertical lines with half-integral $x$-coordinates
between $1+1/2$ and $a-1/2$ (inclusive). Throughout, we refer to the vertical lines in $\lset$ as \emph{auxiliary columns}. Let $\sigma$ be an arbitrary ordering
of the lines of $\lset$ and denote $\sigma=(L_{1},L_{2},\dots,L_{a-1})$.
We define a hierarchical partition of the bounding box $B$ into vertical strips using $\sigma$, as follows.
We perform $a-1$ iterations. In the first iteration, we partition the bounding box $B$, using the line $L_1$, into two vertical strips, $S_L$ and $S_B$. For $1<i\leq a-1$, in iteration $i$ we consider the line $L_i$, and we let $S$ be the unique vertical strip in the current partition that contains the line $L_i$. We then partition $S$ into two vertical sub-strips by the line $L_i$.
When the partitioning algorithm terminates, every vertical strip contains exactly one active column.

\begin{figure}
	\centering
	\includegraphics[width=0.3\textwidth]{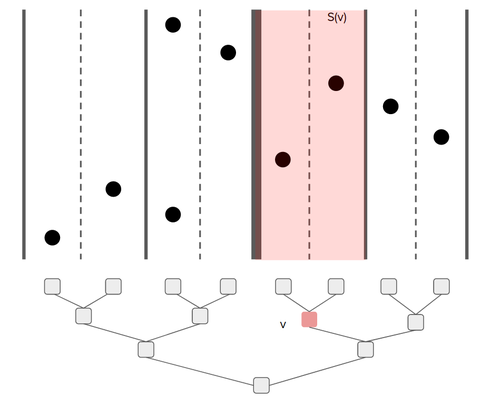}
	\caption{Input point set $X$ and a partitioning tree $T$.}
	\label{fig:Partition}
\end{figure}

This partitioning process can be naturally described by a binary tree $T=T(\sigma)$, that
 we call a \emph{partitioning tree} associated with the ordering
$\sigma$ (see Figure~\ref{fig:Partition}). 
Each node $v\in V(T)$ is associated with a vertical strip
$S(v)$ of the bounding box $B$. The strip $S(r)$ of the root vertex $r$ of $T$ is
the bounding box $B$. For every inner vertex $v\in V(T)$, if $S=S(v)$ is the vertical strip associated with $v$, and if $L\in \lset$ is the first line in $\sigma$ that lies strictly in $S$, then line $L$ partitions $S$ into two sub-strips, that we denote by $S_L$ and $S_R$. Vertex $v$ then has two children, whose corresponding strips are $S_L$ and $S_R$ respectively. We say that  $v$ \emph{owns }the line $L$, and we denote $L=L(v)$. For each leaf
node $v$, the corresponding strip $S(v)$ contains exactly one active column of $X$, and $v$ does not
own any line of $\lset$.
For each vertex $v \in V(T)$, let $\size(v)=|X \cap S(v)|$ be the number
of points from $X$ that lie in $S(v)$, and let $\width(v)$ be the width
of the strip $S(v)$.

Given a partition tree $T$ for point set $X$, we refer to the vertical 
strips in $\{S(v)\}_{v\in T}$ as $T$-strips.

We can use the partitioning trees in order to show the following well known bound on $\opt(X)$\footnote{The algorithm in the claim in fact corresponds to searching in a balanced static BST.}:
\begin{claim}\label{claim: upper bound on OPT static}
	For any semi-permutation $X$, 
	$\opt(X)\leq O(r(X) \log c(X))$, where $r(X)$ and $c(X)$ are the number of active rows and columns of $X$ respectively.
\end{claim}

\subsection{The WB-1 Bound}\label{sec: WB def}

The \WBone\footnote{Also called Interleaving bound \cite{tango}, the first Wilber bound, ``interleave lower bound''~\cite{wilber}, or alternation bound~\cite{in_pursuit}} is defined with respect to an ordering (or a permutation) $\sigma$ of the auxiliary columns, or,  equivalently, with respect to the partitioning tree $T(\sigma)$. 
It will be helpful to keep both these views in mind. 
In this paper, we will make a clear distinction between a weak variant of the \WBone, as defined by Wilber himself in \cite{wilber} and a strong variant, as mentioned in \cite{in_pursuit}.  

Let $X$ be a semi-permutation, and let $\lset$ be the corresponding set of auxiliary columns. 
Consider an arbitrary fixed ordering $\sigma$ of columns in $\lset$ and its corresponding partition tree $T = T(\sigma)$.   
For each inner node $v \in V(T)$, consider the set $X' = X \cap S(v)$ of input points that lie in the strip $S(v)$, and let $L(v)\in \lset$ be the line that $v$ owns.
We denote $X'=\set{p_1,p_2,\ldots,p_k}$, where the points are ordered in the increasing order of their $y$-coordinates; since $X$ is a semi-permutation, no two points of $X$ may have the same $y$-coordinate. 
For $1\leq j<k$, we say that the ordered pair $(p_j,p_{j+1})$ of points form a \emph{crossing} of $L(v)$ iff $p_j,p_{j+1}$ lie on the opposite sides of the line $L(v)$. 
We let $\cost(v)$ be the total number of crossings of $L(v)$ by the points of $X \cap S(v)$. When $L = L(v)$, we also write $\cost(L)$ to denote $\cost(v)$. If $v$ is a leaf vertex, then its cost is set to $0$.

We note a simple observation, that the cost can be bounded by the number of points on the smaller side.
 
\begin{observation}\label{obs: number of crossings at most min of both sides}
	Let $X$ be a semi-permutation, $\sigma$ an ordering of the auxiliary columns in $\lset$, and let $T=T_{\sigma}$ be the corresponding partitioning tree. Let $v\in V(T)$ be any inner vertex of the tree, whose two child vertices are denoted by $v_1$ and $v_2$. 
Then $\cost(v) \leq 2 \min \{|X \cap S(v_1)|, |X \cap S(v_2)| \}$. 
\end{observation}

\begin{definition}[\WBone] 
For any semi-permutation $X$, an ordering $\sigma$ of the auxiliary columns in $\lset$, and the corresponding partitioning tree $T=T_{\sigma}$,
the (weak) \WBone of $X$ with respect to $\sigma$ is: 
\[\WB_{\sigma}(X) =  \WB_{T}(X) = \sum_{v \in V(T)} \cost(v). \]  
The strong \WBone of $X$ is $\WB(X) = \max_{\sigma} \WB_{\sigma}(X)$, where the maximum is taken over all permutations $\sigma$ of the lines in $\lset$.  
\end{definition}

It is well known that the \WBone is a lower bound on the optimal solution cost:

\begin{claim}\label{claim: bounding WB by OPT}
	For any semi-permutation $X$, $\WB(X)\leq 2 \cdot \opt(X)$.
\end{claim} 

The original proof of this fact is due to Wilber~\cite{wilber}, which was later presented in the geometric view by Demaine et al.~\cite{DHIKP09}, via the notion of {\em independent rectangles}. 
In \Cref{sec: WB proof}, we give a simple self-contained proof of this claim.

Combining Claims~\ref{claim: upper bound on OPT static} and \ref{claim: bounding WB by OPT}, we obtain the following immediate corollary.

\begin{corollary}\label{cor: upper bound on WB}
	For any semi-permutation $X$, 
	$\WB(X)\leq O(r(X) \log c(X))$.
\end{corollary}

\section{Geometric Decomposition Theorems} 
\label{sec: decomp}

In this section, we develop several technical tools that will allow us to decompose a given instance into a number of sub-instances. We then analyze the optimal solution costs and the Wilber bound values for the resulting subinstances.
We start by introducing a \emph{split operation} of a given instance into subinstances in \Cref{sec:split instance}. 
In \Cref{sec:decompose opt}, we present a decomposition property of the optimal solution cost, which will be used in our algorithms.  
Finally, in Section \ref{thm: wb}, we discuss the decomposition of the strong \WBone, that our negative results rely on.

\subsection{Split Instances}
\label{sec:split instance} 

Consider a semi-permutation $X$ and its partitioning tree $T$. Let $U\subseteq V(T)$ be a collection of vertices of the tree $T$,  such that the strips $\{S(v)\}_{v \in U}$ partition the bounding box. In other words, every root-to-leaf path in $T$ must contain exactly one vertex of $U$. For instance, $U$ may contain all vertices whose distance from the root of the tree $T$ is the same. We now define splitting an instance $X$ via the set $U$ of vertices of $T$.

\begin{definition}[A Split]
	A \emph{split} of $(X,T)$ at $U$ is a collection of instances $\{X^c, \{X^s_v\}_{v \in U} \}$, defined as follows.
	\begin{itemize}
		\item For each vertex $v \in U$, instance $X^s_v$ is called a {\bf strip instance}, and it contains all points of $X$ that lie in the interior of the strip $S(v)$.

		\item Instance $X^c$ is called a {\bf compressed instance}, and it is obtained from $X$ by collapsing, for every vertex $v \in U$, all active columns in the strip $S(v)$ into a single column.  
	\end{itemize}
\end{definition}

We also partition the tree $T$ into sub-trees that correspond to the new instances: for every vertex $v \in U$,  we let $T_v$ be the sub-tree of $T$ rooted at $v$. Observe that $T_v$ is a partitioning tree for instance $X^s_v$. 
The tree $T^c$ is obtained from $T$ by deleting from it, for all $v\in U$, all vertices of $V(T_v)\setminus \set{v}$. It is easy to verify that $T^c$ is a valid partitioning tree for instance $X^c$. The following lemma lists several basic properties of a split. Recall that, given an instance $X$, $r(X)$ and $c(X)$ denote the number of active rows and active columns in $X$, respectively.

\begin{observation}
\label{obs: trivial decomp}
        If $X$ is a semi-permutation, then the following properties hold for any $(X,T)$-split at $U$: 
        \begin{itemize}
                \item $\sum_{v \in U} r(X^s_v) = r(X)$ 
                \item $\sum_{v \in U} c(X^s_v) = c(X)$ 
                \item $c(X^c) \leq |U|$

		\item $\sum_{v \in U} \WB_{T_v}(X^s_v) + \WB_{T^c}(X^c) = \WB_T(X)$. 
        \end{itemize}
\end{observation}

The first property holds since $X$ is a semi-permutation.
In order to establish the last property, consider any vertex $x\in V(T)$, and let $T'\in \set{T^c}\cup \set{T_v}_{v\in U}$ be the new tree to which $v$ belongs; if $x\in U$, then we set $T'=T_x$. It is easy to see that the cost of $v$ in tree $T'$ is the same as its cost in the tree $T$ (recall that the cost of a leaf vertex is $0$).

The last property can be viewed as a ``perfect decomposition'' property of the weak \WBone. 
In \Cref{thm: wb}, we will show an (approximate) decomposition property of strong \WBone.

\paragraph{Splitting by Lines.} 
We can also define the splitting with respect to any subset $\lset'\subseteq \lset$ of the auxiliary columns for $X$, analogously:
Notice that the lines in $\lset'$ partition the bounding box $B$ into a collection of internally disjoint strips, that we denote by $\set{S'_1,\ldots,S'_k}$.
We can then define the strip instances $X^s_i$ as containing all vertices of $X\cap S_i$ for all $1\leq i\leq k$, and the compressed instance $X^c$, that is obtained by collapsing, for each $1\leq i\leq k$, all active columns that lie in strip $S_i$, into a single column.
We also call these resulting instances \emph{a split of $X$ by $\lset'$}. 

We can also consider an arbitrary ordering $\sigma$ of the lines in $\lset$, such that the lines of $\lset'$ appear at the beginning of $\sigma$, and let $U \subseteq V(T(\sigma))$ contain all vertices $u$ for which the strip $S(u)$ is in $\set{S_i}_{1\leq i\leq k}$.
If we perform a split of $(X,T)$ at $U$, we obtain exactly the same strip instances $X^s_1,\ldots,X^s_k$, and the same compressed  instance $X^c$.

\subsection{Decomposition Theorem for \opt}
\label{sec:decompose opt}  
The following theorem gives a crucial decomposition property of $\opt$. The theorem is used in our algorithm for \minsat.

\begin{theorem}
        \label{thm: opt decomposes} 
        Let $X$ be a semi-permutation, let $T$ be any partitioning tree for $X$, let $U\subseteq V(T)$ be a subset of vertices of $T$ such that the strips in $\set{S(v)\mid v\in U}$ partition the bounding box, and let $\{X^c, \{X^s_v\}_{v \in U}\}$ be an $(X,T)$-split at $U$. Then:
        
         $$\sum_{v \in U}\opt(X^s_v) + \opt(X^c) \leq \opt(X).$$
\end{theorem}

\begin{proof}
        Let $Y$ be an optimal canonical solution for $X$, so that every point of $Y$ lies on an active row and an active column for $X$. For each vertex $v \in U$, let $Y_v$ denote the set of points of $Y$ that lie in the strip $S(v)$; recall that these points must lie in the interior of the strip. Therefore, $Y = \bigcup_{v \in U} Y_v$.

        For every vertex $v\in U$, let $\rset(v)$ denote the set of all rows $R$, such that: (i) $R$ contains a point of $X$; (ii) $R$ contains no point of $X_v^s$; and (iii) at least one point of $Y_v$ lies on $R$. We let $m_v=|\rset(v)|$.
        We need the following claim.
         
        \begin{claim}
                There is a feasible solution $\hat{Y}$ to instance $X^c$, containing at most $\sum_{v \in U} m_v$ points. 
        \end{claim}
        \begin{proof}
                We construct the solution $\hat{Y}$ for $X^c$ as follows. 
                Consider a vertex $v\in U$. Let $C_v$ be the unique column into which the columns lying in the strip $S(v)$ were collapsed. For every point $p\in Y_v$ that lies on a row $R\in\rset(v)$, we add a new point $\phi(p)$ on the intersection of row $R$ and column $C_v$ to the solution $\hat Y$. Once we process all vertices $v\in U$, we obtain a final set of points $\hat Y$. It is easy to verify that  $|\hat{Y}| = \sum_v m_v$. 
In order to see that $\hat Y$ is a feasible solution to instance $X^c$, it is enough to show that the set $X^c\cup \hat Y$ of points is satisfied. 
Notice that set $X\cup Y$ of points is satisfied, and set $X^c\cup \hat Y$ is obtained from $X\cup Y$ by collapsing sets of active columns lying in each strip $S(v)$ for $v\in U$. 
From Observation~\ref{obs: collapsing columns}, the point set $X^c\cup \hat Y$ is satisfied.  
\end{proof}
        
        We now consider the strip instances $X^s_v$ and prove the following claim, that will complete the proof of the lemma. 
        
        \begin{claim}
                For each vertex $v \in U$, $\opt(X^s_v) \leq |Y_v| - m_v$. 
        \end{claim}
        \begin{proof}
                Notice first that the point set $X^s_v \cup Y_v$ must be satisfied. 
                We will modify point set $Y_v$, to obtain another set $Y'_v$, so that $Y'_v$ remains a feasible solution for $X^s_v$, and $|Y'_v| \le |Y_v| - m_v$.
                
                In order to do so, we perform $m_v$ iterations. In each iteration, we will decrease the size of $Y_v$ by at least one, while also decreasing the cardinality of the set $\rset(v)$ of rows by exactly $1$, and maintaining the feasibility of the solution $Y_v$ for $X_v^s$.  
                
                In every iteration, we select two arbitrary rows $R$ and $R'$, such that: (i) $R\in \rset(v)$; (ii) $R'$ is an active row for instance $X^s_v$, and (iii) no point of $Y_v\cup X^s_v$ lies strictly between rows $R$ and $R'$. 
                We collapse the rows $R$ and $R'$ into the row $R'$. 
From Observation \ref{obs: collapsing columns}, the resulting new set $Y_v$ of points remains a feasible solution for instance $X^s_v$. 
We claim that $|Y_v|$ decreases by at least $1$. In order to show this, it is enough to show that there are two points $p,p'\in X^s_v\cup Y_v$, with $p\in R$, $p'\in R'$, such that the $x$-coordinates of $p$ and $p'$ are the same; in this case, after we collapse the rows, $x$ and $x'$ are mapped to the same point. Assume for contradiction that no such two points exist. Let $p\in R\cap (X_v^s\cup Y_v)$, $p'\in R'\cap Y_v$ be a pair of points with smallest horizontal distance. Such points must exist since $R$ contains a point of $X_v^s$ and $R'$ contains a point of $Y_v$. But then no other point of  $X^s_v\cup Y_v$ lies in $\rect_{p,p'}$, so the pair $(p,p')$ is not satisfied in $X^s_v\cup Y_v$, a contradiction.
        \end{proof}
        
\end{proof}

\subsection{Decomposition Theorem for the Strong \WBone}
\label{thm: wb}

In this section we prove the following theorem about the strong \WBone, that we use several times in our negative result.

\begin{theorem}\label{thm: WB for corner and strip instances}
	Let $X$ be a semi-permutation and $T$ be a partitioning tree for $X$. Let $U\subseteq V(T)$ be a set of vertices of $T$ such that the strips in $\set{S(v)\mid v\in U}$ partition the bounding box.
Let  $\{X^c, \{X^s_v\}_{v\in U} \}$ be  the split of $(X,T)$ at $U$.
	Then:
	
	\[\WB(X)\leq 4\WB(X^c)+8\sum_{v\in U}\WB(X^s_v)+O(|X|).   \]
\end{theorem}

This result is somewhat surprising. 
One can think of the expression $\WB(X^c)+\sum_{v\in U}\WB(X^s_v)$ as a WB-1 bound obtained by first cutting along the lines that serve as boundaries of the strips $S(v)$ for $v\in U$, and then cutting the individual strips. However, $\WB(X)$ is the maximum of $\WB_T(X)$ obtained over all trees $T$, including those that do not obey this partitioning order.

The remainder of this section is dedicated to the proof of \Cref{thm: WB for corner and strip instances}.
For convenience, we denote the instances $\{X^s_v\}_{v\in U}$ by $\set{X_1,\ldots,X_N}$, where the instances are indexed in the natural left-to-right order of their corresponding strips, and we denote the instance $X^c$ by $\tX$.
For each $1\leq i\leq N$, we denote by $\bset_i$ be the set of consecutive active columns containing the points of $X_i$, and we refer to it as a \emph{block}. For brevity, we also say ``Wilber bound'' to mean the strong \WBone in this section.

\paragraph{Forbidden Points.}
For the sake of the proof, we need the notion of forbidden points.
Let $\hX$ be some semi-permutation and $\hlset$ be the set of auxiliary columns for $\hX$.
Let $F\subseteq \hX$ be a set of points that we refer to as \emph{forbidden points}.
We now define the strong \WBone with respect to the forbidden points, $\WB^F(\hX)$.

Consider any permutation $\hsigma$ of the lines in $\hlset$. 
Intuitively, $\WB_{\hsigma}^F(\hX)$ counts all the crossings contributed to $\WB_{\hsigma}(\hX)$ but excludes all crossing pairs $(p,p')$ where at least one of $p,p'$ lie in $F$.
Similar to $\WB(X)$, we define $\WB^F(\hX) = \max_{\hsigma} \WB_{\hsigma}^F(\hX)$, where the maximum is over all permutations $\hsigma$ of the lines in $\hlset$.

Next, we define $\WB_{\hsigma}^F(\hX)$ more formally. 
Let $T=T(\hsigma)$ be the partitioning tree associated with $\hsigma$.
For each vertex $v \in V(T)$, let $L = L(v)$ be the line that belongs to $v$, and 
let $\cro_{\hsigma}(L)$ be the set of all crossings $(p,p')$ that contribute to $\cost(L)$; 
that is, $p$ and $p'$ are two points that lie in the strip $S(v)$ on two opposite sides of $L$, and no other point of $\hX\cap S(v)$ lies between the row of $p$ and the row of $p'$.
Let $\cro_{\hsigma} = \bigcup_{L\in \hlset}\cro_{\hsigma}(L)$.
Observe that $\WB_{\hsigma}(X) = |\cro_{\hsigma}|$ by definition.
We say that a crossing $(p,p')\in \cro_{\hsigma}(L)$ is \emph{forbidden} iff at least one of $p,p'$ lie in $F$; otherwise the crossing is \emph{allowed}. We let $\cro^F_{\hsigma}(L)$ be the set of crossings obtained from $\cro_{\hsigma}(L)$ by discarding all forbidden crossings. We then let $\cro^F_{\hsigma}=\bigcup_{L\in \hlset}\cro^F_{\hsigma}(L)$, and $\WB^F_{\hsigma}(\hX)=|\cro^F_{\hsigma}|$.

We emphasize that $\WB^F(\hX)$ is not necessarily the same as $\WB(\hX\setminus F)$, as some crossings of the instance $\hX\setminus F$ may not correspond to allowed crossings of instance $\hX$.

\paragraph{Proof Overview and Notation.}

Consider first the compressed instance $\tX$, that is a semi-permutation. We denote its set of active columns by $\tcset=\set{C_1,\ldots,C_N}$, where the columns are indexed in their natural left-to-right order. Therefore, $C_i$ is the column that was obtained by collapsing all active columns in strip $S_i$. %
It would be convenient for us to slightly modify the instance $\tX$ by simply multiplying all $x$-coordinates of the points in $\tX$ and of the columns in $\tcset$ by factor $2$. Note that this does not affect the value of the optimal solution or of the Wilber bound, but it ensures that every consecutive pair of columns in $\tcset$ is separated by a column with an integral $x$-coordinate. We let $\tlset$ be the set of all vertical lines with half-integral coordinates in the resulting instance $\tX$.

Similarly, we modify the original instance $X$, by inserting, for every consecutive pair $\bset_i,\bset_{i+1}$ of blocks, a new column with an integral coordinate that lies between the columns of $\bset_{i}$ and the columns of $\bset_{{i+1}}$. This transformation does not affect the optimal solution cost or the value of the Wilber bound. For all $1\leq i\leq N$, we denote $q_i=|\bset_i|$.
We denote by $\lset$ the set of all vertical lines with half-integral coordinates in the resulting instance $X$.

Consider any block $\bset_i$. We denote by $\lset_i=\set{L_i^1,\ldots,L_i^{q_i+1}}$ the set of $q_i+1$ consecutive vertical lines in $\lset$, where $L_i^1$ appears immediately before the first column of $\bset_i$, and $L_i^{q_i+1}$ appears immediately after the last column of $\bset_i$. Notice that $\lset=\bigcup_{i=1}^N\lset_i$.

Recall that our goal is to show that $\WB(X)\leq 4\WB(\tX)+8\sum_{i=1}^N\WB(X_i)+O(|X|)$. In order to do so, we fix a permutation $\sigma$ of $\lset$ that maximizes $\WB_{\sigma}(X)$, so that $\WB(X)=\WB_{\sigma}(X)$. We then gradually transform it into a permutation $\tsigma$ of $\tlset$, such that $\WB_{\tsigma}(\tX)\geq \WB_{\sigma}(X)/4-2\sum_{i=1}^N\WB(X_i)-O(|X|)$. This will prove that $\WB(X)\leq 4\WB(\tX)+8\sum_{i=1}^N\WB(X_i)+O(|X|)$.

In order to perform this transformation, we will process every block $\bset_i$ one-by-one. When block $\bset_i$ is processed, we will ``consolidate'' all lines of $\lset_i$, so that they will appear almost consecutively in the permutation $\sigma$, and we will show that this process does not increase the Wilber bound by too much. The final permutation that we obtain after processing every block $\bset_i$ can then be naturally transformed into a permutation $\tsigma$ of $\tlset$, whose Wilber bound cost is similar. The main challenge is to analyze the increase in the Wilber bound in every iteration. In order to facilitate the analysis, we will work with the Wilber bound with respect to forbidden points. Specifically, we will define a set $F\subseteq X$ of forbidden points, such that $\WB_{\sigma}^F(X)\geq \WB_{\sigma}(X)/4-\sum_{i=1}^N\WB(X_i)$. For every block $\bset_i$, we will also define a bit $b_i\in \set{0,1}$, that will eventually guide the way in which the lines of $\lset_i$ are consolidated. As the algorithm progresses, we will modify the set $F$ of forbidden points by discarding some points from it, and we will show that the increase in the Wilber bound with respect to the new set $F$ is small relatively to the original Wilber bound with respect to the old set $F$. We start by defining the set $F$ of forbidden points, and the bits $b_i$ for the blocks $\bset_i$. We then show how to use these bits in order to transform permutation $\sigma$ of $\lset$ into a new permutation $\sigma'$ of $\lset$, which will in turn be transformed into a permutation $\tsigma$ of $\tlset$.

From now on we assume that the permutation $\sigma$ of the lines in $\lset$ is fixed.

\paragraph{Defining the Set $F$ of Forbidden Points.}
Consider any block $\bset_i$, for $1\leq i\leq N$. We denote by $L^*_i\in \lset_i$ the vertical line that appears first in the permutation $\sigma$ among all lines of $\lset_i$, and we denote by $L^{**}_i\in \lset_i$ the line that appears last in $\sigma$ among all lines of $\lset_i$. 

We perform $N$ iteration. In iteration $i$, for $1\leq i\leq N$, we consider the block $\bset_i$. We let $b_i\in \set{0,1}$ be a bit chosen uniformly at random, independently from all other random bits. If $b_i=0$, then all points of $X_i$ that lie to the left of $L^*_i$ are added to the set $F$ of forbidden points; otherwise, all points of $X_i$ that lie to the right of $L^*_i$ are added to the set $F$ of forbidden points. We show that the expected number of the remaining crossings is large.

\begin{claim}\label{claim: many expected crossings}
	The expectation, over the choice of the bits $b_i$, of $|\cro^F_{\sigma}|$ is at least $|\WB(X)|/4-\sum_{i=1}^N\WB(X_i)$.
\end{claim}

\begin{proof}
	Consider any crossing $(p,p')\in \cro_{\sigma}$. We consider two cases. Assume first that there is some index $i$, such that both $p$ and $p'$ belong to $X_i$, and they lie on opposite sides of $L^*_i$. In this case, $(p,p')$ becomes a forbidden crossing with probability $1$. However, the total number of all such crossings is bounded by $\WB(X_i)$. Indeed, if we denote by $\hlset_i$ the set of all vertical lines with half-integral coordinates for instance $X_i$, then permutation $\sigma$ of $\lset$ naturally induces permutation $\sigma_i$ of $\hlset_i$. Moreover, any crossing $(p,p')\in\cro_{\sigma}$ with $p,p'\in X_i$ must also contribute to the cost of $\sigma_i$ in instance $X_i$. Since the cost of $\sigma_i$ is bounded by $WB(X_i)$, the number of crossings  $(p,p')\in \cro_{\sigma}$ with $p,p'\in X_i$ is bounded by $WB(X_i)$.
	
	Consider now any crossing $(p,p')\in \cro_{\sigma}$, and assume that there is no  index $i$, such that both $p$ and $p'$ belong to $X_i$, and they lie on opposite sides of $L^*_i$. Then with probability at least $1/4$, this crossing remains allowed. Therefore, the expectation of $|\cro^F_{\sigma}|$ is at least $|\cro_{\sigma}|/4-\sum_{i=1}^N\WB(X_i)=|\WB_{\sigma}(X)|/4-\sum_{i=1}^N\WB(X_i)=|\WB(X)|/4-\sum_{i=1}^N\WB(X_i) $.
\end{proof}

From the above claim, there is a choice of the bits $b_1,\ldots,b_N$, such that, if we define the set $F$ of forbidden points with respect to these bits as before, then $|\cro^F_{\sigma}|\geq \WB(X)/4-\sum_{i=1}^N\WB(X_i)$. From now on we assume that the values of the bits $b_1,\ldots,b_N$ are fixed, and that the resulting set $F$ of forbidden points satisfies that $|\cro^F_{\sigma}|\geq \WB(X)/4-\sum_{i=1}^N\WB(X_i)$.

\paragraph{Transforming $\sigma$ into $\sigma'$.}
We now show how to transform the original permutation $\sigma$ of $\lset$ into a new permutation $\sigma'$ of $\lset$, which we will later transform into a permutation $\tsigma$ of $\tlset$. We perform $N$ iterations. The input to the $i$th iteration is a permutation $\sigma_i$ of $\lset$ and a subset $F_i\subseteq F$ of forbidden points. The output of the iteration is a new permutation $\sigma_{i+1}$ of $\lset$, and a set $F_{i+1}\subseteq F_i$ of forbidden points. The final permutation is $\sigma'=\sigma_{N+1}$, and the final set $F_{N+1}$ of forbidden points will be empty. The input to the first iteration is $\sigma_1=\sigma$ and $F_1=F$. We now fix some $1\leq i\leq N$, and show how to execute the $i$th iteration. Intuitively, in the $i$th iteration, we consolidate the lines of $\lset_i$. Recall that we have denoted by $L^*_i,L^{**}_i\in \lset_i$ the first and the last lines of $\lset_i$, respectively, in the permutation $\sigma$. We only move the lines of $\lset_i$ in iteration $i$, so this ensures that, in permutation $\sigma_i$, the first line of $\lset_i$ that appears in the permutation is $L^*_i$, and the last line is $L^{**}_i$.

We now describe the $i$th iteration. Recall that we are given as input a permutation $\sigma_i$ of the lines of $\lset$, and a subset $F_i\subseteq F$ of forbidden points. 
We consider the block $\bset_i$ and the corresponding bit $b_i$.

\begin{figure}
	\begin{centering}
		\includegraphics[width=0.75\textwidth]{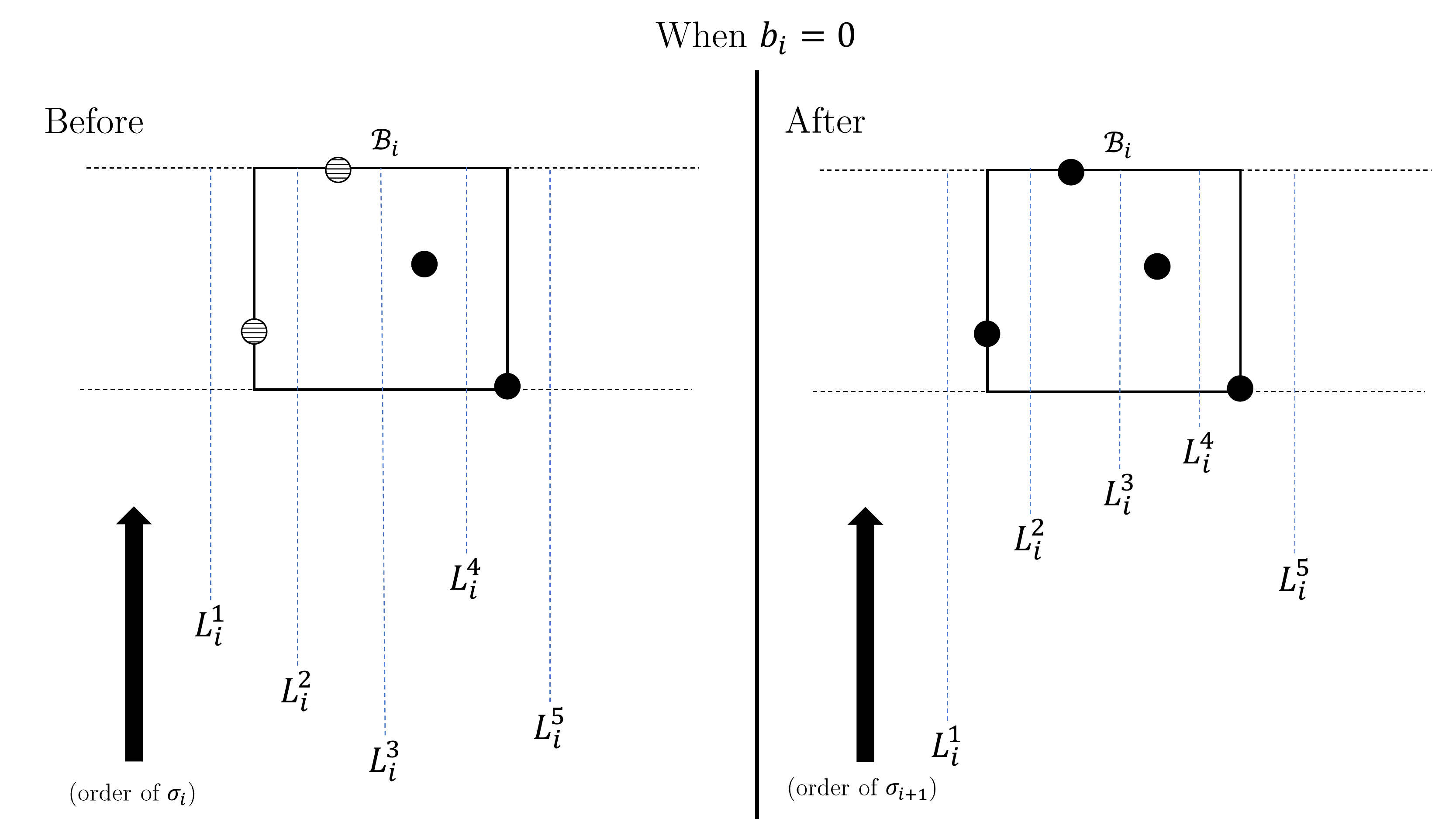}
		\par\end{centering}
	\begin{centering}
		\includegraphics[width=0.75\textwidth]{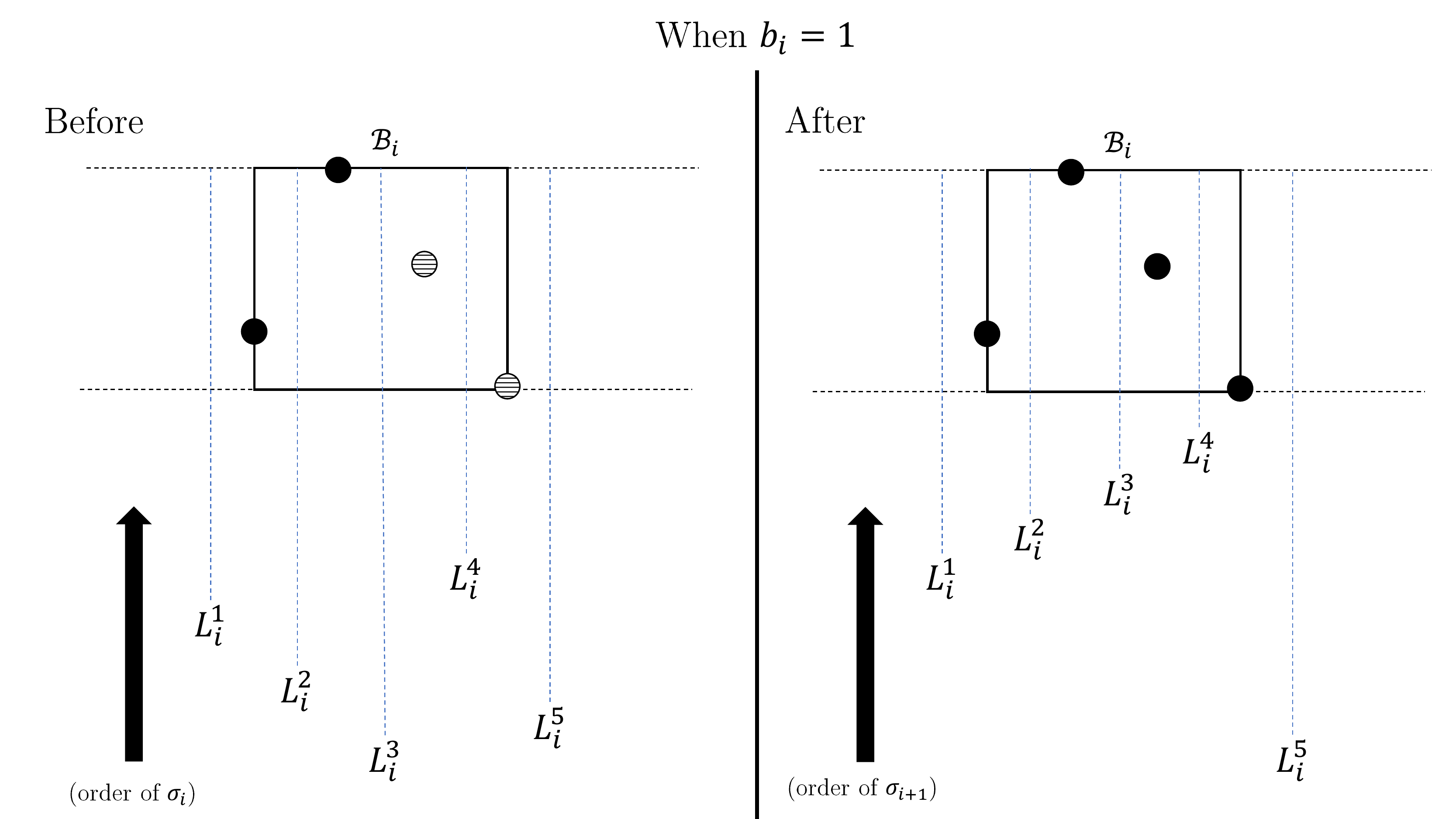}
		\par\end{centering}
	\caption{Modification from $\sigma_{i}$ to $\sigma_{i+1}$. In the figure,
		$\protect\lset_{i}=\{L_{i}^{1},\dots,L_{i}^{5}\}$, $L_{i}^{*}=L_{i}^{3}$
		and $L_{i}^{**}=L_{i}^{4}$. Points with horizontal strips are forbidden.\label{fig:change sigma}}
\end{figure}

Assume first that $b_i=0$; recall that in this case, all points of $X$ that lie on the columns of $\bset_i$ to the left of $L^*_i$ are forbidden (see \Cref{fig:change sigma}). We start by switching the locations of $L^*_i$ and $L^1_i$ in the permutation $\sigma_i$ (recall that $L^1_i$ is the leftmost line in $\lset_i$). Therefore, $L^1_i$ becomes the first line of $\lset_i$ in the resulting permutation. Next, we
consider the location of line $L^{**}_i$ in $\sigma_i$, and we place the lines $L_i^{q_i+1},L_i^2,L_i^3,\ldots,L_i^{q_i}$ in that location, in this order. This defines the new permutation $\sigma_{i+1}$.

Assume now that $b_i=1$;  recall that in this case, all points of $X$ that lie on the columns of $\bset_i$ to the right of $L^*_i$ are forbidden  (see \Cref{fig:change sigma}). We start by switching the locations of $L^*_i$ and $L_i^{q_i+1}$ in the permutation $\sigma_i$ (recall that $L_i^{q_i+1}$ is the rightmost line in $\lset_i$). Therefore, $L_i^{q_i+1}$ becomes the first line of $\lset_i$ in the resulting permutation. Next, we
consider the location of line $L^{**}_i$ in $\sigma_i$, and we place the lines $L_i^1,L_i^2,L_i^3,\ldots,L_i^{q_i}$ in that location, in this order. This defines the new permutation $\sigma_{i+1}$.

Lastly, we discard from $F_i$ all points that lie on the columns of $\bset_i$, obtaining the new set $F_{i+1}$ of forbidden points.

Once every block $\bset_i$ is processed, we obtain a final permutation $\sigma_{N+1}$ that we denote by $\sigma'$, and the final set $F_{N+1}=\emptyset$ of forbidden lines. The following lemma is central to our analysis. It shows that the Wilber bound does not decrease by much after every iteration. The Wilber bound is defined with respect to the appropriate sets of forbidden points.

\begin{lemma}\label{lemma: small change}
	For all $1\leq i\leq N$,  $\WB^{F_{i+1}}_{\sigma_{i+1}}(X)\geq \WB^{F_i}_{\sigma_i}(X)-\WB(X_i)-O(|X_i|)$.
\end{lemma}

Assume first that the lemma is correct. Recall that we have ensured that $\WB^{F_1}_{\sigma_1}(X)=\WB^F_{\sigma}(X)\geq \WB(X)/4-\sum_{i=1}^N\WB(X_i)$. Since $F_{N+1}=\emptyset$, this will ensure that:

\[\WB_{\sigma'}(X)\geq \WB^F_{\sigma}(X)-\sum_i\WB(X_i)-O(|X|)\geq \WB(X)/4-2\sum_i \WB(X_i)-O(|X|).\]

We now focus on the proof of the lemma.

\begin{proof}
	In order to simplify the notation, we denote $\sigma_i$ by $\hsigma$, $\sigma_{i+1}$ by $\hsigma'$. We also denote $F_i$ by $\hF$, and $F_{i+1}$ by $\hF'$.

	Consider a line $L\in \lset$. Recall that $\cro_{\hsigma}(L)$ is the set of all crossings that are charged to the line $L$ in permutation $\hsigma$. Recall that $\cro^{\hF}_{\hsigma}(L)\subseteq \cro_{\hsigma}(L)$ is obtained from the set  $\cro_{\hsigma}(L)$ of crossings, by discarding all crossings $(p,p')$ where $p\in \hF$ or $p'\in \hF$ holds. The set $\cro^{\hF'}_{\hsigma'}(L)$ of crossings is defined similarly.
	
	We start by showing that for every line $L\in \lset$ that does not lie in $\lset_i$, the number of crossings charged to it does not decrease, that is, $\cro^{\hF'}_{\hsigma'}(L)\geq\cro^{\hF}_{\hsigma}(L)$.
	
	\begin{claim}\label{claim: lines outside of block do not change}
		For every line $L\in \lset\setminus \lset_i$, $\cro^{\hF'}_{\hsigma'}(L)\geq\cro^{\hF}_{\hsigma}(L)$.
	\end{claim}
	\begin{proof}
		Consider any line $L\in \lset\setminus \lset_i$. Let $v\in V(T_{\hsigma})$ be the vertex of the partitioning tree $T_{\hsigma}$ corresponding to $\hsigma$ to which $L$ belongs, and let $S=S(v)$ be the corresponding strip. Similarly, we define $v'\in V(T_{\hsigma'})$ and $S'=S(v')$ with respect to $\hsigma'$.  
		Recall that $L^*_i$ is the first line of $\lset_i$ to appear in the permutation $\sigma$, and $L^{**}_i$ is the last such line. We now consider five cases.

		\paragraph{Case 1.}	The first case happens if $L$ appears before line $L^*_i$ in the permutation $\hsigma$. Notice that the prefixes of the permutations $\hsigma$ and $\hsigma'$ are identical up to the location in which $L^*_i$ appears in $\hsigma$. Therefore, $S=S'$, and $\cro_{\hsigma}(L)=\cro_{\hsigma'}(L)$. Since $\hF'\subseteq \hF$, every crossing that is forbidden in $\hsigma'$ was also forbidden in $\hsigma$. So $\cro_{\hsigma}^{\hF}(L)\subseteq \cro_{\hsigma'}^{\hF'}(L)$, and $\cro^{\hF'}_{\hsigma'}(L)\geq\cro^{\hF}_{\hsigma}(L)$.
		
		\paragraph{Case 2.} The second case happens if $L$ appears after $L^{**}_i$ in $\hsigma$. Notice that, if we denote by $\lset'\subseteq \lset$ the set of all lines of $\lset$ that lie before $L$ in $\hsigma$, and define $\lset''$ similarly for $\hsigma'$, then $ \lset'=\lset''$. Therefore, $S=S'$ holds. Using the same reasoning as in Case 1,  $\cro^{\hF'}_{\hsigma'}(L)\geq\cro^{\hF}_{\hsigma}(L)$.
		
		\paragraph{Case 3.}
		The third case is when $L$ appears between $L^*_i$ and $L^{**}_i$ in $\hsigma$, but neither boundary of the strip $S$ belongs to $\lset_i$. If we denote by $\lset'\subseteq \lset\setminus \lset_i$ the set of all lines of $\lset\setminus\lset_i$ that lie before $L$ in $\hsigma$, and define $\lset''\subseteq \lset\setminus\lset_i$ similarly for $\hsigma'$, then $ \lset'=\lset''$. Therefore, $S=S'$ holds. Using the same reasoning as in Cases 1 and 2,  $\cro^{\hF'}_{\hsigma'}(L)\geq\cro^{\hF}_{\hsigma}(L)$.
		
		\paragraph{Case 4.}
		The fourth case is when $L$  appears between $L^*_i$ and $L^{**}_i$ in the permutation $\hsigma$, and the left boundary of $S$ belongs to $\lset_i$. Notice that the left boundary of $S$ must either coincide with $L^*_i$, or appear to the right of it.
		
		Assume first that $b_i=0$, so we have replaced $L^*_i$ with the line $L_i^1$, that lies to the left of $L^*_i$. Since no other lines of $\lset_i$ appear in $\hsigma'$ until the original location of line $L^{**}_i$, it is easy to verify that the right boundary of $S'$ is the same as the right boundary of $S$, and its left boundary is the line $L_i^1$, that is, we have pushed the left boundary to the left. %
		In order to prove that
		$|\cro^{\hF'}_{\hsigma'}(L)|\geq|\cro^{\hF}_{\hsigma}(L)|$, we map every crossing $(p_1,p_2)\in \cro^{\hF}_{\hsigma}(L)$ to some crossing $(p_1',p_2')\in \cro^{\hF'}_{\hsigma'}(L)$, so that no two crossings of $\cro^{\hF}_{\hsigma}(L)$ are mapped to the same crossing of $\cro^{\hF'}_{\hsigma'}(L)$.
		
		\begin{figure}
			\begin{centering}
				\includegraphics[width=0.75\textwidth]{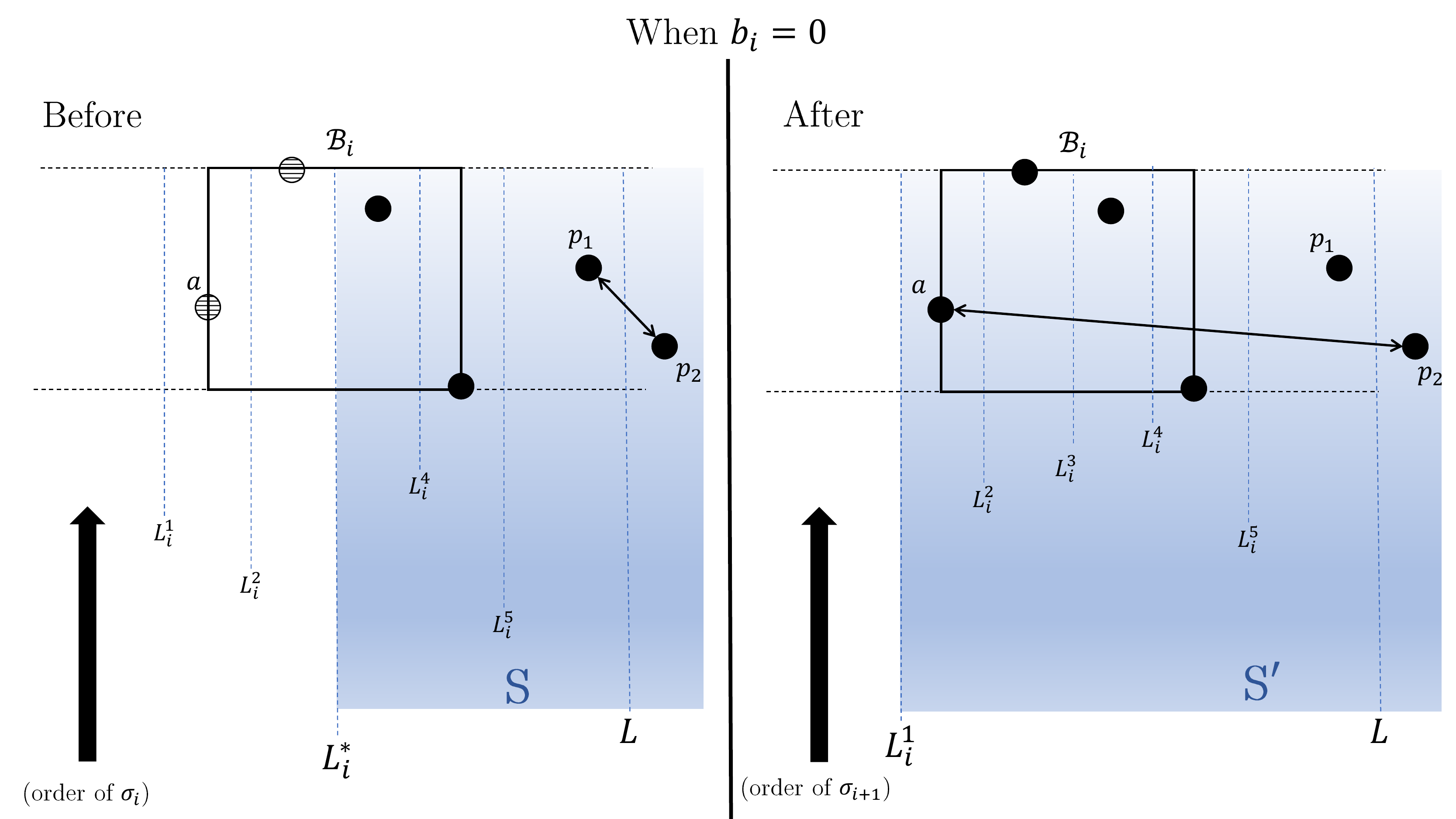}
				\par\end{centering}
			\caption{Illustration of the injective mapping of each crossing in $\cro^{\hF}_{\hsigma}(L)$ to a crossing in $\cro^{\hF'}_{\hsigma'}(L)$.
				Points with horizontal strips are forbidden points from $\hat{F}$.\label{fig:mapping}}
		\end{figure}
	
		Consider any crossing $(p_1,p_2)\in \cro^{\hF}_{\hsigma}(L)$ (see \Cref{fig:mapping}). We know that $p_1,p_2\in S$, and they lie on opposite sides of $L$. We assume w.l.o.g. that $p_1$ lies to the left of $L$. %
		Moreover, no point of $X\cap S$ lies between the row of $p_1$ and the row of $p_2$. 
		It is however possible that $(p_1,p_2)$ is not a crossing of $\cro_{\hsigma'}(L)$, since by moving the left boundary of $S$ to the left, we add more points to the strip, some of which may lie between the row of $p_1$ and the row of $p_2$. Let $A$ be the set of all points that lie between the row of $p_1$ and the row of $p_2$ in $S'$. Notice that the points of $A$ are not forbidden in $\hF'$. Let $a\in A$ be the point of $a$ whose row is closest to the row of $p_2$; if $A=\emptyset$, then we set $a=p_1$. Then $(a,p_2)$ defines a crossing in $\cro_{\hsigma'}(L)$, and, since neither point lies in $\hF'$, $(a,p_2)\in \cro^{\hF'}_{\hsigma'}(L)$. In this way, we map every crossing $(p_1,p_2)\in \cro^{\hF}_{\hsigma}(L)$ to some crossing $(p_1',p_2')\in \cro^{\hF'}_{\hsigma'}(L)$. It is easy to verify that no two crossings of $\cro^{\hF}_{\hsigma}(L)$ are mapped to the same crossing of $\cro^{\hF'}_{\hsigma'}(L)$. We conclude that $|\cro^{\hF'}_{\hsigma'}(L)|\geq|\cro^{\hF}_{\hsigma}(L)|$.

		Lastly, assume that $b_i=1$. Recall that the set of all points of $X$ lying between $L^*_i$ and $L_{i}^{q_i+1}$ is forbidden in $\hF$ but not in $\hF'$, and that we have replaced $L^*_i$ with the line $L_i^{q_i+1}$, that lies to the right of $L^*_i$. Therefore, the right boundary of $S$ remains the same, and the left boundary is pushed to the right. %
		In order to prove that
		$|\cro^{\hF'}_{\hsigma'}(L)|\geq|\cro^{\hF}_{\hsigma}(L)|$, we show that every crossing $(p_1,p_2)\in \cro^{\hF}_{\hsigma}(L)$ belongs to $\cro^{\hF'}_{\hsigma'}(L)$. Indeed, consider any crossing $(p_1,p_2)\in \cro^{\hF}_{\hsigma}(L)$. We know that $p_1,p_2\in S$, and they lie on opposite sides of $L$. We assume w.l.o.g. that $p_1$ lies to the left of $L$. Since $p_1$ cannot be a forbidden point, it must lie to the right of $L_i^{q_i+1}$. Moreover, no point of $X\cap S$ lies between the row of $p_1$ and the row of $p_2$. It is now easy to verify that $(p_1,p_2)$ is also a crossing in $\cro^{\hF'}_{\hsigma'}(L)$.
		
		\paragraph{Case 5.}
		The fifth case happens when $L$ appears between $L^*_i$ and $L^{**}_i$ in the permutation $\hsigma$, and the right boundary of $S$ belongs to $\lset_i$. This case is symmetric to the fourth case and is analyzed similarly.
	\end{proof}
	
	It now remains to analyze the crossings of the lines in $\lset_i$. We do so in the following two claims. The first claim shows that switching $L^*_i$ with $L^1_i$ or $L^{q_i+1}_i$ does not decrease the number of crossings.
	
	\begin{claim}
		If $b=0$, then	$|\cro^{\hF'}_{\hsigma'}(L_i^1)|\geq |\cro^{\hF}_{\hsigma}(L^*_i)|$; if $b=1$, then $|\cro^{\hF'}_{\hsigma'}(L_i^{q_i+1})|\geq |\cro^{\hF}_{\hsigma}(L^*_i)|$.
	\end{claim}
	
	\begin{proof}
		Assume first that $b=0$, so we have replaced $L^*_i$ with $L_i^1$ in the permutation. As before, we let $v\in V(T_{\hsigma})$ be the vertex to which $L^*_i$ belongs, and we let $S=S(v)$ be the corresponding strip. Similarly, we define $v'\in V(T_{\hsigma'})$ and $S'=S(v')$ with respect to line $L_i^1$ and permutation $\hsigma'$. Notice that, until the appearance of $L^*_i$ in $\hsigma$, the two permutations are identical. Therefore, $S=S'$ must hold. Recall also that all points of $X$ that lie between $L^*_i$ and $L_i^1$ are forbidden in $\hF$, but not in $\hF'$. In order to show that $|\cro^{\hF'}_{\hsigma'}(L_i^1)|\geq |\cro^{\hF}_{\hsigma}(L^*_i)|$, it is enough to show that every crossing $(p_1,p_2)\in \cro^{\hF}_{\hsigma}(L^*_i)$ also lies in $\cro^{\hF'}_{\hsigma'}(L_i^1)$. 
		
		Consider now some crossing $(p_1,p_2)\in \cro^{\hF}_{\hsigma}(L^*_i)$. Recall that one of $p_1,p_2$ must lie to the left of $L_i^*$ and the other to the right of it, with both points lying in $S$. Assume w.l.o.g. that $p_1$ lies to the left of $L_i^*$. Since $p_1\not\in \hF$, it must lie to the left of $L^1_i$. Moreover, no point of $X\cap S$ may lie between the row of $p_1$ and the row of $p_2$. It is then easy to verify that $(p_1,p_2)$ is also a crossing in $\cro^{\hF'}_{\hsigma'}(L_i^1)$, and so $|\cro^{\hF'}_{\hsigma'}(L_i^1)|\geq |\cro^{\hF}_{\hsigma}(L^*_i)|$.
		
		The second case, when $b=1$, is symmetric.
	\end{proof}
	
	Lastly, we show that for all lines $L\in \lset_i\setminus\set{L_i^*}$, their total contribution to $\cro^{\hF}_{\hsigma}$ is small, in the following claim.
	
	\begin{claim}\label{claim: lines in li contribute little}
		$\sum_{L\in \lset_i\setminus\set{L_i^*}}|\cro^{\hF}_{\hsigma}(L)|\leq \WB(X_i)+O(|X_i|)$.
	\end{claim}
	
	Assume first that the claim is correct. We have shown so far that the total contribution of all lines in $\lset_i\setminus\set{L_i^*}$ to $\cro^{\hF}_{\hsigma}$ is at most $\WB(X_i)+O(|X_i|)$; that the contribution of one of the lines $L_i^1,L_i^{q_i+1}$ to $\cro^{\hF'}_{\hsigma'}$ is at least as large as the contribution of $L_i^*$ to $\cro^{\hF}_{\hsigma}$; and that for every line $L\not\in \lset_i$, its contribution to $\cro^{\hF'}_{\hsigma'}$ is at least as large as its contribution to $\cro^{\hF}_{\hsigma}$. It then follows that $|\cro^{\hF'}_{\hsigma'}|\geq |\cro^{\hF}_{\hsigma}|-\WB(X_i)+O(|X_i|)$, and so
	$\WB^{F_{i+1}}_{\sigma_{i+1}}(X)\geq \WB^{F_i}_{\sigma_i}(X)-\WB(X_i)+O(|X_i|)$. Therefore, in order to prove Lemma \ref{lemma: small change}, it is now enough to prove Claim \ref{claim: lines in li contribute little}.
	
	\begin{proofof}{Claim \ref{claim: lines in li contribute little}}
		Consider some line $L\in \lset_i\setminus\set{L_i^*}$, and let $v\in V(T_{\hsigma})$ be the vertex to which $L$ belongs. Notice that $L$ appears in $\hsigma$ after $L_i^*$. Therefore, if $S=S(v)$ is the strip that $L$ partitioned, then at least one of the boundaries of $S$ lies in $\lset_i$. If exactly one boundary of $S$ lies in $\lset_i$, then we say that $S$ is an \emph{external strip}; otherwise, we say that $S$ is an \emph{internal} strip. Consider now some crossing $(p,p')\in \cro_{\hsigma}(S)$. Since $L\in \lset_i$, and at least one boundary of $S$ lies in $\lset_i$, at least one of the points $p,p'$ must belong to $X_i$. If exactly one of $p,p'$ lies in $X_i$, then we say that $(p,p')$ is a type-1 crossing; otherwise it is a type-2 crossing. Notice that, if $S$ is an internal strip, then only type-2 crossings of $L$ are possible. We now bound the total number of type-1 and type-2 crossings separately, in the following two observations.

		\begin{observation}\label{obs: type 2 crossings}
			The total number of type-2 crossings in $\bigcup_{L\in \lset_i\setminus\set{L_i^*}}\cro_{\hsigma}(L)$ is at most $\WB(X_i)$.
		\end{observation}
		
		\begin{proof}
			Permutation $\hsigma$ of the lines in $\lset$ naturally induces a permutation $\hsigma_i$ of the lines in $\lset_i$. The number of type-2 crossings charged to all lines in $\lset_i$ is then at most $\WB_{\hsigma_i}(X_i)\leq \WB(X_i)$.%
		\end{proof}

		\begin{observation}\label{obs: type 2 crossings}
			The total number of type-1 crossings in $\bigcup_{L\in \lset_i\setminus\set{L_i^*}}\cro_{\hsigma}(L)\leq O(|X_i|)$.
		\end{observation}
		\begin{proof}
			Consider a line $L\in \lset_i\setminus\set{L_i^*}$, and let $S$ be the strip that it splits. Recall that, if there are any type-1 crossings in $\cro_{\hsigma}(L)$, then $S$ must be an external strip. Line $L$ partitions $S$ into two new strips, that we denote by $S'$ and $S''$. Notice that exactly one of these strips (say $S'$) is an internal strip, and the other strip is external. Therefore, the points of $X_i\cap S'$ will never participate in type-$1$ crossings again. Recall that, from Observation \ref{obs: number of crossings at most min of both sides}, the total number of crossings in $\cro_{\hsigma}(L)$ is bounded by $2|S'\cap X_i|$. We say that the points of $S'\cap X_i$ pay for these crossings. Since every point of $X_i$ will pay for a type-1 crossing at most once, we conclude that the total number of type-1 crossings in $\bigcup_{L\in \lset_i\setminus\set{L_i^*}}\cro_{\hsigma}(L)$ is bounded by $2|X_i|$.
		\end{proof}
		
		We conclude that the total number of all crossings in $\bigcup_{L\in \lset_i\setminus\set{L_i^*}}\cro_{\hsigma}(L)$ is at most $\WB(X_i)+O(|X_i|)$. Since, for every line $L$, $\cro_{\hsigma}^{\hF}(L)\subseteq \cro_{\hsigma}(L)$, we get that $\sum_{L\in \lset_i\setminus\set{L_i^*}}|\cro^{\hF}_{\hsigma}(L)|\leq \WB(X_i)+O(|X_i|)$.
	\end{proofof}			
\end{proof}

To summarize, we have transformed a permutation $\sigma$ of $\lset$ into a permutation $\sigma'$ of $\lset$, and we have shown that  $\WB_{\sigma'}(X)\geq \WB(X)/4-2\sum_{i=1}^N\WB(X_i)-O(|X|)$.

\paragraph{Transforming $\sigma'$ into $\tsigma$.}
In this final step, we transform the permutation $\sigma'$ of $\lset$ into a permutation $\tsigma$ of $\tlset$, and we will show that $\WB_{\tsigma}(\tX)\geq \WB_{\sigma'}(X)-|X|$.

The transformation is straightforward. Consider some block $\bset_i$, and the corresponding set $\lset_i\subseteq \lset$ of lines.
Recall that the lines in $\lset_i$ are indexed $L^1_i,\ldots,L^{q_i+1}_i$ in this left-to-right order, where $L^1_i$ appears to the left of the first column of $\bset_i$, and $L^{q_i+1}_i$ appears to the right of the last column of $\bset_i$. Recall also that, in the current permutation $\sigma'$, one of the following happens: either (i) line $L^1_i$ appears in the permutation first, and lines $L^{q_i+1}_i,L^2_i,\ldots, L^q_i$ appear at some later point consecutively in this order; or (ii) line $L^{q_i+1}_i$ appears in the permutation first, and lines $L^1_i,L^2_i,\ldots,L^q_i$ appear somewhere later in the permutation consecutively in this order. Therefore, for all $2\leq j\leq q$, line $L_i^j$ separates a strip whose left boundary is $L_i^{j-1}$ and right boundary is $L_i^{q_i+1}$. It is easy to see that the cost of each such line $L_i^j$ in permutation $\sigma'$ is bounded by the number of points of $X$ that lie on the unique active column that appears between $L_i^{j-1}$ and $L_i^j$. The total cost of all such lines is then bounded by $|X_i|$.

Let $\tsigma^*$ be a sequence of lines obtained from $\sigma'$ by deleting, for all $1\leq i\leq N$, all lines $L_i^2,\ldots,L_i^q$ from it. Then $\tsigma^*$ naturally defines a permutation $\tsigma$ of the set $\tlset$ of vertical lines. Moreover, from the above discussion, the total contribution of all deleted lines to $\WB_{\sigma'}(X)$ is at most $|X|$, so $\WB_{\tsigma}(\tX)\geq \WB_{\sigma'}(X)-|X|\geq \WB(X)/4-2\sum_i\WB(X_i)-O(|X|)$. We conclude that $\WB(\tX)\geq \WB_{\tsigma}(\tX)\geq \WB(X)/4-2\sum_i\WB(X_i)-O(|X|)$, and $\WB(X)\leq 4\WB(\tX)+8\sum_i\WB(X_i)+O(|X|)$.

\section{Separation of $\opt$ and the Strong Wilber Bound}
\label{sec: negative}

In this section we present our negative results,  proving Theorem \ref{thm:intro_WB}.
We start with a useful claim that allows us to upper-bound WB-costs in Section~\ref{subsec: upper bound WB}.
We then define several important instances and analyze their properties in Section~\ref{sec: sequences}.
In Section~\ref{sec: Instance construction}, we present two structural tools that we use in analyzing our instance.
From Section~\ref{sec:construction} onward, we describe our construction and its analysis.

\subsection{Upper-Bounding WB costs} 
\label{subsec: upper bound WB}

Recall that for an input set $X$ of points, a partitioning tree $T$ of $X$, and a vertex $v\in V(T)$, we denoted by $N(v)$ the number of points of $X$ that lie in the strip $S(X)$.
We use the following claim for bounding the WB cost of subsets of vertices of $T$.

\begin{claim}\label{claim: bounds on WB for path and subtree}
Consider a set $X$ of points that is a semi-permutation, an ordering $\sigma$ of the auxiliary columns in $\lset$ and the corresponding partitioning tree $T=T(\sigma)$.  
 Let $v\in V(T)$ be any vertex of the tree. Then the following hold:
	\begin{itemize}
		\item Let $T_v$ be the sub-tree of $T$ rooted at $v$. Then $\sum_{u\in V(T_v)}\cost(u) = \WB_{T_v}(X \cap S(v)) \leq O(\size(v)\log(\size(v)))$.
		\item Let $u$ be any descendant vertex of $v$, and let $P$ be the unique path in $T$ connecting $u$ to $v$. Then $\sum_{z\in V(P)}\cost(z)\leq 2\size(v)$.
	\end{itemize}
\end{claim}

\begin{proof}
The first assertion follows from the definition of the weak \WBone and \Cref{cor: upper bound on WB}. We now prove the second assertion. Denote $P=(v=v_1,v_2,\ldots,v_k=u)$. For all $1<i\leq k$, we let $v'_i$ be the unique sibling of the vertex $v_i$ in the tree $T$. We also let $X_i$ be the set of points of $X$ that lie in the strip $S(v'_i)$, and we let $X'$ be the set of all points of $X$ that lie in the strip $S(v_k)$. 
It is easy to verify that $X_2,\ldots,X_k,X'$ are all mutually disjoint (since the strips $\{S(v'_i)\}_{i=2}^q$ and $S(v_k)$ are disjoint), and that they are contained in $X\cap S(v)$. 
Therefore, $\sum_{i=2}^q|X_i|+|X'|\leq \size(v)$.
	
	From Observation \ref{obs: number of crossings at most min of both sides}, for all $1\leq i<k$, $\cost(v_i)\leq 2\size(v'_i)=2|X_i|$, and $\cost(v_k)\leq 2\size(v_k)=2|X'|$. Therefore, $\sum_{z\in V(P)}\cost(z)\leq 2\sum_{i=2}^q|X_i|+2|X'|\leq 2\size(v)$.
\end{proof}

\subsection{Some Important Instances}
\label{sec: sequences}

 \paragraph{Monotonically Increasing Sequence.} We say that an input set $X$ of points is \emph{monotonically increasing} iff $X$ is a permutation, and moreover for every pair $p,p'\in X$ of points, if $p.x <p'.x$, then $p.y<p'.y$ must hold.
It is well known that the value of the optimal solution of monotonically increasing sequences is low, and we exploit this fact in our negative results.

\begin{observation}\label{obs: solution for monotone increasing set}
        If $X$ is a monotonically increasing set of points, then $\opt(X)\leq |X|-1$.
\end{observation}
\begin{proof}
We order points in $X$ based on their $x$-coordinates as $X = \{p_1,\ldots, p_m\}$ such that $p_1.x < p_2.x <\ldots < p_m.x$.
For each $i =1,\ldots, m-1$ we define $q_i = ((p_i).x, (p_{i+1}).y)$ and the set $Y = \{q_1,\ldots, q_{m-1}\}$. It is easy to verify that $Y$ is a feasible solution for $X$.   
\end{proof}

\paragraph{Bit Reversal Sequence (\BRS).} 
We use the geometric variant of \BRS, which is more intuitive and easier to argue about. 
Let $\rset \subseteq {\mathbb N}$ and $\cset \subseteq {\mathbb N}$ be sets of integers (which are supposed to represent sets of active rows and columns.)  
The instance ${\sf BRS}(i,\rset,\cset)$ is only defined when $|\rset|=|\cset|=2^i$. It contains $2^i$ points, and it is a permutation, whose sets of active rows and columns are exactly $\rset$ and $\cset$ respectively; so $|\rset| = |\cset| = 2^i$. 
We define the instance recursively. The base of the recursion is instance ${\sf BRS}(0,\{C\}, \{R\})$, containing a single point at the intersection of row $R$ and column $C$. 
Assume now that we have defined, for all $1\leq i'\leq i$, and any sets $\rset',\cset'$ of $2^{i'}$ integers, the corresponding instance ${\sf BRS}(i,\rset',\cset')$. We define  instance ${\sf BRS}(i+1,\rset,\cset)$, where $|\rset| = |\cset| = 2^{i+1}$, as follows. 

Consider the columns in $\cset$ in their natural left-to-right order, and define $\cset_{left}$ to be the first $2^i$ columns and $\cset_{right} = \cset \setminus \cset_{left}$. 
Denote $\rset= \{R_1,\ldots, R_{2^{i+1}}\}$, where the rows are indexed in their natural bottom to top order, and let $\rset_{even}=\{R_2,R_4,\ldots, R_{2^{i+1}}\}$ and $\rset_{odd}=\{R_1,R_3,\ldots, R_{2^{i+1}-1}\}$ be the sets of all even-indexed and all odd-indexed rows, respectively.
Notice that $|\cset_{left}| = |\cset_{right}| = |\rset_{even}| = |\rset_{odd}| = 2^i$. 
The instance ${\sf BRS}(i+1, \rset,\cset)$ is defined to be ${\sf BRS}(i,\rset_{odd},\cset_{left}) \cup {\sf BRS}(i,\rset_{even}, \cset_{right})$. 

For $n= 2^i$, we denote by ${\sf BRS}(n)$ the instance $\BRS(i,\cset,\rset)$, where $\cset$ contains all columns with integral $x$-coordinates from $1$ to $n$, and $\rset$ contains all rows with integral $y$-coordinates from $1$ to $n$;
see Figure~\ref{fig: brs} for an illustration.

\begin{figure}[h]
    \centering
    \includegraphics[width=0.7\textwidth]{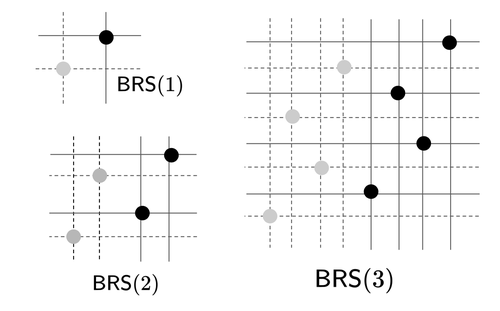}
    \caption{Bit reversal sequences ${\sf BRS}(i)$ for $i=1,2,3$. The sequence ${\sf BRS}(i)$ is constructed by composing two copies of ${\sf BRS}(i-1)$. The left copy is shown in gray, while the right copy is in black. The active rows and columns of the left copies are shown as dotted lines.}
    \label{fig: brs}
\end{figure}

It is well-known that, if $X$ is a bit-reversal sequence on $n$ points, then $\opt(X)\geq \Omega(n \log n)$. 

\begin{claim}
\label{claim:BRS complexity} 
Let $X= {\sf BRS}(i,\cset,\rset)$, for any $i\geq 0$ and any sets $\cset$ and $\rset$ of columns and rows, respectively, with $|\rset|=|\cset|=2^i$. Then $|X| = 2^i$, and $\opt(X) \geq \frac{\WB(X)}{2} \geq \frac{|X|(\log |X|-2)+1}{2} $.  
\end{claim} 

The original proof by Wilber~\cite{wilber} was given in the BST tree view. 
We provide a geometric proof here for completeness. 

\begin{proof} 
In order to prove the claim, it is sufficient to show that $\WB(X) \geq |X|(\log |X|-2)+1$.
We consider a balanced cut sequence $\sigma$ that always cuts a given strip into two sub-strips containing the same number of active columns. 
Notice that these cuts follow the inductive construction of the instance, in the reverse order.
We will prove by induction that,  for any $i\geq 1$, if $X={\sf BRS}(i,\cset,\rset)$, for any sets $\cset$ and $\rset$ of columns and rows, respectively, with $|\rset|=|\cset|=2^i$, then $\WB_{\sigma}(X)\geq 2^i\cdot(i-2)+1$. 

The base case, when $i = 0$, is obvious since the Wilber bound is always non-negative. 

We now consider some integer $i>0$, and we let  $X={\sf BRS}(i,\cset,\rset)$, for some sets $\cset$ and $\rset$ of columns and rows, respectively, with $|\rset|=|\cset|=2^i$.

Consider the first line $L$ in $\sigma$, that partitions the point set $X$ into sets $X'$ and $X''$. From the construction of $\BRS$, we get that $X'={\sf BRS}(i-1,\rset_{odd},\cset')$ and $X''={\sf BRS}(i-1,\rset_{even},\cset'')$ (as before, $\rset_{odd}$ and $\rset_{even}$ are the sets of all odd-indexed and all even-indexed rows, respectively; set $\cset'$ contains the leftmost $2^{i-1}$ active columns, and set $\cset''$ contains the remaining $2^{i-1}$ columns).
Moreover, from the construction, it is easy to see that every consecutive (with respect to their $y$-coordinates) pair of points in $X$ creates a crossing of $L$, so the number of crossings of $L$ is $2^i-1$. If we denote by $\sigma'$ and by $\sigma''$ the partitioning sequence induced by $\sigma$ for $X'$ and $X''$ respectively, then $\WB_{\sigma}(X)=\cost(L)+\WB_{\sigma'}(X')+\WB_{\sigma''}(X'')$. Moreover, by the induction hypothesis, $\WB_{\sigma'}(X'),\WB_{\sigma''}(X'')\geq (i-3)\cdot 2^{i-1}+1$. Altogether, we get that:
\[ \WB_{\sigma}(X)\geq (2^i-1)+2\cdot \left (2^{i-1}(i-3)+1\right ) \geq 2^i(i-2)+1.  \]
\end{proof}

\subsection{Two Additional Tools}\label{sec: Instance construction}
We present two additional technical tools that we use in our construction.

\paragraph{Tool 1: Exponentially Spaced Columns.} 
Recall that we defined the bit reversal instance ${\sf BRS}(\ell, \rset, \cset)$, where $\rset$ and $\cset$ are sets of $2^{\ell}$ rows and columns, respectively, that serve in the resulting instance as the sets of active rows and columns. The instance contains $n=2^{\ell}$ points. 
In the Exponentially-Spaced BRS instance $\ESBRS(\ell,\rset)$, we are still given a set $\rset$ of $2^{\ell}$ rows that will serve as active rows in the resulting instance, but we define the set $\cset$ of columns in a specific way. For an integer $i$, $C_i$ be the column whose $x$-coordinate is $i$. We then let $\cset$ contain, for each $0\leq  i< 2^{\ell}$, the column $C_{2^i}$. Denoting $N=2^n=2^{2^{\ell}}$, the $x$-coordinates of the columns in $\cset$ are $\set{1,2,4,8,\ldots,N/2}$. 
The instance is then defined to be $\BRS(\ell,\rset,\cset)$ for this specific set $\cset$ of columns.
Notice that the instance contains $n= \log N=  2^{\ell}$ input points.  

It is easy to see that any point set $X= \ESBRS(\ell, \rset)$ satisfies $\opt(X) = \Omega(n \log n)$.  
We remark that this idea of exponentially spaced columns is inspired by the instance used by Iacono~\cite{in_pursuit} to prove a gap between the weak WB-1 bound and $\opt(X)$ (see Appendix~\ref{sec: Iacono} for more details).  
However, Iacono's instance is tailored to specific partitioning tree $T$, and it is clear that there is another partitioning tree $T'$ with $\opt(X) = \Theta(WB_{T'}(X))$.
Therefore, this instance does not give a separation result for the strong WB-1 bound, and in fact it does not provide negative results for the weak WB-1 bound when the input point set is a permutation.

\paragraph{Tool 2: Cyclic Shift of Columns.}
Suppose we are given a point set $X$, and let $\cset'=\set{C_0,\ldots,C_{N-1}}$ be any set of columns indexed in their natural left-to-right order, such that all points of $X$ lie on columns of $\cset$ (but some columns may contain no points of $X$).
Let $0\leq s<N$ be any integer. We denote by $X^s$ a cyclic shift of $X$ by $s$ units, obtained as follows. 
For every point $p\in X$, we add a new point $p^s$ to $X^s$, whose $y$-coordinate is the same as that of $p$, and whose $x$-coordinate is $p.x + s \mod N$.
In other words, we shift the point $p$ by $s$ steps to the right in a circular manner. Equivalently, we move the last $s$ columns of $\cset'$ to the beginning of the instance.
The following claim will show that the value of the optimal solution does not decrease significantly in the shifted instance.

\begin{claim}\label{claim: shifting does not hurt opt}
Let $X$ be any point set that is a semi-permutation. Let $0\leq s<N$ be a shift value, and let $X'=X^s$ be the instance obtained from $X$ by a cyclic shift of its points by $s$ units to the right. 
Then $\opt(X')\geq \opt(X)-|X|$.
\end{claim}

\begin{proof}
Let $Y'$ be an optimal canonical solution to instance $X'$. 
We partition $Y'$ into two subsets: set $Y'_1$ consists of all points lying on the first $s$ columns with integral coordinates, and set $Y'_2$ consists of all points lying on the remaining columns. 
We also partition the points of $X'$ into two subsets $X'_1$ and $X'_2$ similarly.  Notice that $X'_1\cup Y_1'$ must be a satisfied set of points, and similarly, $X'_2\cup Y'_2$ is a satisfied set of points.
	
Next, we partition the set $X$ of points into two subsets: set $X_1$ contains all points lying on the last $s$ columns with integral coordinates, and set $X_2$ contains all points lying on the remaining columns. 
Since $X_1$ and $X_2$ are simply horizontal shifts of the sets $X'_1$ and $X'_2$ of points, we can define a set $Y_1$ of $|Y_1'|$ points such that $Y_1$ is a canonical feasible solution for $X_1$, and we can define a set $Y_2$ for $X_2$ similarly. 
Let $C$ be a column with a half-integral $x$-coordinate that separates $X_1$ from $X_2$ (that is, all points of $X_1$ lie to the right of $C$ while all points of $X_2$ lie to its left.) We construct a new set $Z$ of points, of cardinality $|X|$, such that $Y_1\cup Y_2\cup Z$ is a feasible solution to instance $X$. In order to construct the point set $Z$,
for each point $p \in X$, we add a point $p'$ with the same $y$-coordinate, that lies on column $C$, to $Z$. Notice that $|Z|=|X|$.

We claim that $Z \cup (Y_1 \cup Y_2)$ is a feasible solution  for $X$, and this will complete the proof.  
Consider any two points $p,q \in Z \cup (Y_1 \cup Y_2) \cup X$ that are not collinear.
Let $B_1$ and $B_2$ be the strips obtained from the bounding box $B$ by partitioning it with column $C$, so that $X_1\subseteq B_1$ and $X_2\subseteq B_2$. 
We consider two cases:  
\begin{itemize} 
\item The first case is when both $p$ and $q$ lie in the interior of the same strip, say $B_1$. But then $p,q\in X_1\cup Y_1$ must hold, and, since set $X_1\cup Y_1$ of points is satisfied, the pair $(p,q)$ of points is satisfied in $X\cup Y_1\cup Y_2\cup Z$.

\item The second case is when one of the points (say $p$) lies in the interior of one of the strips (say $B_1$), while the other point either lies on $C$, or in the interior of $B_2$. Then $p\in X_1\cup Y_1$ must hold. Moreover, since $Y_1$ is a canonical solution for $X_1$, point $p$ lies on a row that is active for $X_1$. Therefore, some point $p'\in X_1$ lies on the same row (where possibly $p'=p$). But then a copy of $p'$ that was added to the set $Z$ and lies on the column $C$ satisfies the pair $(p,q)$. 
\end{itemize} 
\end{proof}

\subsection{Construction of the Bad Instance} 
\label{sec:construction}
We construct two instances: instance $\hat X$ on $N^*$ points, that is a semi-permutation (but is somewhat easier to analyze), and instance $X^*$ in $N^*$ points, which is a permutation; the analysis of instance $X^*$ heavily relies on the analysis of instance $\hat X$. 
We will show that the optimal solution value of both instances is $\Omega(N^*\log\log N^*)$, but the cost of the Wilber Bound is at most $O(N^*\log\log\log N^*)$.
Our construction uses the following three parameters. We let $\ell\geq 1$ be an integer, and we set $n=2^{\ell}$ and $N = 2^n$.

\subsubsection{First Instance}
In this subsection we construct our first final instance $\hat X$, which is a semi-permutation containing $N$ columns. 
Intuitively, we create $N$ instances $X^0,X^1,\ldots,X^{N-1}$, where instance $X^s$ is an exponentially-spaced \BRS\ instance that is shifted by $s$ units. We then stack these instances on top of one another in this order.

Formally, for all $0\leq j\leq N-1$, we define a set $\rset_j$ of $n$ consecutive rows with integral coordinates, such that the rows of $\rset_0,\rset_1,\ldots,\rset_{N-1}$ appear in this bottom-to-top order. Specifically, set $\rset_j$ contains all rows whose $y$-coordinates are in $\set{jn+1,jn+2,\ldots,(j+1)n}$. 

For every integer $0\leq s\leq N-1$, we define a set of points $X^s$, which is a cyclic shift of instance $\ESBRS(\ell, \rset_s)$ by $s$ units.  
Recall that $|X^s| = 2^{\ell} = n$ and that the points in $X^s$ appear on the rows in $\rset_s$ and a set $\cset_s$ of columns, whose $x$-coordinates are in $\{\paren{2^j + s} \mod N: 1\leq j\leq n\}$.  
We then let our final instance be $\hat{X} = \bigcup_{s=0}^{N-1} X^s$.
From now on, we denote by $N^*=|\hX|$. Recall that $|N^*|=N\cdot n=N\log N$.

Observe that the number of active columns in $\hat{X}$ is $N$. Since the instance is symmetric and contains $N^* = N \log N$ points, every column contains exactly $\log N$ points. 
Each row contains exactly one point, so $\hat X$ is a semi-permutation.   
(See Figure~\ref{fig:hard-instance} for an illustration).   

\begin{figure}
	\centering
	\includegraphics[width=0.6\textwidth]{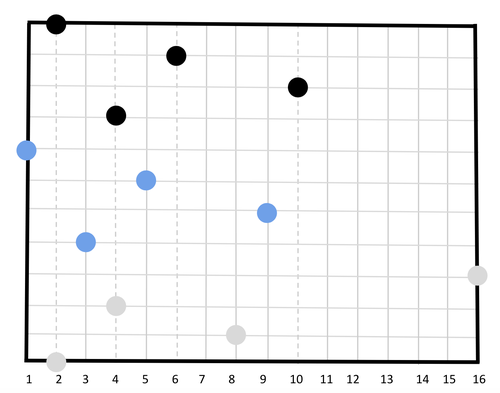}
	\caption{A small example of $X^0$, $X^1$, and $X^2$ (shown in different colors) placed from bottom to top, when the parameter $\ell=2$. In this case, each $X^s$ contains  $4$ points, and there are $N = 16$ copies of shifted sub-instances (but we show only $3$). 
The dotted columns are the active columns of $X^2$.}
	\label{fig:hard-instance}
\end{figure}

Lastly, we need the following bound on the value of the optimal solution of instance $\hat X$.

\begin{observation}\label{obs: optimal solution cost for semi perm instance}
$\opt(\hat{X}) = \Omega(N^* \log \log N^*)$ 
\end{observation} 
\begin{proof}
	From  Claims~\ref{claim:BRS complexity} and~\ref{claim: shifting does not hurt opt}, for each $0\leq s\leq N-1$, each sub-instance $X^s$ has $\opt(X^s) \geq \Omega(n \log n)=\Omega(\log N \log \log N)$. 
Therefore, $\opt(\hat{X}) \geq \sum_{s=0}^{N-1} \opt(X^s) =\Omega(N\log N\log\log N)= \Omega(N^* \log \log N^*)$ (we have used the fact that $N^*=N\log N$).  
\end{proof}

\subsubsection{Second Instance}

In this step we construct our second and final instance, $X^*$, that is a permutation. 
In order to do so, we start with the instance $\hat{X}$, and, for every active column $C$ of $\hat{X}$, we create $n= \log N$ new columns (that we view as copies of $C$), $C^1,\ldots,C^{\log N}$, which replace the column $C$. 
We denote this set of columns by $\bset(C)$, and we refer it as the \emph{block of columns representing $C$}. 
Recall that the original column $C$ contains $\log N$ input points of $\hat{X}$. 
We place each such input point on a distinct column of $\bset(C)$, so that the points form a monotonically increasing sequence (see the definition in Section~\ref{sec: sequences}). This completes the definition of the final instance $X^*$.
We use the following bound on the optimal solution cost of instance $X^*$.

\begin{claim}
	$\opt(X^*)\geq \opt(\hat{X})=\Omega(N^*\log\log N^*)$.
\end{claim}

\begin{proof}
Let $\cset$ be the set of active columns for $\hat{X}$.  
Notice that we can equivalently define $\hat{X}$ as the instance that is obtained from $X^*$ by collapsing, for every column $C\in \cset$, the columns of $\bset(C)$ into the column $C$.  
Consider an optimal canonical solution $Y$ for $X^*$ and the corresponding satisfied set $Z=Y \cup X^*$. 
From Observation~\ref{obs: collapsing columns}, for a column $C\in \cset$, point set $Z_{\mid \bset(C)}$ remains satisfied. 
We keep applying the collapse operations for all $C \in \cset$ to $Z$, and we denote the final resulting set of points by  $\hat{Z}$. 
Notice that $\hat{Z}$ contains every point in $\hat{X}$ and therefore the solution $\hat{Z} \setminus \hat{X}$ is a feasible solution for $\hat X$. Moreover, the cost of the solution is bounded by $|Z| - |\hat{X}| = |Y|$. 
Therefore, $ \opt(\hat{X})\leq \opt(X^*)$.  
\end{proof}

\subsection{Upper Bound for $\WB(\hat X)$}
The goal of this section is to prove the following theorem.

\begin{theorem}\label{thm: bound for first final instance}
$\WB(\hat X)\leq O(N^*\log\log\log N^*)$.
\end{theorem}

Consider again the instance $\hat X$. Recall that it consists of $N$ instances $X^0,X^1,\ldots,X^{N-1}$ that are stacked on top of each other vertically in this order. 
We rename these instances as $X_1,X_2,\ldots,X_{N}$, so $X_j$ is exactly $\ESBRS(\log N)$, that is shifted by $(j-1)$ units to the right. 
Recall that $|\hat X|=N^*=N\log N$, and each instance $X_s$ contains exactly $\log N$ points.
We denote by $\cset$ the set of $N$ columns, whose $x$-coordinates are $1,2,\ldots,N$. 
All points of $\hat X$ lie on the columns of $\cset$. For convenience, for $1\leq j\leq N$, we denote by $C_j$ the column of $\cset$ whose $x$-coordinate is $j$.

\paragraph{Types of crossings:} 
Let $\sigma$ be any ordering of the auxiliary columns in $\lset$, and let $T=T_{\sigma}$ be the corresponding partitioning tree.
Our goal is to show that, for any such ordering $\sigma$, the value of $\WB_{\sigma}(\hat{X})$ is small (recall that we have also denoted this value by $\WB_T(\hat X)$). 
Recall that $WB_{\sigma}(\hat{X})$ is the sum, over all vertices $v\in V(T)$, of $\cost(v)$. The value of $\cost(v)$ is defined as follows. 
If $v$ is a leaf vertex, then $\cost(v)=0$. Otherwise, let $L=L(v)$ be the line of $\lset$ that $v$ owns.
Index the points in $X \cap S(v)$ by $q_1,\ldots, q_z$ in their bottom-to-top order.
A consecutive pair $(q_j, q_{j+1})$ of points is a \emph{crossing} iff they lie on different sides of $L(v)$.  
We distinguish between the two types of crossings that contribute towards $\cost(v)$. 
We say that the crossing $(q_j, q_{j+1})$ is of {\em type-$1$} if both $q_j$ and $q_{j+1}$ belong to the same shifted instance $X_s$ for some $0\leq s\leq N-1$.
Otherwise, they are of type-$2$. 
Note that, if $(q_j,q_{j+1})$ is a crossing of type $2$, with $q_j\in X_s$ and $q_{j+1}\in X_{s'}$, then $s,s'$ are not necessarily consecutive integers, as it is possible that for some indices $s''$, $X_{s''}$ has no points that lie in the strip $S(v)$.
We now let $\cost_1(v)$ be the total number of type-1 crossings of $L(v)$, and $\cost_2(v)$ the total number of type-2 crossings. Note that $\cost(v)=\cost_1(v)+\cost_2(v)$. 
We also define $\cost_1(\sigma) = \sum_{v \in V(T)} \cost_1(v)$ and $\cost_2(\sigma) = \sum_{v \in V(T)} \cost_2(v)$. 
Clearly, $\WB_{\sigma}(\hat X)=\cost_1(\sigma)+\cost_2(\sigma)$. In the next subsections, we bound $\cost_1(\sigma)$ and $\cost_2(\sigma)$ separately. We will show that each of these costs is bounded by $O(N^*\log\log\log N^*)$.

\subsubsection{Bounding Type-1 Crossings}\label{subsec: type 1 crossings}

The goal of this subsection is to prove the following theorem.

\begin{theorem}\label{thm: crossings inside instances}
	For every ordering $\sigma$ of the auxiliary columns in $\lset$, $\cost_1(\sigma)\leq O(N^*\log\log\log N^*)$.
\end{theorem}

We prove this theorem by a probabilistic argument. 
Consider the following experiment. 
Fix the permutation $\sigma$ of $\lset$.  
Pick an integer $s \in \set{0,\ldots, N-1}$ uniformly at random, and let $X$ be the resulting instance $X_s$. 
This random process generates an input $X$ containing $n = \log N$ points.
Equivalently, let $p_1, p_2, \ldots, p_{\log N}$ be the points in $\BRS(\ell)$ ordered from left to right.  
Once we choose a random shift $s$, we move these points  to columns in $\cset_s = \{2^j + s \mod N \}$, where point $p_j$ would be moved to $x$-coordinate $2^j + s \mod N$. 
Therefore, in the analysis, we view the location of points $p_1,\ldots, p_{\log N}$ as random variables.
  
We denote by $\mu(\sigma)$ the expected value of $\WB_{\sigma}(X)$, over the choices of the shift $s$. The following observation is immediate, and follows from the fact that the final instance $\hat X$ contains every instance $X_s$ for all shifts $s\in \set{0,\ldots,N-1}$.

\begin{observation}
$\cost_1(\sigma) = N \cdot \mu(\sigma)$ 
\end{observation}  

Therefore, in order to prove Theorem \ref{thm: crossings inside instances}, it is sufficient to show that, for every fixed permutation $\sigma$ of $\lset$, $\mu(\sigma)\leq  O(\log N\log\log\log N)$ (recall that $N^*=N\log N$).

We assume from now on that the permutation $\sigma$ (and the corresponding partitioning tree $T$) is fixed, and we analyze the expectation $\mu(\sigma)$.
Let $v \in V(T)$. We say that $S(v)$ is a {\em seam strip} iff point $p_1$ lies in the strip $S(v)$. 
We say that $S(v)$ is a \emph{bad} strip (or that $v$ is a bad node) if the following two conditions hold: (i) $S(v)$ is not a seam strip; and (ii) $S(v)$ contains at least $100\log\log N$ points of $X$. 
Let $\event(v)$ be the bad event that $S(v)$ is a bad strip.

\begin{claim}\label{claim: prob bad strip}
For every vertex $v\in V(T)$, $\prob{\event(v)}\leq \frac{8\width(S(v))}{N\log^{100} N}$.
\end{claim}

\begin{proof}
	Fix a vertex $v\in V(T)$.
	For convenience, we denote $S(v)$ by $S$. Let $s$ be the random integer chosen by the algorithm and let $X_s=X$ be the resulting point set.
	Assume that $S$ is a bad strip, and let $L$ be the vertical line that serves as the left boundary of $S$. Let $p_j$ be the point of $X_s$ that lies to the left of $L$, and among all such points, we take the one closest to $L$. Recall that for each $1\leq j< \log N$, there are $2^j-1$ columns of $\cset$ that lie between the column of $p_j$ and the column of $p_{j+1}$. %
	
	If $S$ is a bad strip, then it must contain points $p_{j+1},p_{j+2},\ldots,p_{j+q}$, where $q= 100\log\log N$. Therefore, the number of columns of $\cset$ in strip $S$ is at least $2^{j+q-2}$, or, equivalently, $\width(S)\geq 2^{j+q}/4\geq (2^j\log^{100}N)/4$. In particular, $2^j\leq 4\width(S)/\log^{100}N$. 
	
	Therefore, in order for $S$ to be a bad strip, the shift $s$ must be chosen in such a way that the point $p_j$, that is the rightmost point of $X_s$ lying to the left of $L$, has $2^j\leq 4\width(S)/\log^{100}N$. It is easy to verify that the total number of all such shifts $s$ is bounded by $8\width(S)/\log^{100}N$.
	
	In order to see this, consider an equivalent experiment, in which we keep the instance $X_1$ fixed, and instead choose a random shift $s\in \set{0,\ldots,N-1}$ for the line $L$. For the bad event $\event(v)$ to happen, the line $L$ must fall in the interval between $x$-coordinate $0$ and $x$-coordinate $8\width(S)/\log^{100}N$. Since every integral shift $s$ is chosen with the same probability $1/N$, the probability that $\event(v)$ happens is at most $\frac{8\width(S)}{N\log^{100} N}$.
\end{proof}

Consider now the partitioning tree $T$. We partition the vertices of $T$ into $\log N+1$ classes $Q_1,\ldots,Q_{\log N+1}$. A vertex $v\in V(T)$ lies in class $Q_i$ iff $2^i\leq \width(S(v))< 2^{i+1}$. Therefore, every vertex of $T$ belongs to exactly one class.

Consider now some vertex $v\in V(T)$, and assume that it lies in class $Q_i$. We say that $v$ is an \emph{important vertex} for class $Q_i$ iff no ancestor of $v$ in the tree $T$ belongs to class $Q_i$. Notice that, if $u$ is an ancestor of $v$, and $u\in Q_j$, then $j\geq i$ must hold.

For each $1\leq i\leq \log N+1$, let $U_i$ be the set of all important vertices of class $Q_i$. 

\begin{observation}
\label{obs: bounded U_i} 
For each $1\leq i\leq \log N+1$, $|U_i| \leq N/2^i$. 
\end{observation}
\begin{proof}
Since no vertex of $U_i$ may be an ancestor of another vertex, the strips in $\set{S(v)\mid v\in U_i}$ are mutually disjoint, except for possibly sharing their boundaries. 
Since each strip has width at least $2^i$, and we have exactly $N$ columns, the number of such strips is bounded by $N/2^i$. 
\end{proof}

Let $\event$ be the bad event that there is some index $1\leq i\leq \log N+1$, and some important vertex $v\in U_i$ of class $Q_i$, for which  the event $\event(v)$ happens. 
Applying the Union Bound to all strips in $\set{S(v)\mid v\in U_i}$ and  all indices $1\leq i\leq \log N$, we obtain the following corollary of Claim \ref{claim: prob bad strip}.

\begin{corollary}\label{cor: no bad important strips}
	 $\prob{\event}\leq \frac{32}{\log^{99}N}$.
\end{corollary}
\begin{proof}
	Fix some index $1\leq i\leq \log N$. Recall that for every important vertex $v\in U_i$, the probability that the event $\event(v)$ happens is at most $\frac{8\width(S(v))}{N\log^{100} N}\leq \frac{8\cdot 2^{i+1}}{N\log^{100}N}=\frac{16\cdot 2^{i}}{N\log^{100}N}$. 
From Observation~\ref{obs: bounded U_i}, $|U_i|\leq N/2^i$. 
From the union bound, the probability that event $\event(v)$ happens for any $v\in U_i$ is bounded by $\frac{16}{\log^{100}N}$. 
Using the union bound over all  $1\leq i\leq \log N+1$, we conclude that  $\prob{\event}\leq \frac{32}{\log^{99}N}$.
\end{proof}

Lastly, we show that, if event $\event$ does not happen, then the cost of the Wilber Bound is sufficiently small.
	
\begin{lemma}\label{lem: no bad event low cost}
	Let $1\leq s\leq N$ be a shift for which $\event$ does not happen. Then:
	
	 $$WB_{\sigma}(X_s)\leq O(\log N\log\log\log N).$$
\end{lemma}

\begin{proof}
	Consider the partitioning tree $T=T(\sigma)$. 
We say that a vertex $v\in V(T)$ is a \emph{seam vertex} iff $S(v)$ is a seam strip, that is, the point $p_1$ in instance $X_s$ lies in $S(v)$.  
Clearly, the root of $T$ is a seam vertex, and for every seam vertex $v$, exactly one of its children is a seam vertex. 
Therefore, there is a root-to-leaf path $P$ that only consists of seam vertices, and every seam vertex lies on $P$. 
We denote the vertices of $P$ by $v_1,v_2,\ldots,v_q$, where $v_1$ is the root of $T$, and $v_q$ is a leaf. 
For $1<i\leq q$, we denote by $v'_i$ the sibling of the vertex $v_i$. 
Note that all strips $S(v'_2),\ldots,S(v'_q)$ are mutually disjoint, except for possibly sharing boundaries, and so $\sum_{i=2}^q\size(v'_i)\leq |X_s|= \log N$. 
Moreover, from Claim \ref{claim: bounds on WB for path and subtree}, $\sum_{i=1}^q\cost(v_i)\leq 2|X_s|=2\log N$.
	
For each $1<i\leq q$, let $T_i$ be the sub-tree of $T$ rooted at the vertex $v'_i$. We prove the following claim:
	
	\begin{claim}\label{claim: bounding subtrees}
		For all $1<i\leq q$, the total cost of all vertices in $T_i$ is at most $O(\size(v'_i)\log\log\log N)$.
	\end{claim}
	
	Assume first that the above claim is correct. Notice that every vertex of $T$ that does not lie on the path $P$ must belong to one of the trees $T_i$. 
The total cost of all vertices lying in all trees $T_i$ is then bounded by $\sum_{i=2}^qO(\size(v'_i)\log\log\log N)\leq O(\log N\log\log\log N)$. 
Since the total cost of all vertices on the path $P$ is bounded by $2\log N$, overall, the total cost of all vertices in $T$ is bounded by $O(\log N\log\log\log N)$.

In order to complete the proof of Lemma \ref{lem: no bad event low cost}, it now remains to prove Claim \ref{claim: bounding subtrees}.

\begin{proofof}{Claim \ref{claim: bounding subtrees}}
We fix some index $1<i\leq q$, and consider the vertex $v'_i$. If the parent $v_{i-1}$ of $v'_i$ belongs to a different class than $v'_i$, then $v'_i$ must be an important vertex in its class. In this case, since we have assumed that Event $\event$ does not happen, $\size(v'_i)\leq O(\log\log N)$. 
From Claim \ref{claim: bounds on WB for path and subtree}, the total cost of all vertices in $T_i$ is bounded by 
$$\sum_{v \in V(T_i)} \cost(v) \leq O(\size(v'_i)\log(\size(v'_i)))\leq O(\size(v'_i)\log\log\log N)$$  
Therefore, we can assume from now on that $v_{i-1}$ and $v'_i$ both belong to the same class, that we denote by $Q_j$.
Notice that, if a vertex $v$ belongs to class $Q_j$, then at most one of its children may belong to class $Q_j$; the other child must belong to some class $Q_{j'}$ for $j'<j$, and it must be an important vertex in its class.

We now construct a path $P_i$ in tree $T_i$ iteratively, as follows. 
The first vertex on the path is $v'_i$. 
We then iteratively add vertices at to the end of path $P_i$ one-by-one, so that every added vertex belongs to class $Q_j$. 
In order to do so, let $v$ be the last vertex on the current path $P_i$. 
If some child $u$ of $v$ also lies in class $Q_i$, then we add $u$ to the end of $P_i$ and continue to the next iteration. 
Otherwise, we terminate the construction of the path $P_i$.

Denote the sequence of vertices on the final path $P_i$ by $(v'_i=u_1,u_2,\ldots,u_z)$; recall that every vertex on $P_i$ belongs to class $Q_j$, and that path $P_i$ is a sub-path of some path connecting $v'_i$ to a leaf of $T_i$. Let $Z$ be a set of vertices containing, for all $1<z'\leq z$ a sibling of the vertex $u_{z'}$, and additionally the two children of $u_z$ (if they exist). Note that every vertex $x\in Z$ is an important vertex in its class, and, since we have assumed that Event $\event$ did not happen, $\size(x)\leq O(\log\log N)$. For every vertex $x\in Z$, we denote by $T'_x$ the sub-tree of $T$ rooted at $x$. From Claim \ref{claim: bounds on WB for path and subtree}, $\cost(T'_x)\leq O(\size(x)\log(\size(x)))=O(\size(x)\log\log\log N)$.

Notice that all strips in $\set{S(x)\mid x\in Z}$ are disjoint from each other, except for possibly sharing a boundary. It is then easy to see that $\sum_{x\in Z}\size(x)\leq \size(v'_i)$. Therefore, altogether $\sum_{x\in Z}\cost(T'_x)\leq O(\size(v'_i)\log\log\log N))$.

Lastly, notice that every vertex of $V(T_i)$ either lies on $P_i$, or belongs to one of the trees $T'_x$ for $x\in Z$. Since, from Claim \ref{claim: bounds on WB for path and subtree}, the total cost of all vertices on $P_i$ is bounded by $\size(v'_i)$, altogether, the total cost of all vertices in $T_i$ is bounded by $O(\size(v'_i)\log\log\log N))$.
\end{proofof}
\end{proof}

To summarize, if the shift $s$ is chosen such that Event $\event$ does not happen, then 	 $WB_{\sigma}(X_s)\leq O(\log N \log\log\log N)$. Assume now that the shift $s$ is chosen such that Event $\event$ happens. From Corollary \ref{cor: no bad important strips}, the probability of this is at most 	 $\prob{\event}\leq \frac{32}{\log^{99}N}$. Since $|X_s|=\log N$, from Corollary \ref{cor: upper bound on WB}, $\WB_{\sigma}(X_s)\leq |X_s|\log(|X_s|)\leq \log N\log\log N$. Therefore, altogether, we get that $\mu(\sigma)\leq O(\log N\log\log\log N)$, and $\cost_1(\sigma)=N\cdot \mu(\sigma)=O(N\log N\log\log\log N)=O(N^*\log\log\log N^*)$, as $N^*=N\log N$.

\subsubsection{Bounding Type-2 Crossings}
In this section we show that for every ordering $\sigma$ of the auxiliary columns in $\lset$, $\cost_2(\sigma)\leq O(N^*\log\log\log N^*)$. For the remainder of the section, we fix an ordering $\sigma$ of $\lset$, and we denote the corresponding partitioning tree $T(\sigma)$ by $T$. Recall that we have defined a set $\cset$ of $N$ columns, so that $|\cset|,|\lset|\leq N$, and  $|V(T)|=O(N)$. It is then sufficient to prove the following theorem:

\begin{theorem}\label{thm: bound cost 2}
	For every vertex $v\in V(T)$, $\cost_2(v)\leq O(\log N)+O(\cost_1(v))$.
\end{theorem}

Assuming that the theorem is correct, we get that $\cost_2(\sigma)\leq O(\cost_1(\sigma))+O(|V(T)|\cdot \log N)=O(N^*\log\log\log N^*)+O(N\log N)=O(N^*\log\log\log N^*)$.

The remainder of this subsection is dedicated to the proof of Theorem \ref{thm: bound cost 2}. We fix a vertex $v\in V(T)$, and we denote $S=S(v)$. We also let $L = L(v)$ be the vertical line that $v$ owns. 
Our goal is to show that the number of type-$2$ crossings of $L$ is bounded by $O(\cost_1(v))+O(\log N)$.

Recall that instances $X_1,\ldots,X_N$ are stacked on top of each other, so that the first $\log N$ rows with integral coordinates  belong to $X_1$, the next $\log N$ rows belong to $X_2$, and so on. If we have a crossing $(p,p')$, where $p\in X_s$ and $p'\in X_{s'}$, then we say that the instances $X_s$ and $X_{s'}$ are \emph{responsible} for this crossings. Recall that $p,p'$ may only define a crossing if they lie on opposite sides of the line $L$, and if no point of $\hat X$ lies in the strip $S$ between the row of $p$ and the row of $p'$. It is then clear that every instance $X_s$ may be responsible for at most two type-2 crossings of $L$: one in which the second instance $X_{s'}$ responsible for the crossing has $s'<s$, and one in which $s'>s$.

We further partition the type-2 crossings into two sub-types. Consider a crossing $(p,p')$, and let $X_s,X_{s'}$ be the two instances that are responsible for it. If either of $X_s,X_{s'}$ contributes a type-1 crossing to the cost of $L$, then we say that $(p,p')$ is a type-(2a) crossing; otherwise it is a type-(2b) crossing. Clearly, the total number of type-(2a) crossings is bounded by $O(\cost_1(v))$. It is now sufficient to show that the total number of all type-(2b) crossings is bounded by $O(\log N)$.

Consider now some type-(2b) crossing $(p,p')$, and let $X_s$ and $X_{s'}$ be the two instances that are responsible for it, with $p\in X_s$. We assume that $s<s'$. Since neither instance contributes a crossing to $\cost_1(v)$, it must be the case that all points of $X_s\cap S$ lie to the left of $L$ and all points of $X_{s'}\cap S$ lie to the right of $L$ or vice versa.
Moreover, if $s'>s+1$, then for all $s<s''<s'$, $X_{s''}\cap S=\emptyset$.

It would be convenient for us to collapse each of the instances $X_1,\ldots,X_{N}$ into a single row. In order to do so, for each $1\leq s\leq N$, we replace all rows on which the points of $X_s$ lie with a single row $R_s$. If some point of $X_s$ lies on some column $C$, then we add a point at the intersection of $R_s$ and $C$.

We say that a row $R_s$ is \emph{empty} if there are no input points in $R_s\cap S$. We say that it is a \emph{neutral} row, if there are points in $R_s\cap S$ both to the left of $L$ and to the right of $L$. We say that it is a \emph{left row} if $R_s\cap S$ only contains points lying to the left of $L$, and we say that it is a {\em right row} if $R_s\cap S$ only contains points lying to the right of $L$.

If we now consider any type-(2b) crossing $(p,p')$, and the instances $X_s,X_{s'}$ that are responsible for it, with $s<s'$, then it must be the case that exactly one of the rows $R_s,R_{s'}$ is a left row, and the other is a right row. Moreover, if $s'>s+1$, then every row lying between $R_s$ and $R_{s'}$ is an empty row.

Let us denote the points in $X_1$ by $p_1,\ldots,p_{\log N}$, where for each $1\leq i\leq \log N$, point $p_i$ lies in column $C_{2^i}$. 
In each subsequent instance $X_2,X_3,\ldots$, the point is shifted by one unit to the right, so that in instance $X_s$ it lies in column $C_{2^i+s-1}$; every column in $\cset$ must contain exactly one copy of point $p_i$.

Consider now all copies of the point $p_i$ that lie in the strip $S$. Let $\rset_i$ be the set of rows containing these copies. Then two cases are possible: either (i) $\rset_i$ is a contiguous set of rows, and the copies of $p_i$ appear on $\rset_i$ diagonally as an increasing sequence (the $j$th row of $\rset_i$ contains a copy of $p_i$ that lies in the $j$th column of $\cset$ in the strip $S$); or $\rset_i$ consists of two consecutive sets of rows; the first set, that we denote by $\rset'_i$, contains $R_1$, and the second set, that we denote by $\rset''_i$, contains the last row $R_{N}$. The copies of the point $p_i$ also appear diagonally in $\rset'_i$ and in $\rset''_i$; in $\rset''_i$ the first copy lies on the first column of $\cset$ in $S$; in $\rset'_i$ the last copy lies on the last column of $\cset$ in $S$ (see Figure \ref{fig: strip-crossing}).

\begin{figure}[h]
	\centering
	\begin{subfigure}{.4\textwidth}
		\centering
		\includegraphics[width=.6\linewidth]{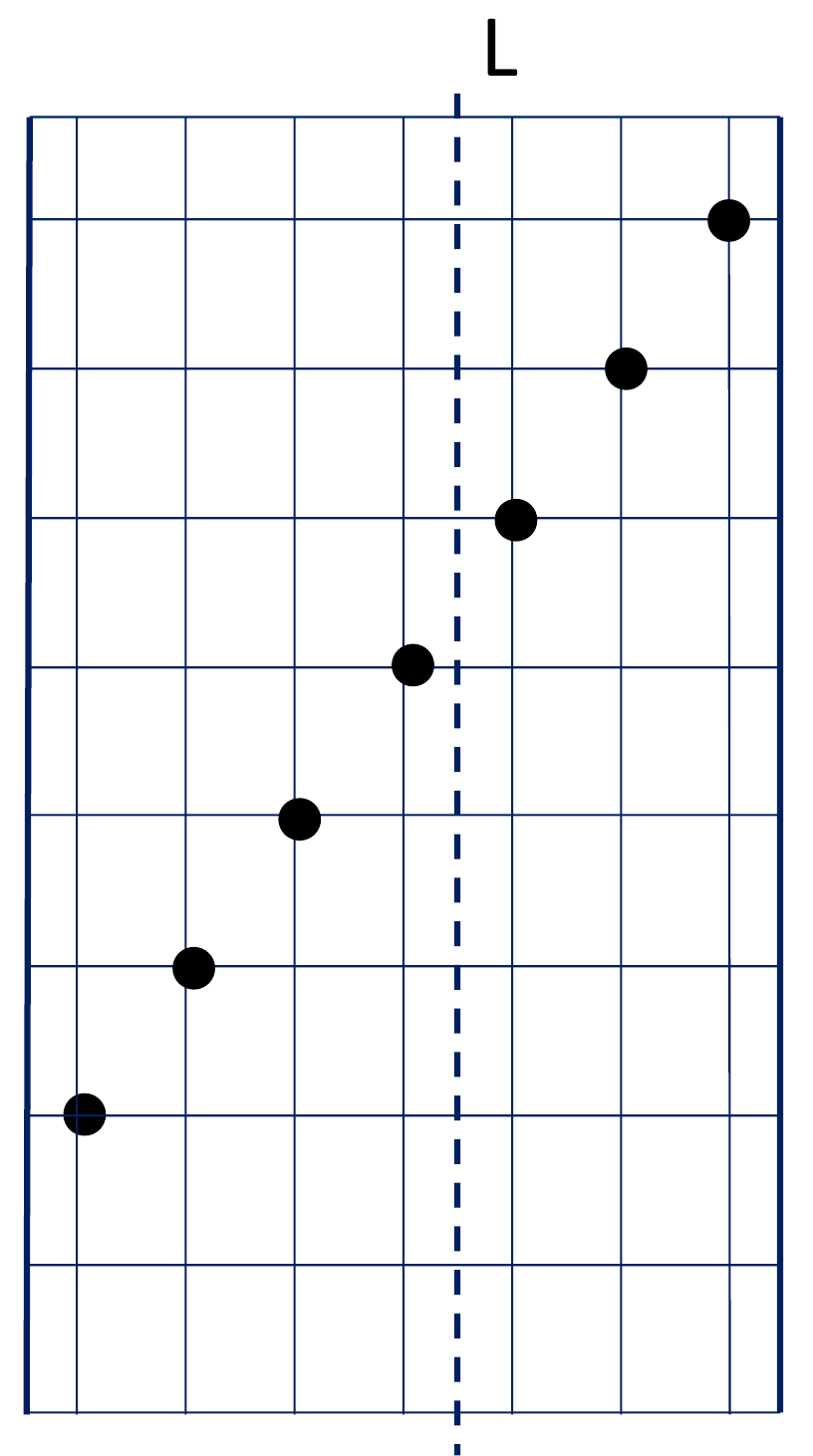}
		\caption{The consecutive set of rows on which the copies of $p_i$ appear is denoted by $\rset_i$.}
	\end{subfigure}%
	\hspace{.1\textwidth}
	\begin{subfigure}{.4\textwidth}
		\centering
		\includegraphics[width=.6\linewidth]{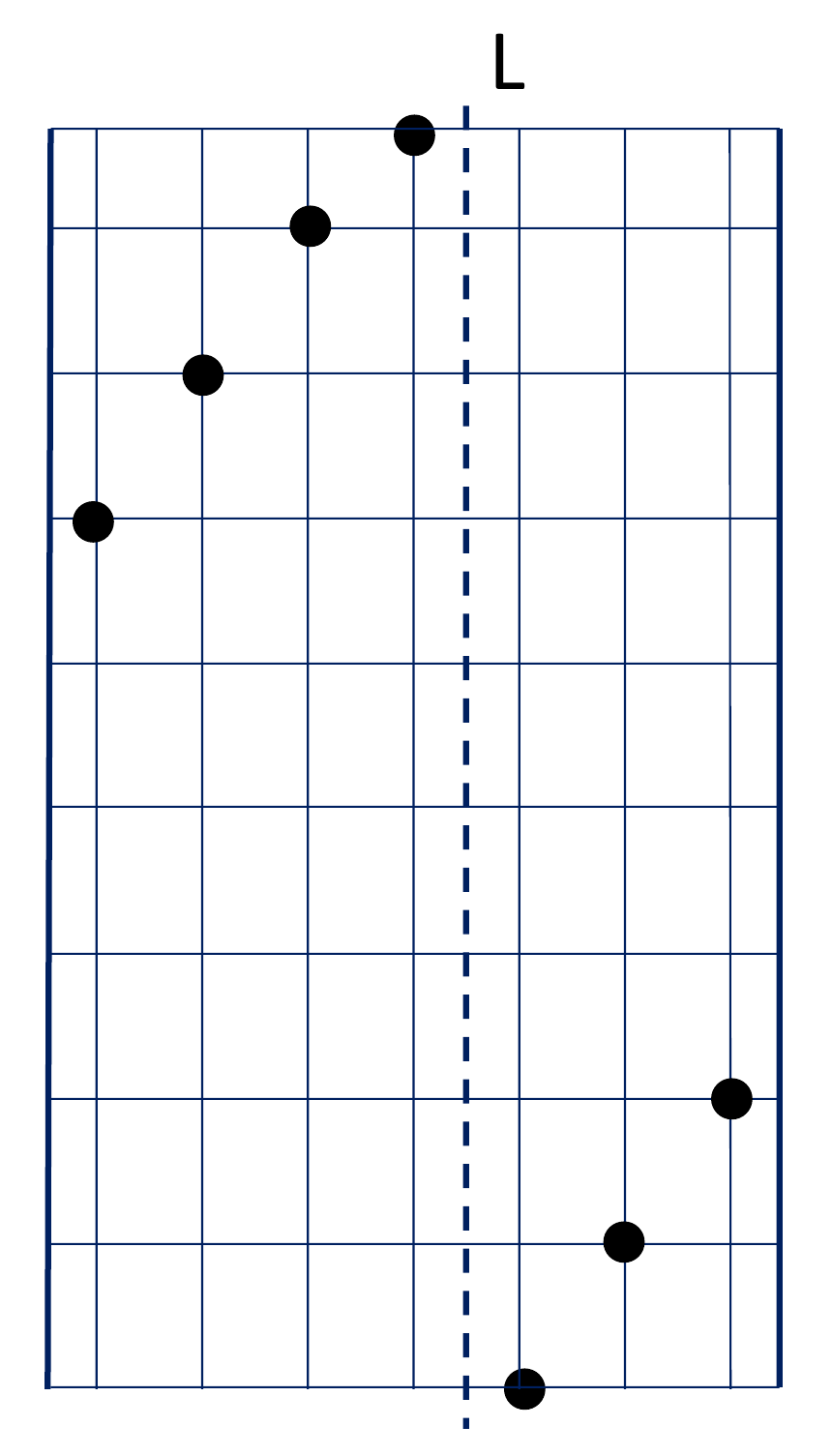}
		\caption{The two consecutive sets of rows on which the copies of $p_i$ appear are $\rset'_i$ (on the bottom) and $\rset''_i$ (on the top).}
	\end{subfigure}
	\caption{Two patterns in which copies of $p_i$ may appear on strip $S$. \label{fig: strip-crossing}}
\end{figure}

We show that for each $1\leq i\leq \log N$, there are at most four type-(2b) crossings of the line $L$ in which a copy of $p_i$ participates. Indeed, consider any type-(2b) crossing $(p,p')$ in which a copy of $p_i$ participates. We assume that the row of $p$ lies below the row of $p'$.
Assume first that both $p$ and $p'$ lie on rows of $\rset_i$, and let $R,R'$ be these two rows, with $p\in R$, $p'\in R'$. Recall that, in order for $(p,p')$ to define a crossing, all input points that lie on $R\cap S$ must lie to the left of $L$, and all point points that lie on $R'\cap S$ must lie to the right of $L$, or the other way around. It is easy to verify (see Figure \ref{fig: strip-crossing}) that one of two cases must happen: either $R$ contains a copy of $p_i$ lying closest to $L$  on its left, and $R'$ contains a copy of $p_i$ lying closest to  $L$ on its right; or $R$ is the last row of $\rset'_i$, and  $R'$ is the first row of $\rset''_i$. Therefore, only two such crossing, with $R,R'\in \rset_i$ are possible.

Assume now that $R\in \rset_i$ and $R'\not\in \rset_i$; recall that we assume that $R'$ lies above $R$. Then all rows that lie between $R$ and $R'$ must be empty, so it is easy to verify that $R$ must be the last row of $\rset_i$ (or it must be the last row of $\rset'_i$). In either case, at most one such crossing is possible.

Lastly, we assume that $R\not\in \rset_i$ and $R'\in \rset_i$. The analysis is symmetric; it is easy to see that at most one such crossing is possible.

We conclude that for each $1\leq i\leq \log N$, at most four type-(2b) crossings of the line $L$ may involve copies of $p_i$, and so the total number of type-(2b) crossings of $L$ is bounded by $O(\log N)$.

To summarize, we have shown that for every ordering $\sigma$ of the auxiliary columns in $\lset$, $\cost_1(\sigma),\cost_2(\sigma)\leq O(N^*\log\log\log N^*)$, and so $\WB_{\sigma}(\hat X)=O(N^*\log\log\log N^*)$. Since $\opt(\hat X)=\Omega(N^*\log\log N^*)$, we obtain a gap of $\Omega(\log\log N^*/\log\log\log N^*)$ between $\opt(\hat X)$ and $\WB(\hat X)$.

\subsection{Upper Bound for $\WB(X^*)$}

In this section we show that $\WB(X^*)=O(N^*\log\log\log N^*)$. 
Recall that instance $X^*$ is obtained from instance $\hat X$ by replacing every active column $C$ of $X^*$ with a block $\bset(C)$ of columns, and then placing the points of $C$ on the columns of $\bset(C)$ so that they form a monotone increasing sequence, while preserving their $y$-coordinates. 
The resulting collection of all blocks $\bset(C)$ partitions the set of all active columns of $X^*$. We denote this set of blocks by $\bset_1,\ldots,\bset_N$. 
The idea is to use Theorem \ref{thm: WB for corner and strip instances} in order to bound $\WB(X^*)$.

Consider a set of lines $\lset'$ (with half-integral $x$-coordinates) that partition the bounding box $B$ into strips, where the $i$th strip contains the block $\bset_i$ of columns, so $|\lset'| = (N-1)$. 
We consider a split of instance $X^*$ by $\lset'$: This gives us  a collection of strip instances $\set{X^s_i}_{1\leq i\leq N}$ and the compressed instance $X^c$. 
Notice that the compressed instance is precisely $\hat{X}$, and each strip instance $X^s_i$ is a monotone increasing point set.

Since each strip instance $X^s_i$ is monotonously increasing, from Observation \ref{obs: solution for monotone increasing set} and Claim \ref{claim: bounding WB by OPT}, for all $i$, $\WB(X^s_i)\leq O(\opt(X^s_i))\leq O(|X^s_i|)$. From Theorem \ref{thm: WB for corner and strip instances}, we then get that:

\[\WB(X^*)\leq 4\WB(\hat X)+8\sum_i\WB(X^s_i)+O(|X^*|)\leq 4\WB(\hat X)+O(|X^*|)\leq O(N^*\log\log\log N^*).
\]

\subsection{Separating $\WBtwo$ and $\WB$}
\label{subsec: separating two wilbers}

In this section, we extend our $\Omega(\frac{\log\log n}{\log\log\log n})$-factor
separation between $\WB$ and $\opt$ to a separation between $\WB$
and the \emph{second Wilber bound} (denoted by $\WBtwo$), which is defined below.\footnote{Wilber originally defined this bound based on the tree view. We use an equivalent geometric definition as discussed in \cite{DHIKP09,in_pursuit}.}

Let $X$ be a set of points that is a semi-permutation, with $|X|=m$. Consider any point $p\in X$. The funnel of $p$, denoted by $\funnel(X,p)$ is the set of all points $q\in X$, such that $q,y<p.y$, and $\square_{p,q}$ contains no point of $X\setminus\set{p,q}$. 
Denote $\funnel(X,p)=\set{a_1,\ldots,a_r}$, where the points are indexed in the increasing order of their $y$-coordinates. Let $\alt(X,p)$ be the number of indices $1\leq i<r$, such that $a_i$ lies strictly to the left of $p$ and $a_{i+1}$ lies strictly to the right of $p$, or the other way around. The
\emph{second Wilber bound } is:
\[
\WBtwo(X)=m+\sum_{p\in X}\alt(X,p).
\]

The goal of this section is to prove the following:
\begin{theorem}
\label{thm:sep WBtwo}For every integer $n'$, there is an integer
$n\ge n'$ and a point set $X$ that is a permutation with
$|X|=n$, such that $\WBtwo(X)\ge\Omega(n\log\log n)$ but $\WB(X)\le O(n\log\log\log n)$.
\end{theorem}

As it is known that $\opt(X)\ge\WBtwo(X)$ for any point set $X$ \cite{wilber},
\Cref{thm:sep WBtwo} is a stronger statement than \Cref{thm:intro_WB}.
To prove \Cref{thm:sep WBtwo}, we use exactly the same permutation
sequence $\Xstar$ of size $N^{*}$ that is constructed in \Cref{sec:construction}. Since we already showed that  $\WB(\Xstar)\le O(N^{*}\log\log\log N^{*})$, 
it remains to show that $\WBtwo(\Xstar)\ge\Omega(N^{*}\log\log N^{*})$.

We use the following claim of Wilber~\cite{wilber}. %
bitwise reversal sequences:
\begin{claim}
[\cite{wilber}]$\WBtwo(\BRS(n))=\Omega(n\log n)$ for any $n$.
\end{claim}

We extend this bound to the cyclically shifted BRS in the following lemma.
\begin{lemma}
\label{lem:shifted BRS} For integers $n>0,s$ with  $0\leq s<n$, let $X$ be the sequence obtained by performing
a cyclic shift to $\BRS(n)$ by $s$ units. Then $\WBtwo(X)=\Omega(n\log n)$. 
\end{lemma}

\begin{proof}
Observe that, for any choice of $s$, there must exists
a subsequence $X'$ of $X$ such that $X'$ is a copy of $\BRS(n-1)$ (
see \Cref{fig: brs}). It is shown in Lemma 6.2 of \cite{LevyT19}
that for any pair of sequences $Z,Z'$ with $Z'\subseteq Z$, $\WBtwo(Z')\le\WBtwo(Z)$ holds. Therefore, we conclude that
\[
\WBtwo(Y)\ge\WBtwo(\BRS(n-1))\ge \Omega(n\log n).
\]
\end{proof}

Now, we are ready to bound  $\WBtwo(\Xstar)$.

\begin{lemma}
$\WBtwo(\Xstar)=\Omega(N^{*}\log\log N^{*})$.
\end{lemma}

\begin{proof}
Recall that $\Xhat$ is the union of the sets $X^{0},X^{1},\dots,X^{N-1}$ of points, where for all $0\leq s\leq N-1$, set $X^{s}$ is an exponentially-spaced $\BRS$ instance that is shifted
by $s$ units. From the definition
of $\WBtwo$, it is easy to see that $\WBtwo(\Xhat)\ge\sum_{s=0}^{N-1}\WBtwo(X^{s})$. This is since, for all $0\leq s\leq N-1$, for ever point $p\in X^{s}$, $\funnel(X^s,p)\subseteq \funnel(\Xhat, p)$, and moreover $\alt(\Xhat,p)\geq \alt(X^s,p)$.
From \Cref{lem:shifted BRS}, we get that $\WBtwo(X^{s})=\Omega(n\log n)$, where $n=|X^s|=\log N$. Therefore,
\[
\WBtwo(\Xhat)\ge N\cdot\Omega(n\log n)=\Omega(N^{*}\log\log N^{*}).
\]
Finally, recall that the sequence $X^{*}$ is obtained from $\Xhat$
by replacing each column $C$ $\Xhat$ with a block $\bset(C)$ of columns, and placing all points of $\Xhat$ lying in $C$ on the columns of $\bset(C)$ so that they form a monotonically increasing sequence of length $n$. It is not hard to see that $\WBtwo(\Xstar)\ge\WBtwo(\Xhat)$
which concludes the proof.
\end{proof}

\section{Extensions of Wilber Bound}\label{sec: extension to extended}

In this section we define two extensions of the strong Wilber bound and extend our negative results to  one of these bounds.
In subsection \ref{subsec: defs}, we provide formal definitions of these bounds, and we present our negative result in subsequent subsections.  
 
\subsection{Definitions} 
\label{subsec: defs}

Assume that we are given an input set $X$ of $n$ points, that is a permutation. Let $\lset^V$ be the set of all vertical lines with half-integral $x$-coordinates between $1/2$ and $n-1/2$, and let $\lset^H$ be the set of all horizontal lines with half-integral $y$-coordinates between $1/2$ and $n-1/2$. Recall that for every permutation $\sigma$ of $\lset^V$, we have defined a bound $\WB_{\sigma}(X)$. We can similarly define a bound $\WB'_{\sigma'}(X)$ for every permutation $\sigma'$ of $\lset^H$. We also let $\WB'(X)$ be the maximum, over all permutations $\sigma'$ of $\lset^H$, of $\WB'_{\sigma'}(X)$. Equivalently, let $X'$ be an instance obtained from $X$ by rotating it by $90$ degrees clockwise. Then $\WB'(X)=\WB(X')$. We denote by $B$ a bounding box that contains all points of $X$.

\subsubsection{First Extension -- Consistent Guillotine Bound}

In this section we define the {\em consistent Guillotine Bound}, $\XWB(X)$. Let $\sigma$ be any permutations of all lines in $\lset^V\cup \lset^H$. 
We start from a bounding box $B$ containing all points of $X$ and maintain a partition $\pset$ of the plane into rectangular regions, where initially $\pset = \set{B}$.  
We process the lines in $\lset^V \cup \lset^H$ according to their ordering in $\sigma$. 
Consider an iteration when a line $L$ is processed. Let $P_1,\ldots, P_k$ be all rectangular regions in $\pset$ that intersect the line $L$. For each such region $P_j$, let $P'_j$ and $P''_j$ be the two rectangular regions into which the line $L$ splits $P_j$. 
We update $\pset$ by replacing each region $P_j$, for $1\leq j\leq k$, with the regions $P'_j$ and $P''_j$. 
Once all lines in $\lset^V\cup \lset^H$ are processed, we terminate the process.

This recursive partitioning procedure can be naturally associated with a partitioning tree  $T=T_{\sigma}$ that is defined as follows: 

\begin{itemize}
\item Each vertex $v\in V(T)$ is associated with a rectangular region $S(v)$ of the plane. If $r$ is the root of $T$, then $S(r) =B$. 

\item Each non-leaf vertex $v$  is associated with a line $L(v)\in \lset^H\cup \lset^V$ that was used to partition $S(v)$ into two sub-regions, $S'$ and $S''$. Vertex $v$ has two children $v_1,v_2$ in $T$, with $S(v_1)=S'$ and $S(v_2)=S''$.

\item For each leaf node $v$, the region $S(v)$ contains at most one point of $X$.
\end{itemize}

We now define the cost $\cost(v)$ of each node $v\in V(T)$. If the region $S(v)$ contains no points of $X$, or it contains a single point of $X$, then $\cost(v)=0$. Otherwise, we define $\cost(v)$ in the same manner as before. Assume first that the line $L(v)$ is vertical. Let $p_1, \ldots, p_k$  be all points in $X\cap S(v)$, indexed in the increasing order of their $y$-coordinates. A pair $(p_j, p_{j+1})$ of consecutive  points forms a crossing of $L(v)$ for $S(v)$, if they lie on the opposite sides of $L(v)$. We then let $\cost(v)$ be the number of such crossings. 

When $L(v)$ is a horizontal line, $\cost(v)$ is defined analogously: we index the points of $X\cap S(v)$ in the increasing order of their $x$-coordinates. We then say that a consecutive pair of such points is a crossing iff they lie on opposite sides of $L(v)$. We let $\cost(v)$ be the number of such crossings.

For a fixed ordering $\sigma$ of the lines in $\lset^V\cup \lset^H$, and the corresponding partition tree $T= T(\sigma)$, we define $\XWB_{\sigma}(X) = \sum_{v \in V(T)} \cost(v)$. 

Lastly, we define the consistent Guillotine Bound for a point set $X$ that is a permutation to be the maximum, over all orderings $\sigma$ of the lines in $\lset^V\cup \lset^H$, of $\XWB_{\sigma}(X)$.

In the following subsection we define an even stronger bound, that we call the \emph{Guillotine bound}, and we show that for every point set $X$ that is a permutation, $\XWB(X)\leq \XXWB(X)$, and moreover that $\XXWB(X)\leq O(\opt(X))$. It then follows that for every point set $X$ that is a permutation, $\XWB(X)\leq O(\opt(X))$.

In subsequent subsections, we prove the following negative result for $\XWB$, which extends the results of Theorem \ref{thm:intro_WB}.

\begin{theorem}
	\label{thm:negative for XWB}
	For every integer $n'$, there is an integer $n\geq n'$, and a set $X$ of points that is a permutation with $|X|=n$, such that $\opt(X)\geq \Omega(n\log\log n)$ but $\XWB(X)\leq O(n\log\log\log n)$.
\end{theorem}

The following lemma will be helpful in the proof of Theorem \ref{thm:negative for XWB}; recall that $\WB'(X)$ is the basic Wilber Bound, where we cut via horizontal lines only.

\begin{lemma}\label{lem: extended WB from vertical and horizontal}
	For every instance $X$ that is a permutation, $\XWB(X)\leq \WB(X)+\WB'(X)$.
\end{lemma}

\begin{proof}
	Let $\sigma$ be a permutation of $\lset^V\cup \lset^H$, such that $\XWB(X)=\XWB_{\sigma}(X)$. Notice that $\sigma$ naturally induces a permutation $\sigma'$ of $\lset^V$ and a permutation $\sigma''$ of $\lset^H$. We show that $\XWB_{\sigma}(X)\leq \WB_{\sigma'}(X)+\WB'_{\sigma''}(X)$. In order to do so, it is enough to show that for every vertical line $L\in \lset^V$, the cost that is charged to $L$ in the bound $\XWB_{\sigma}(X)$ is less than or equal to the cost that is charged to $L$ in the bound $\WB_{\sigma'}(X)$, and similarly, for every horizontal line $L'\in \lset^H$, the cost that is charged to $L'$ in the bound $\XWB_{\sigma}(X)$ is less than or equal to the cost that is charged to $L'$ in the bound $\WB'_{\sigma''}(X)$. We show the former; the proof of the latter is similar.
	
	Consider any line $L\in \lset^V$. We let $T$ be the partitioning tree associated with $\XWB_{\sigma}(X)$, just before line $L$ is processed, and we let $T'$ be defined similarly for $\WB_{\sigma'}(X)$. Let $v\in V(T')$ be the leaf vertex with $L\subseteq S(v)$, and let $U$ be the set of all leaf vertices $u$ of the tree $T$ with $S(u)\cap L\neq \emptyset$. Observe that the set of vertical lines that appear before $L$ in $\sigma$ and $\sigma'$ is identical. Therefore, $S(v)=\bigcup_{u\in U}S(u)$. It is easy to verify that, for every vertex $u\in U$, every crossing that contributes to $\cost(u)$ is also a crossing that is charged to the line $L$ in the strip $S(v)$. Therefore, the total number of crossings of line $L$ in tree $T'$ that contribute to $\WB_{\sigma'}(X)$ is greater than or equal to the number of crossings of the line $L$ that contribute to $\XWB_{\sigma}(X)$.
	
	To conclude, we get that $\XWB(X)=\XWB_{\sigma}(X)\leq \WB_{\sigma'}(X)+\WB'_{\sigma''}(X)\leq \WB(X)+\WB'(X)$.
\end{proof}

\subsubsection{Guillotine Bound}

In this section, we define a second extension of Wilber Bound, that we call Guillotine Bound, and denote by $\XXWB$. The bound is more convenient to define using a partitioning tree instead of a sequence of lines. 
Let $X$ be a point set which is a permutation.  

We define a guillotine partition of a point set $X$, together with the corresponding partitioning tree $T$. 
As before, every node $v\in V(T)$ of the partitioning tree $T$ is associated with a rectangular region $S(v)$ of the plane. At the beginning, we add the root vertex $r$ to the tree $T$, and we let $S(r)=B$, where $B$ is the bounding box containing all points of $X$. We then iterate, as long as some leaf vertex $v$ of $T$ has $S(v)\cap X$ containing more than one point. In each iteration, we select any such leaf vertex $v$, and we select an arbitrary vertical or horizontal line $L(v)$, that is contained in $S(v)$, and partitions $S(v)$ into two rectangular regions, that we denote by $S'$ and $S''$, such that $X\cap S',X\cap S''\neq \emptyset$. We then add two child vertices $v_1,v_2$ to $v$, and set $S(v_1)=S'$ and $S(v_2)=S''$. Once every leaf vertex $v$ has $|S(v)\cap X|=1$, we terminate the process and obtain the final partitioning tree $T$.

The cost $\cost(v)$ of every vertex $v\in V(T)$ is calculated exactly as before. We then let $\GB_T(X)=\sum_{v\in V(T)}\cost(v)$, and we let $\GB(X)$ be the maximum, over all partitioning trees $T$, of $\GB_T(X)$. 

We note that the main difference between $\cGB(X)$ and $\GB(X)$ is that in $\cGB$ bound, the partitioning lines must be chosen consistently across all regions: that is, we choose a vertical or a horizontal line $L$ that crosses the entire bounding box $B$, and then we partition every region that intersects $L$ by this line $L$. In contrast, in the $\GB$ bound, we can partition each region $S(v)$ individually, and choose different partitioning lines for different regions. It is then easy to see that $\GB$ is more general than $\cGB$, and, in particular, for every point set $X$ that is a permutation, $\cGB(X)\leq \GB(X)$.

Lastly, we show that $\GB$ is a lower bound on the optimal solution cost, in the following lemma whose proof is deferred to Section \ref{sec: bound for GB} of Appendix.

\begin{lemma}\label{lem: GB is lb for opt}
For any point set $X$ that is a permutation, $\GB(X) \leq 2 \opt(X)$.  
\end{lemma}

\subsection{Negative Result for the $\XXWB$ Bound}\label{subsec: Instance construction for extended}
In this section we prove Theorem \ref{thm:negative for XWB}.

\newcommand{\EESBRS}{\mbox{\sf 2D-ES-BRS}}
We use three main parameters. Let $\ell\geq 1$ be an integer, and let $n = 2^{\ell}$ and $N= 2^n$.
As before, we will first construct point set $\hat{X}$ that is not a permutation (in fact, it is not even a semi-permutation), and then we will turn it into our final instance $X^*$ which is a permutation.  

\paragraph{2D exponentially spaced bit reversal:} We define the instance $\EESBRS(\ell)$ to be a bit-reversal sequence $\BRS(\ell,\rset,\cset)$, where the sets $\rset$ and $\cset$ of active rows and columns are defined as follows. Set $\cset$ contains all columns with $x$-coordinates in $\set{2^j\mid 1\leq j\leq n}$, and similarly set $\rset$ contains all rows with $y$-coordinates in $\set{2^j\mid 1\leq j\leq n}$. Note that set $X$ contains $n$ points, whose $x$- and $y$-coordinates are integers between $1$ and $N$.

\paragraph{2D cyclic shifts:} Next, we define the shifted and exponentially spaced instance, but this time we shift both vertically and horizontally. 
We assume that we are given a horizontal shift $0\leq s<N$ and a vertical shift $0\leq s'<N$. 
In order to construct the instance $X^{s,s'}$, we start with the instance $X=\EESBRS(\ell)$, and then perform the following two operations. 
First, we perform a horizontal shift by $s$ units as before, by moving the last $s$ columns with integral $x$-coordinates  to the beginning of the instance. 
Next, we perform a vertical shift, by moving the last $s'$ rows with integral $y$-coordinates to the bottom of the instance. 
We let $X^{s,s'}$ denote the resulting instance. 
By applying Claim \ref{claim: shifting does not hurt opt} twice, once for the horizontal shift, and once for the vertical shift, we get that $\opt(X^{s,s'})\geq \opt(X)-2|X|\geq \Omega(\log N\log\log N)$, since $|X|=\log N$.

\paragraph{Instance $\hat{X}$:} Next, we construct an instance $\hat X$, by combining the instances $X^{s,s'}$ for $0\leq s,s'<N$. 
In order to do so, let $\hat \cset$ be a set of $N^2$ columns, with integral $x$-coordinates $1,\ldots,N^2$. We partition $\hat \cset$ into subsets $\cset_1, \cset_2,\ldots,\cset_N$, each of which contains $N$ consecutive columns, they appear in this left-to-right order. 
We call each such set $\cset_i$ a \emph{super-column}. We denote by $S^V_i$ the smallest vertical strip containing all columns of $\cset_i$.

Similarly, we let $\hat \rset$ be a set of $N^2$ rows, with integral $y$-coordinates $1,\ldots,N^2$. We partition $\hat \rset$ into subsets $\rset_1,\ldots,\rset_N$, each of which contains $N$ consecutive rows, such that $\rset_1,\ldots,\rset_N$ appear in this bottom-to-top order. We call each such set $\rset_i$ a \emph{super-row}.  We denote by $S^H_i$ the smallest horizontal strip containing all rows of $\rset_i$.
For all $1\leq i,j\leq N$, we let $B(i,j)$ be the intersection of the horizontal strip $S^H_i$ and the vertical strip $S^V_i$. We plant the instance $X^{(i-1),(j-1)}$ into the box $B(i,j)$. 
This completes the construction of instance $\hat{X}$. 
Let $N^* = |\hat{X}| = N^2 \log N$ (recall that  each instance $X^{s,s'}$ contains $\log N$ points.)

Observe that, for each vertical strip $S^V_i$, all instances planted into $S^V_i$ have the same vertical shift - $(i-1)$; the horizontal shift $s'$ of each instance increases from $0$ to $N-1$ as we traverse $S^V_i$ from bottom to top. 
In particular, the instance planted into $S^V_1$ is precisely the instance $\hat X$ from Section \ref{sec: Instance construction} (if we ignore inactive rows).  
For each $i>1$, the instance planted into $S^V_i$ is very similar to the instance  $\hat X$ from Section \ref{sec: Instance construction}, except that each of its corresponding sub-instances is shifted vertically by exactly $(i-1)$ rows.

Similarly, for each horizontal strip  $S^H_j$, all instances planted into $S^H_j$ have the same horizontal shift - $(j-1)$; the vertical shift $s'$ of each instance increases from $0$ to $N-1$ as we traverse $S^H_j$ from left to right.

Since, for every  instance $X^{s,s'}$, $\opt(X^{s,s'})=\Omega(\log N\log\log N)$, we obtain the following bound.
\begin{observation} 
 $\opt(\hat X)=\Omega(N^2\log N\log\log N)=\Omega(N^*\log\log N^*)$. 
\end{observation} 

Since instance $\hat X$ is symmetric, and every point lies on one of the $N^2$ rows of $\hat\rset$ and on one of the $N^2$ rows of $\hat\cset$, we obtain the following.

\begin{observation} 
Every row in $\hat\rset$ contains exactly $\log N$ points of $\hat X$. Similarly, every column of $\hat\cset$  contains exactly $\log N$ points of $\hat X$.
\end{observation}

\paragraph{Final instance:} 
Lastly, in order to turn $\hat X$ into a permutation $X^*$, we perform a similar transformation to that in Section \ref{sec: Instance construction}: for every column $C\in \cset$, we replace $C$ with a collection $\bset(C)$ of $\log N$ consecutive columns, and we place all points that lie on $C$ on the columns of $\bset(C)$, so that they form an increasing sequence, while preserving their $y$-coordinates. We replace every row $R\in \rset$ by a collection $\bset(R)$ of $\log N$ rows similarly. The resulting final instance $X^*$ is now guaranteed to be a permutation, and it contains $N^*=N^2\log N$ points. Using the same reasoning as in Section \ref{sec: Instance construction}, it is easy to verify that $\opt(X^*)\geq \opt(\hX)\geq \Omega(N^*\log\log N^*)$. In the remainder of this section, we will show that $\XWB(X^*)=O(N^*\log\log \log N^*)$. 

Abusing the notation, for all $1\leq i\leq N^2$, we denote by $S^V_i$ the vertical strip obtained by taking the union of all blocks $\bset(C)$ of columns, where $C$ belonged to the original strip $S^V_i$. We define the horizontal strips $S^H_i$ similarly. Note that, from Lemma \ref{lem: extended WB from vertical and horizontal}, it is enough to prove that $\WB(X^*)=O(N^*\log\log \log N^*)$ and that  $\WB'(X^*)=O(N^*\log\log \log N^*)$. We do so in the following two subsections.

\subsection{Handling Vertical Cuts}\label{subsec: bounding WB vertical}
The goal of this section is to prove the following theorem:

\begin{theorem}\label{thm: boundinb WB vertical}
	$\WB(X^*)\leq O(N^*\log\log\log N^*)$.
\end{theorem}

For all $1\leq i\leq N$, we denote by $\bset_i$ the set of active columns that lie in the vertical strip $S^V_i$, so that $\bset_1,\ldots,\bset_N$ partition the set of active columns of $X^*$.
Let $\lset'$ be a collection of lines at half-integral coordinates that partitions the bounding box $B$ into $N$ strips where each strip contains exactly the block $\bset_i$ of columns. 
We consider the split of $X^*$ by the lines $\lset'$: This is a collection of $N$ strip instances (that we will denote by $X^*_1,\ldots,X^*_N$) and a compressed instance, that we denote by $\tilde X$.     
In order to prove Theorem \ref{thm: boundinb WB vertical}, we bound $\WB(X^*_i)$ for every strip instance $X^*_i$, and $\WB(\tilde X)$ for the compressed instance $\tilde X$, and then combine them using Theorem \ref{thm: WB for corner and strip instances} in order to obtain the final bound on $\WB(X^*)$.

\subsubsection{Bounding Wilber Bound for Strip Instances}
In this subsection, we prove the following lemma.

\begin{lemma}\label{lem: bounding vertical WB for strip}
	For all $1\leq i\leq N$, $\WB(X^*_i)\leq O(N\log N\log\log\log N)$.
\end{lemma}

From now on we fix an index $i$, and consider the instance $X^*_i$. 
Recall that in order to construct instance $X^*_i$, we started with the instances $X^{0,i},X^{1,i},\ldots,X^{N-1,i}$, each of which has the same vertical shift (shift $i$), and horizontal shifts ranging from $0$ to $N-1$. Let $\hat X_i$ be the instance obtained by stacking these instances one on top of the other, similarly to the construction of instance $\hat X$ in Section \ref{sec: Instance construction}. As before, instance $\hat X_i$ is a semi-permutation, so every row contains at most one point. Every column of $\hat X_i$ contains exactly $\log N$ points of $\hat X_i$. Let $\cset$ denote the set of all active columns of instance $\hat X_i$. For every column $C\in \cset$, we replace $C$ with a block $\bset(C)$ of $\log N$ columns, and place all points of $\hat X_i\cap C$ on the columns of $\bset(C)$, so that they form an increasing sequence, while preserving their $y$-coordinates. The resulting instance is equivalent to $X^*_i$ (to obtain instance $X^*_i$ we also need to replace every active row $R$ with a block $\bset(R)$ of $\log N$ rows; but since every row contains at most one point of $\hat X_i$, this amounts to inserting empty rows into the instance).

The analysis of $\WB(X^*_i)$ is very similar to the analysis of $\WB(X^*)$ for instance $X^*$ constructed in Section \ref{sec: Instance construction}. Notice that, as before, it is sufficient to show that $\WB(\hat X_i)\leq  O(N\log N\log\log\log N)$.
Indeed, consider the partition $\set{\bset(C)}_{C\in \cset}$ of the columns of $X^*_i$. Then $\hat X_i$ can be viewed as the compressed instance for $X^*_i$ with the respect to this partition. Each resulting strip instance (defined by the block $\bset(C)$ of columns) is an increasing sequence of $\log N$ points, so the Wilber Bound value for such an instance is $O(\log N)$. Altogether, the total Wilber Bound of all such strip instances is $O(N\log N)$. Therefore, from Theorem \ref{thm: WB for corner and strip instances}, in order to prove Lemma \ref{lem: bounding vertical WB for strip}, it is now sufficient to show that $\WB(\hat X_i)\leq O(N\log N\log\log\log N)$.

Let $\lset$ be the set of all vertical lines with half-integral  coordinates for the instance $\hat X_i$, and let $\sigma$ be any permutation of these lines. Our goal is to prove that $\WB_{\sigma}(\hat X_i)\leq O(N\log N\log\log\log N)$. Let $T=T(\sigma)$ be the partitioning tree associated with $\sigma$. Consider some vertex $v\in V(T)$ and the line $L$ that $v$ owns. As before, we classify crossings that are charged to $L$ into several types. A crossing $(p,p')$ is a type-1 crossing, if $p$ and $p'$ both lie in the same instance $X^{j,i}$. We say that instance $X^{j,i}$ is \emph{bad} for $L$, if it contributes at least one type-$1$ crossing to the cost of $L$. If $p\in X^{j,i}$ and $p'\in X^{j',i}$ for $j\neq j'$, then we say that $(p,p')$ is a type-2 crossing. If either instance $X^{j,i}$ or $X^{j',i}$ is a bad instance for $L$, then the crossing is of type (2a); otherwise it is of type (2b). 

We now bound the total number of crossings of each of these types separately.

\paragraph{Type-1 Crossings}
We bound the total number of all type-1 crossings exactly like in Section \ref{subsec: type 1 crossings}. We note that the proof does not use the vertical locations of the points in the sub-instances $X^{j,i}$, and only relies on two properties of instance $\hat X$: (i) the points in the first instance $X_0$ (corresponding to instance $X^{0,i}$) are exponentially spaced horizontally, so the $x$-coordinates of the points are integral powers of $2$, and they are all distinct; and (ii) each subsequent instance $X_s$ (corresponding to instance $X^{s,i}$) is a copy of $X_0$ that is shifted horizontally by $s$ units. Therefore, the same analysis applies, and the total number of type-1 crossings in $\hat X_i$ can be bounded by $O(N\log N\log\log\log N)$ as before.

\paragraph{Type-(2a) Crossings}
As before, we charge each type-(2a) crossing to one of the corresponding bad instances, to conclude that the total number of type-(2a) crossings is bounded by the total number of type-1 crossings, which is in turn bounded by $O(N\log N\log\log\log N)$.

\paragraph{Type-(2b) Crossings}
Recall that in order to bound the number of type-(2b) crossings, we have collapsed, for every instance $X_s$, all rows of $X_s$ into a single row. If we similarly collapse, for every instance $X^{s,i}$, all rows of this instance into a single row, we will obtain an identical set of points. This is because the only difference between instances $X_s$ and $X^{s,i}$ is vertical position of their points. Therefore, the total number of type-(2b) crossings in $\hat X_i$ is bounded by $O(N\log N)$ as before.

This finishes the proof of Lemma \ref{lem: bounding vertical WB for strip}. We conclude that $\sum_{i=1}^N\WB(X_i^*)\leq O(N^2\log N\log\log\log N)=O(N^*\log\log\log N^*)$.

\subsubsection{Bounding Wilber Bound for the Compressed Instance}
In this subsection, we prove the following lemma.

\begin{lemma}\label{lem: bounding vertical WB for compressed}
	 $\WB(\tilde X)\leq O(N^*)$.
\end{lemma}

We denote the active columns of $\tilde X$ by $C_1,\ldots,C_N$. 
Recall that each column $C_i$ contains exactly $N\log N$ input points. 
Let $\rset$ be the set of all rows with integral coordinates, so $|\rset|=N^2\log N$. Let $\bset_1,\bset_2,\ldots,\bset_{N^2}$ be a partition of the rows in $\rset$ into blocks containing $\log N$ consecutive rows each, where the blocks are indexed in their natural bottom-to-top order. Recall that each such block $\bset_i$ represents some active row of instance $\hat X$, and the points of $\tilde X$ that lie on the rows of $\bset_i$ form an increasing sequence. We also partition the rows of $\rset$ into super-blocks, $\hat \bset_1,\ldots,\hat \bset_N$, where each superblock is the union of exactly $N$ consecutive blocks. For each subinstance $X^{s,s'}$, the points of the subinstance lie on rows that belong to a single super-block.

Let $\lset$ be the set of all columns with half-integral coordinates for $\tX$, so $|\lset|\leq N$. We fix any permutation $\sigma$ of $\lset$, and prove that $\WB_{\sigma}(\tX)\leq O(N^*)$. Let $T$ be the partitioning tree associated with the permutation $\sigma$.

Consider any vertex $v\in V(T)$, its corresponding vertical strip $S=S(v)$, and the vertical line $L= L(v)$ that $v$ owns. 
Let $(p,p')$ be a crossing of $L$, so $p$ and $p'$ both lie in $S$ on opposite sides of $L$, and no point of $\tX\cap S$ lies between the row of $p$ and the row of $p'$. Assume that the row of $p$ is below the row of $p'$. We say that the crossing is \emph{left-to-right} if $p$ is to the left of $L$, and we say that it is \emph{right-to-left} otherwise. In order to bound the number of crossings, we use the following two claims.

\begin{claim}\label{claim: left to right crossing}
	Assume that $(p,p')$ is a left-to-right crossing, and assume that $p$ lies on a row of $\bset_i$ and $p'$ lies on a row of $\bset_j$, with $i\leq j$. Then either $j\leq i+1$ (so the two blocks are either identical or consecutive), or block $\bset_i$ is the last block in its super-block.
\end{claim}

\begin{proof}
Assume that $p$ lies on column $C_{s'}$ and on a row of super-block $\hat\bset_s$, so this point originally belonged to instance $X^{s,s'}$. Recall that instance $X^{s,s'+1}$  (that lies immediately to the right of $X^{s,s'}$) is obtained by circularly shifting all points in instance $X^{s,s'}$ by one unit up. In particular, a copy $p^c$ of $p$ in $X^{s,s'+1}$ should lie one row above the copy of $p$ in $X^{s,s'}$, unless $p$ lies on the last row of $X^{s,s'}$. In the latter case, block $\bset_i$ must be the last block of its superblock $\hat \bset_s$. In the former case, since point $p^c$ does not lie between the row of $p$ and the row of $p'$, and it lies on column $C_{s+1}$, the block of rows in which point $p'$ lies must be either $\bset_i$ or $\bset_{i+1}$, that is, $j\leq i+1$.
\end{proof}

\begin{claim}\label{claim: right to left crossing}
	Assume that $(p,p')$ is a right-to-left crossing, and assume that $p$ lies on a row of $\bset_i$ and $p'$ lies on a row of $\bset_j$, with $i\leq j$. Then either (i) $j\leq i+1$ (so the two blocks are identical or consecutive); or (ii) block $\bset_i$ is the last block in its super-block; or (iii) block $\bset_j$ is the first block it its super-block; or (iv) $p$ lies on the last active column in strip $S$, and $p'$ lies on the first active column in strip $S$.
\end{claim}

\begin{proof}
	Assume that $p$ lies on column $C_{s'}$ and on a row of superblock $\hat \bset_s$, so this point originally belonged to instance $X^{s,s'}$. 
	Assume for that $C_{s'}$ is not the last active column of $S$, so $C_{s'+1}$ also lies in $S$. 
	
	Recall that instance $X^{s,s'+1}$ is obtained by circularly shifting all points in instance $X^{s,s'}$ by one unit up. In particular, a copy $p^c$ of $p$ in $X^{s,s'+1}$ should lie one row above the copy of $p$ in $X^{s,s'}$, unless $p$ lies on the last row of $X^{s,s'}$. In the latter case, block $\bset_i$ must be the last block of its superblock. In the former case, since point $p^c$ does not lie between the row of $p$ and the row of $p'$, the block of rows in which point $p'$ lies is either $\bset_i$ or $\bset_{i+1}$,  that is, $j\leq i+1$.
	
	Using a symmetric argument, if $p'$ does not lie on the first active column of $S$, then either $j\leq i+1$, or $\bset_j$ is the first block in its super-block.
\end{proof}

We can now categorize all crossings charged to the line $L$ into types as follows. Let $(p,p')$ be a crossing, and assume that $p$ lies on a row of $\bset_i$, $p'$ lies on a row of $\bset_j$, and $i\leq j$. We say that $(p,p')$ is a crossing of type 1, if $j\leq i+1$. We say that 
it is a crossing of type 2 if either $\bset_i$ or $\bset_j$ are the first or the last blocks in their superblock. We say that it is of type 3 if $p$ lies on the last active column of $S$ and $p'$ lies on the first active column of $S$.

We now bound the total number of all such crossings separately.

\paragraph{Type-1 crossings}
Consider any pair $\bset_i,\bset_{i+1}$ of consecutive blocks, and let $\tilde X'_i$ be the set of all points lying on the rows of these blocks. Recall that all points lying on the rows of $\bset_i$ form an increasing sequence of length $\log N$, and the same is true for all points lying on the rows of $\bset_{i+1}$. It is then easy to see that $\opt(\tilde X'_i)\leq O(\log N)$, and so the total contribution of crossings between the points of $\tilde X'_i$ to $\WB_{\sigma}(\tilde X)$ is bounded by $O(\log N)$. Since the total number of blocks $\bset_i$ is bounded by $N^2$, the total number of type-1 crossings is at most $O(N^2\log N)$.

\paragraph{Type-2 crossings}
In order to bound the number of type-2 crossings, observe that $|\lset|\leq N$. If $L\in \lset$ is a vertical line, and $S$ is a strip that $L$ splits, then there are $N$ superblocks of rows that can contribute type-2 crossings to $\cost(L)$, and each such superblock may contribute at most one crossing. Therefore, the total number of type-2 crossings charged to $L$ is at most $N$, and the total number of all type-2 crossings is $O(N^2)$.

\paragraph{Type-3 crossings}
In order to bound the number of type-3 crossings, observe that every column contains $N\log N$ points. Therefore, if $L\in \lset$ is a vertical line, then the number of type-3 crossings charged to it is at most $2N\log N$. As $|\lset|\leq N$, we get that the total number of type-3 crossings is $O(N^2\log N)$.

To conclude, we have shown that $\WB_{\sigma}(\tX)=O(N^2\log N)=O(N^*)$, proving Lemma \ref{lem: bounding vertical WB for compressed}. By combining Lemmas \ref{lem: bounding vertical WB for strip} and \ref{lem: bounding vertical WB for compressed}, together with Theorem \ref{thm: WB for corner and strip instances}, we conclude that $\WB(X^*)=O(N^*\log\log\log N^*)$, proving Theorem \ref{thm: boundinb WB vertical}.

\subsection{Handling Horizontal Cuts}\label{subsec: bounding WB horizontal}
We show the following analogue of Theorem \ref{thm: boundinb WB vertical}.

\begin{theorem}\label{thm: boundinb WB horizontal}
	$\WB'(X^*)\leq O(N^*\log\log\log N^*)$.
\end{theorem}

The proof of the theorem is virtually identical to the proof of Theorem \ref{thm: boundinb WB vertical}. 
In fact, consider the instance $X^{**}$, that is obtained from $X^*$, by rotating it by $90$ degrees clockwise.
Consider the sequence $\BRS'(\ell, \rset, \cset)$ that is obtained by rotating the point set $\BRS(\ell, \rset, \cset)$ by $90$ degrees.  
Consider now the following process. Our starting point is the rotated Bit Reversal Sequence. We then follow exactly the same steps as in the construction of the instance $X^*$. 
Then the resulting instance is precisely (a mirror reflection of) instance $X^{**}$.
Notice that the only place where our proof uses the fact that we start with the Bit Reversal Sequence is in order to show that $\opt(X^*)$ is sufficiently large. In fact we could replace the Bit Reversal Sequence with any other point set that is a permutation, and whose optimal solution cost is as large, and in particular the Bit Reversal Sequence that is rotated by $90$ degrees would work just as well. The analysis of the Wilber Bound works exactly as before, and Theorem \ref{thm: boundinb WB horizontal} follows.

\section{The Algorithms} \label{sec: alg}

In this section we provide the proof of Theorem \ref{thm:intro_alg}. Both the polynomial time and the sub-exponential time algorithms follow the same framework. We start with a high-level overview of this framework. For simplicity, assume that the number of active columns in the input instance $X$ is an integral power of $2$.
The key idea is to decompose the input instance into smaller sub-instances,  using the split instances defined in Section \ref{sec:split instance}. We solve the resulting instances recursively and then combine the resulting solutions.

Suppose we are given an input point set $X$ that is a semi-permutation, with $|X|=m$, such that the number of active columns is $n$. We consider a \emph{balanced} partitioning tree $T$, where for every vertex $v\in V(T)$, the line $L(v)$ that $v$ owns splits the strip $S(v)$ in the middle, with respect to the active columns that are contained in $S(v)$. Therefore, the height of the partitioning tree is $\log n$.

Consider now the set $U$ of vertices of $T$ that lie in the middle layer of $T$. We consider the split of $(X,T)$ at $U$, obtaining a new collection of instances $(X^c, \set{X^s_i})_{i=1}^k$ where $k = \Theta(\sqrt{n})$. Note that each resulting strip instance $X^s_i$ contains $\Theta(\sqrt{n})$ active columns, and so does the compressed instance $X^c$.
 
We will recursively solve each such instance and then  combine the resulting solutions.
The key in the algorithm and its analysis is to show that there is a collection $Z$ of $O(|X|)$ points, such that, if we are given any  solution $Y^c$ to instance $X^c$, and, for all $1\leq i\leq k$, any  solution $Y_i$ to instance $X^s_i$, then $Z\cup Y^c\cup \left(\bigcup_{i=1}^NY_i\right) $ is a feasible solution to instance $X$. Additionally, we show that the total number of input points that appear in all instances that participate in the same recursive level is bounded by $O(\opt(X))$. This ensures that in every recursive level we add at most $O(\opt(X))$ points to the solution, and the total solution cost is at most $O(\opt(X))$ times the number of the recursive levels, which is in turn bounded by $O(\log\log n)$.

In order to obtain the sub-exponential time algorithm, we restrict the recursion to $D$ levels, and then solve each resulting instance directly in time $r(X) c(X)^{O(c(X))}$ (recall that $r(X)$ and $c(X)$ are the number of active rows and active columns in instance $X$, respectively). Altogether, this approach gives an $O(D)$-approximation algorithm with running time at most  $\poly(m)\cdot \exp\paren{n^{1/2^{\Omega(D)}} \log n}$ as desired. 
We now provide a formal definition of the algorithm and its analysis.
In the remainder of this section, all logarithms are to the base of $2$.

\subsection{Special Solutions} 

Our algorithm will produce feasible solutions of a special form, that we call \emph{special solutions}. Recall that, given a semi-permutation point set $X$, the auxiliary columns for $X$ are a set $\lset$ of vertical lines with half-integral coordinates. We say that a solution $Y$ for $X$ is \emph{special} iff every point of $Y$ lies on an row that is active for $X$, and on a column of $\lset$. In particular, special solutions are by definition non-canonical (see Figure~\ref{fig:canonoical and special solutions} for an illustration).  
The main advantage of  the special solutions is that they allow us to easily use the divide-and-conquer approach.

\begin{observation}
        \label{obs: canonical to special} 
        There is an algorithm, that, given a set $X$ of points that is a semi-permutation, and a canonical solution $Y$ for $X$, computes a special solution $Y'$ for $X$, such that $|Y'| \leq 2 |X| + 2 |Y|$.  The running time of the algorithm is $O(|X|+|Y|)$.
\end{observation}

\begin{proof}
        We construct $Y'$ as follows: For each point $p \in X \cup Y$, we add two points, $p'$ and $p''$ to $Y'$, whose $y$-coordinate is the same as that of $p$, such that $p'$ and $p''$ lie on the lines of $\lset$ appearing immediately to the left and immediately to the right of $p$, respectively. %
                It is easy to verify that $|Y'|=2|X|+2|Y|$.
        
        We now verify that $X \cup Y'$ is a feasible set of points. 
        Consider two points $p,q$ in $X \cup Y'$, such that $p.x < q.x$. Notice that $p$ is either an original point in $X$ or it is a copy of some point $\hat p\in X \cup Y$. 
        If $p \in X$ or $p = \hat p'$ for $\hat p \in X \cup Y$, then the point $\hat p''$ lies in the rectangle $\rect_{p,q}$. Therefore, we can assume that $p=\hat p''$ for some point $\hat p\in X\cup Y$.
        
        Using a similar reasoning, we can assume that $q=\hat q'$ for some point $\hat q\in X\cup Y$. Since $X\cup Y$ is a satisfied point set, there must be a point $r\in X\cup Y$ that lies in the rectangle $\rect_{\hat p,\hat q}$. Moreover, from Observation ~\ref{obs: aligned point}, we can choose $r$ such that either $r.x=\hat p.x$ or $r.y=\hat p.y$. In either case, point $r''$ also lies in $\rect_{\hat p'',\hat q'}=\rect_{p,q}$, so $(p,q)$ is satisfied by $X\cup Y'$. %
        Therefore, $X\cup Y'$ is a feasible set of points. It is easy to see that $Y'$ is a special solution.
\end{proof}

\begin{figure}
        \begin{subfigure}{.5\textwidth}
                \centering
                \includegraphics[width=.8\linewidth]{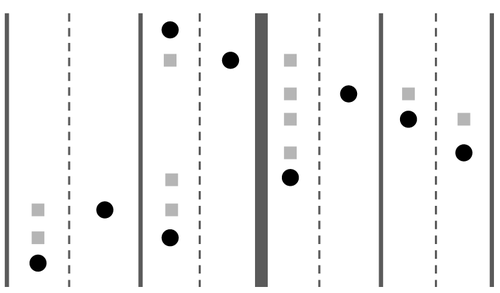}
                \caption{Canonical Solution}
                \label{fig:canonical}
        \end{subfigure}%
        \begin{subfigure}{.5\textwidth}
                \centering
                \includegraphics[width=.8\linewidth]{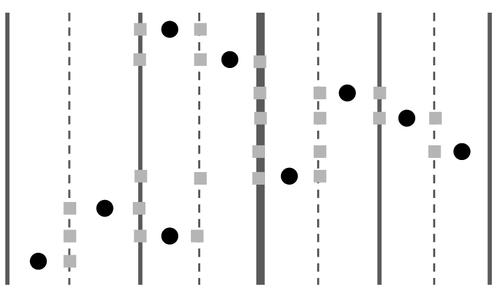}
                \caption{Special Solution}
                \label{fig:special}
        \end{subfigure}
        \caption{Canonical and $T$-special solutions of $X$. The input points are shown as circles; the points that belong to the solution $Y$ are shown as squares.}
        \label{fig:canonoical and special solutions}
\end{figure}

If $\sigma$ is any ordering of the auxiliary columns in $\lset$, and $T=T_{\sigma}$ is the corresponding partitioning tree, then any point set $Y$ that is a special solution for $X$ is also called a \emph{$T$-special solution} (although the notion of the solution $Y$ being special does not depend on the tree $T$, this notion will be useful for us later; in particular, a convenient way of thinking of a $T$-special solution is that every point of $Y$ must lie on an active row of $X$, and on a column that serves as a boundary for some strip $S(v)$, where $v\in V(T)$.)

\subsection{Redundant Points and Reduced Point Sets} 
Consider a semi-permutation $X$, that we think of as a potential input to the \minsat problem.
We denote $X= \set{p_1,\ldots, p_{m}}$, where the points are indexed in their natural bottom-to-top order, so $(p_1).y < (p_2).y < \ldots < (p_m).y$.    
A point $p_i$ is said to be \emph{redundant}, iff  $(p_i).x = (p_{i+1}).x = (p_{i-1}).x$. 
We say that a semi-permutation $X$ is in the {\em reduced form} if there are no redundant points in $X$; in other words, if $p_{i-1},p_i,p_{i+1}$ are three points lying on three consecutive active rows, then their $x$-coordinates are not all equal. We use the following observation and lemma. %

\begin{observation}\label{obs: reduced set does not increase opt}
Let $X$ be a semi-permutation, and let $X' \subseteq X$ be any point set, that is obtained from $X$ by repeatedly removing redundant points. Then $\opt(X')\leq \opt(X)$.
\end{observation}
\begin{proof}
	It is sufficient to show that, if $X''$ is a set of points obtained from $X$ by deleting a single redundant point $p_i$, then $\opt(X'')\leq \opt(X)$. Let $R$ denote the row on which point $p_i$ lies, and let $R'$ be the row containing $p_{i-1}$. Let $Y$ be the optimal solution to instance $X$. We assume w.l.o.g. that $Y$ is a canonical solution. Consider the set $Z=X\cup Y$ of point, and let $Z'$ be obtained from $Z$ by collapsing the rows $R,R'$ into the row $R'$ (since $Y$ is a canonical solution for $X$, no points of $X\cup Y$ lie strictly between the rows $R,R'$). From Observation \ref{obs: collapsing columns}, $Z'$ remains a satisfied point set. Setting $Y'=Z'\setminus X''$, it is easy to verify that $Y'$ is a feasible solution to instance $X''$, and moreover, $|Y'|\leq |Y|$. Therefore, $\opt(X'')\leq |Y'|\leq |Y|\leq \opt(X)$.
\end{proof}

\begin{lemma}\label{lem: reduced point set}
	Let $X$ be a semi-permutation, and let $X' \subseteq X$ be any point set, that is obtained from $X$ by repeatedly removing redundant points. Let $Y$ be any feasible solution for $X'$ such that every point of $Y$ lies on a row that is active for $X'$. Then $Y$ is also a feasible solution for $X$.
\end{lemma}

\begin{proof}
It is sufficient to show that, if $X''$ is a set of points obtained from $X$ by deleting a single redundant point $p_i$, and $Y$ is any canonical solution for $X''$, then $Y$ is also a feasible solution for $X$. We can then apply this argument iteratively, until we obtain a set of points that is in a reduced form.
	
	Consider any feasible canonical solution $Y$ to instance $X''$. We claim that $X \cup Y$ is a feasible set of points. Indeed,
	consider any two points $p,q \in X \cup Y$ that are not collinear. If both points are distinct from the point $p_i$, then they must be satisfied in $X\cup Y$, since both these points lie in $X''\cup Y$. Therefore, we can assume that $p=p_i$. Notice that $q \neq p_{i-1}$ and  $q \neq p_{i+1}$, since otherwise $p$ and $q$ must be collinear. Moreover, $q$ cannot lie strictly between the row of $p_{i-1}$ and the row of $p_{i+1}$, as we have assumed that every point of $Y$ lies on a row that is active for $X'$ $X''$. But then it is easy to verify that either point $p_{i-1}$ lies in $\rect_{p,q}$ (if $q$ is below $p$), or point $p_{i+1}$ lies in $\rect_{p,q}$ (otherwise). In either case, the pair $(p,q)$ is satisfied in $X\cup Y$.
\end{proof}

From Lemma \ref{lem: reduced point set}, whenever we need to solve the \minsat problem on an instance $X$, it is sufficient to solve it on a sub-instance, obtained by iteratively removing redundant points from $X$. 
We obtain the following immediate corollary of Lemma \ref{lem: reduced point set}. %

\begin{corollary}\label{cor: removing redundant poitns does not increase special solutions cost}
	Let $X$ be a semi-permutation, and let $X' \subseteq X$ be any point set, that is obtained from $X$ by repeatedly removing redundant points. Let $Y$ be any special feasible solution for $X'$. Then $Y'$ is also a special feasible solution for $X$. %
\end{corollary}

Lastly, we need the following lemma.

\begin{lemma}\label{lem: bound on WB for reduced}
        Let $X$ be a point set that is a semi-permutation in reduced form. Then $\opt(X)\geq |X|/4-1$. 
\end{lemma}

\begin{proof}
Since $X$ is a semi-permutation, every point of $X$ lies on a distinct row; we denote $|X|=n$. Let $X=\set{p_1,\ldots,p_n}$, where the points are indexed in the increasing order of their $y$-coordinates. 
Let $\Pi=\set{(p_i,p_{i+1})\mid 1\leq i<n}$ be the collection of all consecutive pairs of points in $X$. We say that the pair $(p_i,p_{i+1})$ is \emph{bad} iff both $p_i$ and $p_{i+1}$ lie on the same column. From the definition of the reduced form, if $(p_i,p_{i+1})$ is a bad pair, then both $(p_{i-1},p_i)$ and $(p_{i+1},p_{i+2})$ are good pairs. Let $\Pi'\subseteq \Pi$ be the subset containing all good pairs. Then $|\Pi'|\geq (|\Pi|-1)/2\geq n/2-1$. Next, we select a subset $\Pi''\subseteq \Pi'$ of pairs, such that $|\Pi''|\geq |\Pi'|/2\geq n/4-1$, and every point in $X$ belongs to at most one pair in $\Pi''$. Since every point in $X$ belongs to at most two pairs in $\Pi'$, it is easy to see that such a set exists. Let $Y$ be an optimal solution to instance $X$.

Consider now any pair $(p_i,p_{i+1})$ of points in $\Pi''$. Then there must be a point $y_i\in Y$ that lies in the rectangle $\rect_{p_i,p_{i+1}}$. Moreover, since all points of $X$ lie on distinct rows, and each such point belongs to at most one pair in $\Pi''$, for $i\neq j$, $y_i\neq y_j$. Therefore, $|Y|\geq |\Pi''|\geq n/4-1$.
\end{proof}

\subsection{The Algorithm Description} 
\newcommand{\progspace}{\hspace{0.5cm}} 

Suppose we are given an input set $X$ of points that is a semi-permutation. Let $T$ be any partitioning tree for $X$. We say that $T$ is a \emph{balanced} partitioning tree for $X$ iff for every non-leaf vertex $v\in V(T)$ the following holds. Let $v'$ and $v''$ be the children of $v$ in the tree $T$. Let $X '$ be the set of all input points lying in strip $S(v)$, and let $X'',X'''$ be defined similarly for $S(v')$ and $S(v'')$. Let $c$ be the number of active columns in instance $X'$, and let $c'$ and $c''$ be defined similarly for $X''$ and $X'''$. Then we require that $c',c''\leq \ceil{c/2}$.

Given a partitioning tree $T$, we denote by $\Lambda_i$ the set of all vertices of $T$ that lie in the $i$th \emph{layer} of $T$ -- that is, the vertices whose distance from the root of $T$ is $i$ (so the root belongs to $\Lambda_0$). The \emph{height} of the tree $T$, denoted by $\height(T)$, is the largest index $i$ such that $\Lambda_i\neq \emptyset$. If the height of the tree $T$ is $h$, then we call the set  $\Lambda_{\ceil{h/2}}$ of vertices the \emph{middle layer} of $T$. Notice that, if $T$ is a balanced partitioning tree for input $X$, then its height is at most $2 \log c(X)$.

Our algorithm takes as input a set $X$ of points that is a semi-permutation, a balanced partition tree $T$ for $X$,  and an integral parameter $\rho>0$. 

Intuitively, the algorithm uses the splitting operation to partition the instance $X$ into subinstances that are then solved recursively, until it obtains a collection of instances whose corresponding partitioning trees have height at most $\rho$. We then either employ dynamic programming, or use a trivial $O(\log c(X))$-approximation algorithm. 
The algorithm returns a special feasible solution for the instance. 
Recall that the height of the tree $T$ is bounded by $2\log c(X)\leq 2\log n$.

The following two theorems will be used as the recursion basis.

\begin{theorem}
	\label{thm: leaf problem1} 
	There is an algorithm called {\sc LeafBST-1} that, given a semi-permutation instance $X$ of \minsat in reduced form, and a partitioning tree $T$ for it, produces a feasible $T$-special solution for $X$ of cost at most $2 |X| + 2\opt(X)$, in time $|X|^{O(1)} \cdot c(X)^{O(c(X))}$.
\end{theorem}

\begin{theorem}
	\label{thm: leaf problem2} 
	There is an algorithm called {\sc LeafBST-2} that, given a semi-permutation instance $X$ of \minsat in reduced form, and a partitioning tree $T$ for it, produces a feasible $T$-special solution of cost at most $2 |X| \height(T)$, in time $\poly(|X|)$. 
\end{theorem}

The proofs of both theorems are deferred to the next two subsections, after we complete the proof of Theorem \ref{thm:intro_alg} using them.
We now provide a schematic description of our algorithm. Depending on the guarantees that we would like to achieve, whenever the algorithm calls procedure  {{\sc LeafBST}, it will call either procedure {\sc LeafBST-1} or procedure {\sc LeafBST-2}; we specify this later.

\program{{\sc RecursiveBST}($X,T,\rho$)}{

\item Keep removing redundant points from $X$ until $X$ is in reduced form. 

\item IF $T$ has height at most $\rho$, 

\item \progspace {\bf return} {\sc LeafBST}($X,T$)

\item Let $U$ be the set of vertices lying in the middle middle layer of $T$. 

\item Compute the split $(X^c,\{X^s_v\}_{v \in U})$ of $(X,T)$ at $U$.

\item Compute the corresponding sub-trees $(T^c, \{T^s_v\}_{v \in U})$ of $T$. 

\item  For each vertex $v \in U$, call to {\sc RecursiveBST} with input $(X^s_v, T^s_v,\rho)$, and let $Y_v$ be the solution returned by it.

\item  Call {\sc RecursiveBST} with input $(X^c, T^c, \rho)$, and let $\hat Y$ be the solution returned by it.

\item Let $Z$ be a point set containing, for each vertex $v\in U$, for each point $p \in X^s_v$, two copies $p'$ and $p''$ of $p$ with $p'.y=p''.y=p.y$, where $p'$ lies on the left boundary of $S(v)$, and $p''$ lies on the right boundary of $S(v)$.

\item {\bf return} $Y^*=Z \cup \hat{Y} \cup (\bigcup_{v \in U} Y_v)$ 
}

\subsection{Analysis}

We start by showing that the solution that the algorithm returns is $T$-special.

\begin{observation}\label{obs: get special sol}
Assuming that {\sc LeafBST}$(X,T)$ returns a $T$-special solution, the solution $Y^*$ returned by Algorithm {\sc RecursiveBST}$(X,T,\rho)$ is  a $T$-special solution. 
\end{observation}
\begin{proof}
The proof is by induction on the recursion depth. The base of the induction is the calls to Procedure {\sc LeafBST}$(X,T)$, which return $T$-special solutions by our assumption.

Consider now some call to Algorithm {\sc RecursiveBST}$(X,T,\rho)$. From the induction hypothesis, the resulting solution $\hat{Y}$ for instance $X^c$ is $T^c$-special, and, for every vertex $v \in U$, the resulting solution $Y_v$ for instance $X^s_v$ is $T_v$-special. 
Since both $T^c$ and every tree $\set{T_v}_{v\in U}$ are subtrees of $T$, and since the points of $Z$ lie on boundaries of strips in $\set{S(v)}_{v\in U}$, the final solution $Y^*$ is $T$-special. 
\end{proof}

We next turn to prove that the solution $Y^*$ computed by Algorithm {\sc RecursiveBST}$(X,T,\rho)$ is feasible. In order to do so, we will use the following immediate observation.

\begin{observation}
	\label{obs:points structure} 
	Let $Y^*$ be the solution returned by Algorithm {\sc RecursiveBST}$(X,T,\rho)$, and let $u\in U$ be any vertex. Then:
	\begin{itemize}
		\item Any point $y\in Y^*$ that lies in the interior of $S(u)$ must lie on an active row of instance $X^s_u$.
		\item Any point $y\in Y^*$ that lies on the boundary of $S(u)$ must belong to in $\hat{Y}\cup Z$. Moreover, the points of $\hat{Y} \cup Z$ may not lie in the interior of $S(u)$. 
\item If $R$ is an active row for instance $X^s_u$, then set $Z$ contains two points, lying on the intersection of $R$ with the left and the right boundaries of $S(u)$, respectively.
		
	\end{itemize}
\end{observation}

We are now ready to prove that the algorithm returns feasible solutions.
In the following proof, when we say that a row $R$ is an active row of strip $S(u)$, we mean that some point of instance $X^s_v$ lies on row $R$, or equivalently, $R$ is an active row for instance $X^s_v$.

\begin{theorem} \label{thm: algorithm feasibility}
	Assume that the recursive calls to Algorithm {\sc RecursiveBST} return a feasible special solution $\hat Y$ for instance $X^c$, and for each $v\in U$, a feasible special solution $Y_v$ for the strip instance $X^s_v$. Then the point set $Y^*=Z \cup \hat{Y} \cup (\bigcup_{v \in U} Y_v)$ is a feasible solution for instance $X$. 
\end{theorem}

\begin{proof}
	It would be convenient for us to consider the set of all points in  $X^c\cup X\cup Y^*$ simultaneously. 
	In order to do so, we start with the set $X\cup Y^*$ of points. For every vertex $v\in U$, we select an arbitrary active column $C_v$ in strip $S(v)$, and we then add a copy of every point $p\in X^s_v$ to column $C_v$. The resulting set of points, obtained after processing all strips $S(v)$ is identical to the set $X^c$ of points (except possibly for horizontal spacing between active columns), and we do not distinguish between them.

	Consider any pair of points $p,q$ that lie in $Y^*\cup X$, which are not collinear. Our goal is to prove that some point $r\neq p,q$ with $r\in X\cup Y^*$ lies in $\rect_{p,q}$. We assume w.l.o.g. that $p$ lies to the left of $q$. We also assume that $p.y<q.y$ (that is, point $p$ is below point $q$); the other case is symmetric.
	
	Assume first that at least one of the two points (say $p$) lies in the interior of a strip $S(u)$, for some $u\in U$. We then consider two cases. First, if $q$ also lies in the interior of the same strip, then $p,q\in X^s_u \cup Y_u$, and, since we have assumed that $Y_u$ is a feasible solution for instance $X^s_u$, the two points are satisfied in $X^s_u\cup Y_u$, and hence in $X\cup Y^*$. Otherwise, $q$ does not lie in the interior of strip $S(u)$. Then, from Observation \ref{obs:points structure}, if $R$ is the row on which point $p$ lies, then $R$ is an active row for instance $X^s_u$, and the point that lies on the intersection of the row $R$ and the right boundary of strip $S(u)$ was added to $Z$. This point satisfies the pair $(p,q)$.

Therefore, we can now assume that both $p$ and $q$ lie on boundaries of strips $\set{S(u)\mid u\in U}$. Since every pair of consecutive strips share a boundary, and since, from the above assumptions, $p$ lies to the left of $q$, we can choose the strips $S(u)$ and $S(v)$, such that $p$ lies on the left boundary of $S(u)$, and $q$ lies on the right boundary of $S(v)$. Notice that it is possible that $S(v)=S(u)$.

Notice that, if $p$ lies on a row that is active for strip $S(u)$, then a point that lies on the same row and belongs to the right boundary of the strip $S(u)$ has been added to $Z$; this point satisfies the pair $(p,q)$. Similarly, if $q$ lies on a row that is active for strip $S(v)$, the pair $(p,q)$ is satisfied by a point of $Z$ that lies on the same row and belongs to the left boundary of $S(v)$.

Therefore, it remains to consider the case where point $p$ lies on a row that is inactive for $S(u)$, and point $q$ lies on a row that is inactive for $S(v)$. From Observation  \ref{obs:points structure}, both $p$ and $q$ belong to $\hat Y\cup Z$.

The following observation will be useful for us. Recall that the points of $X^c$ are not included in $X\cup Y^*$.

\begin{observation}\label{obs: satisfied by Xc}
	Assume that there is some point $r\neq p,q$, such that $r\in \rect_{p,q}$, and $r\in X^c$. Then the pair $(p,q)$ is satisfied in set $X\cup Y^*$.
	\end{observation}

\begin{proof}
	Since $r\in X^c$, and $r\in \rect_{p,q}$, there must be some strip $S(w)$, that lies between strips $S(u)$ and $S(v)$ (where possibly $w=u$ or $w=v$ or both), such that point $r$ lies on the column $C_w$ (recall that this is the unique active column of $S(w)$ that may contain points of $X^c$). But then the row $R$ to which point $r$ belongs is an active row for strip $S(w)$. Therefore, two  points, lying on the intersection of $R$ with the two boundaries of strip $S(w)$ were added to $Z$, and at least one of these points must lie in $\rect_{p,q}$. Since $Z\subseteq Y^*$, the observation follows.
\end{proof}

From the above observation, it is sufficient to show that some point $r\in \hat Y\cup Z\cup X^c$ that is distinct from $p$ and $q$, lies in $\rect_{p,q}$.
We distinguish between three cases.

The first case happens when $p,q\in \hat Y$. Since set $\hat Y$ is a feasible solution for instance $X^c$, there is some point $r\in \rect_{p,q}$ that is distinct from $p$ and $q$, and lies in $\hat Y\cup X^c$. %

The second case happens when neither $p$ nor $q$ lie in $\hat Y$, so both $p,q\in Z$. Let $u'\in U$ be the vertex for which strip $S(u')$ lies immediately to the left of strip $S(u)$. Since $p$ lies on a row that is inactive for strip $S(u)$, but $p\in Z$, such a vertex must exist, and moreover, the row $R$ to which $p$ belongs must be active for strip $S(u')$. Therefore, the point lying on the intersection of the column $C_{u'}$ (the unique active column of $S(u')$ containing points of $X^c$) and $R$ belongs to $X^c$; we denote this point by $p'$. 

Similarly, we let $v'\in U$ be the vertex  for which strip $S(v')$ lies immediately to the right of $S(v)$. Since $q$ lies on a row that is inactive for strip $S(v)$, but $q\in Z$, such a vertex must exist, and moreover, the row $R'$ to which $q$ belongs must be active for strip $S(v')$. Therefore, the point lying on the intersection of $C_{v'}$ and $R'$ belongs to $X^c$; we denote this point by $q'$.

Since the set $X^c\cup \hat Y$ of points is satisfied, some point $r\in X^c\cup \hat Y$ that is distinct from $p'$ and $q'$, lies in $\rect_{p',q'}$. Moreover, from Observation \ref{obs: aligned point}, we can choose this point so that it lies on the boundary of the rectangle $\rect_{p',q'}$. Assume first that $r$ lies on the left boundary of this rectangle. Then, since $\hat Y$ is a special solution for instance $X^c$, $r\in X^c$ must hold. If $R''$ denotes the row on which $r$ lies, then $R''$ is an active row for strip $S(u')$, and so a point that lies on the intersection of row $R''$ and the right boundary of $S(u')$ belongs to $Z$. That point satisfies $\rect_{p,q}$. The case where $r$ lies on the right boundary of $\rect_{p',q'}$ is treated similarly.

Assume now that $r$ lies on the top or the bottom boundary of $\rect_{p',q'}$, but not on one of its corners. Then, since solution $\hat Y$ is special for $X^c$, point $r$ must lie in the rectangle $\rect_{p,q}$. Moreover, since we have assumed that neither $p$ nor $q$ lie in $\hat Y$, $r\neq p,q$. But then $r\in X^c\cup \hat Y$ lies in $\rect_{p,q}\setminus\set{p,q}$ and by Observation \ref{obs: satisfied by Xc}, pair $(p,q)$ is satisfied in $X\cup Y^*$.

The third case happens when exactly one of the two points (say $p$) lies in $\hat Y$, and the other point does not lie in $\hat Y$, so $q\in Z$ must hold. We define the strip $S(v')$ and the point $q'\in X^c$ exactly as in Case 2. Since $p,q'\in X^c\cup \hat Y$, there must be a point $r\in \rect_{p,q'}\setminus\set{p,q'}$ that lies in $X^c\cup \hat Y$. From Observation \ref{obs:points structure}, we can choose the point $r$, so that it lies on the left or on the bottom boundary of $\rect_{p,q'}$ (that is, its $x$- or its $y$-coordinate is aligned with the point $p$). If $r$ lies on the left boundary of $\rect_{p,q'}$, then it also lies in $\rect_{p,q}$, and from Observation  \ref{obs: satisfied by Xc}, pair $(p,q)$ is satisfied in $X\cup Y^*$. If it lies on the bottom boundary of $\rect_{p,q'}$, but it is not the bottom right corner of the rectangle, then, using the same reasoning as in Case 2, it must lie in $\rect_{p,q}$, and it is easy to see that $r\neq q$. Lastly, if $r$ is the bottom right corner of $\rect_{p,q'}$, then $r\in X^c$. As before, there is a copy of $r$, that lies on the left boundary of strip $S(v')$ and belongs to $Z$, that satisfies the pair $(p,q)$. 
\end{proof}

In order to analyze the solution cost, consider the final solution $Y^*$ to the input instance $X$. We distinguish between two types of pints in $Y^*$: a point $p\in Y^*$ is said to be of type-2 if it was added to the solution by Algorithm \leafBST, and otherwise we say that it is of type 1. We start by bounding the number of points of type $1$ in $Y^*$.

\begin{claim}\label{claim: type 1 points}
	The number of points of type 1 in the solution $Y^*$ to the original instance $X$ is at most $O(\log(\height(T)/\rho))\cdot \opt(X)$.
	\end{claim}

\begin{proof}
	Observe that the number of recursive levels is bounded by $\lambda =O(\log(\height(T)/\rho))$. 
	This is since, in every recursive level, the heights of all trees go down by a constant factor, and we terminate the algorithm once the tree heights are bounded by $\rho$.
	 For each $1\leq i\leq \lambda$, let $\xset_i$ be the collection of all instances in the $i$th recursive level, where the instances are in the reduced form. Notice that the only points that are added to the solution by Algorithm \recBST directly are the points in the sets $Z$. The number of such points added at recursive level $i$ is bounded by $\sum_{X'\in \xset_i}2|X'|$. It is now sufficient to show that for all $1\leq i\leq \lambda$, $\sum_{X'\in \xset_i}|X'|\leq O(\opt(X))$. We do so using the following observation.
	
	\begin{observation}\label{obs: type 1 points}
		For all $1\leq i\leq \lambda$, $\sum_{X'\in \xset_i}\opt(X')\leq \opt(X)$.
		\end{observation}
	
	Assume first that the observation is correct. For each instance $X'\in \xset_i$, let $T'$ be the partitioning tree associated with $X'$. From Lemma \ref{lem: bound on WB for reduced}, $|X'|\leq O(\opt(X'))$. Therefore, the number of type-1 points added to  the solution at recursive level $i$ is bounded by $O(\opt)$.

We now turn to prove Observation \ref{obs: type 1 points}.

\begin{proofof}{Observation \ref{obs: type 1 points}}
	The proof is by induction on the recursive level $i$. It is easy to see that the claim holds for $i=1$, since, from Observation \ref{obs: reduced set does not increase opt}, removing redundant points from $X$ to turn it into reduced form cannot increase $\opt(X)$.
	
	Assume now that the claim holds for level $i$, and consider some level-$i$ instance $X'\in \xset_i$. Let $(X^c,\set{X^s_u}_{u\in U})$ be the split of  $(X',T')$ that we have computed. Then, from  Theorem \ref{thm: opt decomposes}, $\sum_{v \in U}\opt(X^s_v) + \opt(X^c) \leq \opt(X')$. Since, from Observation \ref{obs: reduced set does not increase opt}, removing redundant points from an instance does not increase its optimal solution cost, the observation follows.
\end{proofof}		 
\end{proof}

In order to obtain an efficient $O(\log\log n)$-approximation algorithm, we set $\rho$ to be a constant (it can even be set to $1$), and we use  algorithm \leafBST-2 whenever the algorithm calls to subroutine \leafBST. Observe that the depth of the recursion is now bounded by $O(\log\log n)$, and so the total number of type-$1$ points in the solution is bounded by $O(\log\log n)\cdot \opt(X)$. Let $\iset$ denote the set of all instances to which Algorithm \leafBST is applied. Using the same arguments as in Claim \ref{claim: type 1 points}, $\sum_{X'\in \iset}|X'|=O(\opt(X))$. The number of type-2 points that Algorithm \leafBST adds to the solution for each instance $X'\in \iset$ is bounded by $O(|X'|\cdot \rho)=O(|X'|)$. Therefore, the total number of type-2 points in the solution is bounded by $O(\opt(X))$. Overall, we obtain a solution of cost at most $O(\log\log n)\cdot \opt(X)$, and the running time of the algorithm is polynomial in $|X|$.

Finally, in order to obtain the sub-exponential time algorithm, we set the parameter $\rho$ to be such that the recursion depth is bounded by $D$. Since the number of active columns in instance $X$ is $c(X)$, and the height of the partitioning tree $T$ is bounded by $2\log c(X)$, while the depth of the recursion is at most $2\log(\height(T)/\rho)$, it is easy to verify that $\rho=O\left (\frac{\log c(X)}{2^{D/2}}\right )=\frac{\log c(X)}{2^{\Omega(D)}}$. We use algorithm \leafBST-1 whenever the algorithm calls to subroutine \leafBST. As before, let $\iset$ be the set of all instances to which Algorithm \leafBST is applied.
 Using the same arguments as in Claim \ref{claim: type 1 points}, $\sum_{X'\in \iset}(|X'|+\opt(X'))=O(\opt(X))$. For each such instance $X'$, Algorithm \leafBST-1 produces a solution of cost $O(|X'|+\opt(X'))$. Therefore, the total number of type-$2$ points in the final solution is bounded by $O(\opt(X))$. The total number of type-$1$ points  in the solution is therefore bounded by $O(D)\cdot \opt(X)$ as before. Therefore, the algorithm produces a factor-$O(D)$-approximate solution. Finally, in order to analyze the running time of the algorithm, we first bound the running time of all calls to procedure \leafBST-1. The number of such calls is bounded by $|X|$. Consider now some instance $X'\in \iset$, and its corresponding partitioning tree $T'$. Since the height of $T'$ is bounded by $\rho$, we get that $c(X')\leq 2^{\rho}\leq 2^{\log c(X)/2^{\Omega(D)}}\leq (c(X))^{1/2^{\Omega(D)}}$. 
 Therefore, the running time of \leafBST-1 on instance $X'$ is bounded by $|X'|^{O(1)}\cdot (c(X'))^{O(c(X'))}\leq |X'|^{O(1)}\cdot \exp\left({O(c(X')\log c(X')}\right )\leq |X'|^{O(1)}\cdot \exp\left({c(X)^{1/2^{\Omega(D)}}\cdot \log c(X)}\right )$.
 
The running time of the remainder of the algorithm, excluding the calls to \leafBST-1, is bounded by $\poly(|X|)$. We conclude that the total running time of the algorithm is bounded by $|X|^{O(1)}\cdot \exp\left({c(X)^{1/2^{\Omega(D)}}\cdot \log c(X)}\right )\leq \poly(m)\cdot \exp\left(n^{1/2^{\Omega(D)}}\cdot \log n \right )$. 
It now remains to prove Theorems \ref{thm: leaf problem1} and \ref{thm: leaf problem2}, which we do next.

\subsection{Leaf Instances: Proof of Theorem \ref{thm: leaf problem1}}

The goal of this subsection is to prove Theorem \ref{thm: leaf problem1}. For convenience, given an input instance $X$ of \minsat, we denote $r(X)=m$ and $c(X)=n$. Our goal is to compute an optimal canonical solution $Y$ for $X$ in time $\poly(m)\cdot n^{O(n)}$. Using Observation \ref{obs: canonical to special}, we can  then compute a special solution for instance $X$ of cost at most $2|X|+2|Y|\leq 2|X|+2\opt(X)$, in time $\poly(|X|)$. Therefore, in order to complete the proof of Theorem \ref{thm: leaf problem1}, it is enough to show an algorithm that, given an input $X$ as above, computes an optimal canonical solution for $X$ in time $\poly(m)\cdot n^{O(n)}$

We start by providing several definitions and structural observations that will be helpful in designing the algorithm.

\subsubsection{Conflicting sets} 

Our algorithm uses the notion of conflicting point sets, defined as follows.

\begin{definition}[Conflicting Sets]  
Let $Z$ and $Z'$ be sets of points. 
We say that $Z$ and $Z'$ are {\em conflicting} if there is a pair of points $p\in Z$ and $p' \in Z'$, such that $(p,p')$ is not satisfied in $Z\cup Z'$. In other words, set $Z\cup Z'$ contains no point in $\Box_{p,p'} \setminus \set{p,p'}$.  %
\end{definition} 

\begin{lemma}
\label{lem: conflict sets} 
Let $Z$ and $Z'$ be two point sets, each of which is a satisfied set. 
Then the sets $Z$ and $Z'$ are not conflicting if and only if $Z \cup Z'$ is satisfied. 
\end{lemma} 
\begin{proof}
	Assume first that $Z\cup Z'$ is satisfied. Then it is clear that for every pair $p\in Z$, $p'\in Z'$ of points, some point $r\in Z\cup Z'$ lies in $\Box_{p,p'} \setminus \set{p,p'})$, so $Z$ and $Z'$ are not conflicting.
	
	Assume now that $Z$ and $Z'$ are conflicting, and let $p\in Z$ and $p'\in Z'$ be the corresponding pair of points. Then $(p,p')$ is not satisfied in $Z\cup Z'$, and so $Z\cup Z'$ is not a satisfied set of points.
\end{proof}

The following definition is central to our algorithm.

\begin{definition}[Top representation]  
Let $Z$ be any set of points.
A {\em top representation} of $Z$, that we denote by $\topset(Z)$, is a subset $Z' \subseteq Z$ of points, obtained as follows: for every column $C$ that contains points from $Z$, we add the topmost point of $Z$ that lies on $C$ to $Z'$.
\end{definition} 

\begin{observation} 
\label{obs:top points sufficient} 
Let $\topset(Z) \subseteq Z$ be the top representation of $Z$, and let $R$ be a row lying strictly above all points in $Z$. Let $Y$ be any set of points lying on row $R$. Then $top(Z)$ is conflicting with $Y$ if and only if $Z$ is conflicting with $Y$.  
\end{observation} 

\begin{proof}
	Assume that $Y$ is conflicting with $Z$, and let $p\in Y$, $q\in Z$ be a pair of points, such that no point of $Y\cup Z$ lies in $\rect_{p,q}\setminus\set{p,q}$. But then $q\in \topset(Z)$ must hold, and no point of $\topset(Z)\cup Y$ lies in $\rect_{p,q}\setminus\set{p,q}$, so $\topset(Z)$ and $Y$ are conflicting.
	
	Assume now that $\topset(Z)$ and $Y$ are conflicting, and let $p\in Y$, $q\in \topset(Z)$ be a pair of points, such that no point of $Y\cup \topset(Z)$ lies in $\rect_{p,q}\setminus\set{p,q}$. But then no point of $Y\cup Z$ lies in $\rect_{p,q}\setminus\set{p,q}$, and, since $\topset(Z)\subseteq Z$, sets $\topset(Z)$ and $Y$ are conflicting.
\end{proof}

\subsubsection{The Algorithm} 

Let $X$ be the input point set, that is a semi-permutation.
We denote by $\rset=\set{R_1,\ldots,R_m}$ the set of all active rows for $X$, and we assume that they are indexed in their natural bottom-to-top order. We denote by $\cset=\set{C_1,\ldots,C_n}$ the set of all active columns of $X$, and we assume that they are indexed in their natural left-to-right order. We also denote $X=\set{p_1,\ldots,p_m}$, where for all $1\leq i\leq m$, $p_i$ is the unique point of $X$ lying on row $R_i$.
For an index $1\leq t\leq m$, we denote by $\rset_{\leq t}=\set{R_1,\ldots,R_t}$, and we denote by $X_{\leq t}=\set{p_1,\ldots,p_t}$.

Note that, if $Y$ is a feasible solution to instance $X$, then for all $1\leq t\leq m$, the set $X_{\leq t}\cup Y_{\leq t}$ of points must be satisfied (here $Y_{\leq t}$ is the set of all points of $Y$ lying on rows of $\rset_{\leq t}$.) Our dynamic programming-based algorithm constructs the optimal solution row-by-row, using this observation. We use height profiles, that we define next, as the ``states'' of the dynamic program.

\paragraph{Height Profile:} 
A \emph{height profile} $\pi$ assigns, to every column $C_i\in \cset$, a value $\pi(C_i)\in \set{1,\ldots,n,\infty}$. Let $\Pi$ be the set of all possible height profiles, so $|\Pi|\leq n^{O(n)}$. For a profile $\pi\in \Pi$, we denote by $M(\pi)$ the largest value of $\pi(C_i)$ for any column $C_i\in \cset$ that is not $\infty$. Given a height profile $\pi$, let $\cset(\pi)\subseteq \cset$ be the set of columns $C_i$ with $\pi(C_i)<\infty$, and let $\cset'(\pi)=\cset\setminus\cset(\pi)$. We can then naturally associate an ordering $\rho_{\pi}$ of the columns in $\cset(\pi)$ with $\pi$ as follows: for columns $C_i,C_j\in \cset(\pi)$, $C_i$ appears before $C_j$ in $\pi$ iff either  (i) $\pi(C_i)<\pi(C_{j})$; or (ii) $\pi(C_i)=\pi(C_{j})$ and $i<j$.

Consider now any point set $Z$, where every point lies on a column of $\cset$. Let $Z'=\topset(Z)$, and let $\sigma(Z)$ denote the ordering of the points in $Z'$, such that $p$ appears before $p'$ in $\sigma$ iff either (i) $p.y<p'.y$; or (ii) $p.y=p'.y$ and $p.x<p'.x$. 
Consider now any profile  $\pi\in \Pi$.
We say that point set $Z$ is \emph{consistent with profile $\pi$} (see Figure~\ref{fig:height} for an illustration) iff the following hold:

\begin{itemize}
	\item For every column $C_i\in \cset$, if $\pi(C_i)=\infty$, then no point of $Z$ lies on $C_i$; and
	\item For all $1\leq i\leq |Z'|$, the $i$th point in $\sigma(Z)$ lies on the $i$th column of $\rho_{\pi}$.
\end{itemize}

\begin{figure}[h]
    \centering
    \includegraphics[width=0.3\textwidth]{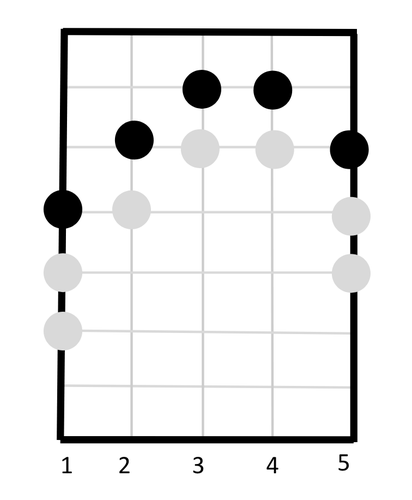}
    \caption{An illustration of height profile that is consistent with a set $Z$ of points. The set $\topset(Z)$ is shown by dark points. The height profile is $\pi(C_1) = 1, \pi(C_2) = \pi(C_5) = 2$ and $\pi(C_3) = h_Z(C_4) = 3$.
    }  
    \label{fig:height}
\end{figure}

\paragraph{Dynamic Programming Table.}

Consider some integer $1\leq t\leq m$ and some height profile $\pi$; recall that $M(\pi)$ is the largest value of $\pi(C_j)$ for $C_j\in \cset$ that is not $\infty$.
We say that $\pi$ is a \emph{legal profile} for $t$ iff  $M(\pi)\leq t$, and, if $C_i$ is the column containing the input point $p_t$, then $\pi(C_i)=M(\pi)$ (that is, column $C_i$ has the largest value $\pi(C_i)$ that is not $\infty$; we note that it is possible that other columns $C_j\neq C_i$ also have $\pi(C_j)=M(\pi)$).

For every integer $1\leq t\leq m$, and every height profile $\pi$ that is legal for $t$, there is an entry $T[t,\pi]$ in the dynamic programming table. The entry is supposed to store the minimum-cardinality set $Z$ of points with the following properties:

\begin{itemize}
	\item $X_{\leq t}\subseteq Z$;
	\item all points of $Z$ lie on rows $R_1,\ldots, R_t$;
\item 	$Z$ is a satisfied point set; and 
\item $Z$ is consistent with $\pi$.
\end{itemize}

Clearly, there number of entries in the dynamic programming table is bounded by $m n^{O(n)}$. We fill out the entries in the increasing order of the index $t$.  

We start with $t=1$. Consider any profile $\pi$ that is consistent with $t$. Recall that for every column $C_j\in \cset$, $\pi(C_j)\in \set{1,\infty}$ must hold, and moreover, if $C_i$ is the unique column containing the point $p_1\in X$, then $\pi(C_i)=1$ must hold.
We let $T[1,\pi]$ contain the following set of points: for every column $C_j\in \cset$ with $\pi(C_j)=1$, we add the point lying in the intersection of column $C_j$ and row $R_1$ to the point set stored in $T[t,\pi]$. It is immediate to verify that the resulting point set is consistent with $\pi$, it is satisfied, it contains $X_{\leq 1}=\set{p_1}$, and it is the smallest-cardinality point set with these properties.

We now assume that for some $t\geq 1$, we have computed correctly all entries $T[t',\pi']$ for all $1\leq t'\leq t$, and for all profiles $\pi'$ consistent with $t'$. We now fix some profile $\pi$ that is consistent with $t+1$, and show how to compute entry $T[t+1,\pi]$.

Let $r=M(\pi)$ (recall that this is the largest value of $\pi(C_j)$ that is not $\infty$), and let $\cset_1\subseteq \cset$ be the set of all columns $C_j$ with $\pi(C_j)=r$. Recall that, if $C_j$ is the column containing the input point $p_{t+1}$, then $C_j\in \cset_1$ must hold. Let $P$ be the set of $|\cset_1|$ points that lie on the intersection of the row $R_{t+1}$ and the columns in $\cset_1$.

Consider now any profile $\pi'$ that is legal for $t$, and let $\hat Z=T[t,\pi']$. Denote $\hat Z'=\topset(\hat Z)$. We say that profile $\pi'$ is a \emph{candidate profile} if (i) $\pi'$ is legal for $t$; (ii) the point sets $P,\hat Z'$ do not conflict; (iii) $\cset'(\pi')\subseteq \cset'(\pi)\cup \cset_1$; and (iv) if we discard from $\rho_{\pi}$ and from $\rho_{\pi'}$ the columns of $\cset_1$, then the two orderings are defined over the same set of columns and they are identical. We select a candidate profile $\pi'$ that minimizes $|T[t,\pi']|$, and let $Z=\hat Z\cup P$, where $\hat Z$ is the point set stored in $T[t,\pi']$. We then set $T[t+1,\pi]=Z$.

We now verify that set $Z$ has all required properties. If the entry $T[t,\pi']$ was computed correctly, then point set $\hat Z$ is satisfied. Since $\topset(\hat Z)$ and $P$ are not conflicting, from Observation \ref{obs:top points sufficient} neither are $\hat Z$ and $P$. Therefore, from Lemma \ref{lem: conflict sets}, set $Z$ is satisfied. 
If the entry $T[t,\pi']$ was computed correctly, then $X_{\leq t}\subseteq \hat Z$. Since $p_{t+1}\in P$, we get that $X_{\leq (t+1)}\subseteq Z$. 

Next, we show that $Z$ is consistent with the profile $\pi$. Let $Z'=\topset(Z)$, and consider the ordering $\sigma(Z)$ of the points in $Z'$. It is easy to verify that the last $|P|$ points in this ordering are precisely the points of $P$. The remaining points in this ordering can be obtained from the point set $\hat Z'$, by first ordering them according to ordering $\sigma(\hat Z)$, and then deleting all points lying on columns of $\cset_1$. From the definition of candidate profiles, it is easy to verify that for all $1\leq i\leq |Z'|$, the $i$th point in $\sigma(Z)$ lies on the $i$th column of $\rho_{\pi}$. Therefore, $Z$ is consistent with profile $\pi$.

Lastly, it remains to show that the cardinality of $Z$ is minimized among all sets with the above properties. Let $Z^*$ be the point set that contains $X_{\leq t+1}$, is satisfied, is consistent with profile $\pi$, with every point of $Z^*$ lying on rows $R_1,\ldots,R_{t+1}$,  such that $|Z^*|$ is minimized among all such sets. It is easy to verify that the set of points of $Z^*$ lying on row $R_{t+1}$ must be precisely $P$. This is since point $p_{t+1}$ must belong to $Z^*$, and so every column $C_i$ with $\pi(C_i)=M(\pi)$ must have a point lying on the intersection of $C_i$ and $R_{t+1}$ in $Z^*$. Let $\hat Z=Z^*\setminus P$.
Clearly, $\hat Z$ is satisfied, all points of $\hat Z$ lie on rows $R_1,\ldots,R_t$, and $X_{\leq t}\subseteq \hat Z$. Moreover, there is no conflict between $\topset(\hat Z)$ and $P$.

We define a new height profile $\pi'$ as follows: for every column $C_i\in \cset$, if no point of $\hat Z$ lies on $C_i$, then $\pi'(C_i)=\infty$. Otherwise, let $p$ be the unique point in $\hat Z'$ that lies on $C_i$, and assume that it lies on row $R_{t'}$. Then we set $\pi'(C_i)=t'$. Notice that $\cset'(\pi')\subseteq \cset'(\pi)\cup \cset_1$, and, if we discard from $\rho_{\pi}$ and from $\rho_{\pi'}$ the columns of $\cset_1$, then the two orderings are defined over the same set of columns and they are identical.

It is easy to verify that the resulting profile $\pi'$ must be legal for $t$. Therefore, profile $\pi'$ was considered by the algorithm. 
Since $\hat Z'$ and $P$ are not conflicting, it is easy to verify that for any other set $\tilde Z$ of points that lie on rows $R_1,\ldots,R_t$ and are consistent with profile $\pi'$, set $\topset(\tilde Z)$ does not conflict with $P$. Therefore, $\pi'$ must be a candidate profile for $\pi$. We conclude that $|T[t,\pi']|\leq |\hat Z|$, and  $|T[t,\pi]|\leq |T[t,\pi']|+|P|=|Z^*|$, so $|Z|\leq |Z^*|$ must hold.

The output of the algorithm is a set $T[m,\pi]\setminus X$ of smallest cardinality among all profiles $\pi\in \Pi$ that are consistent with $m$. It is immediate to verify that this is a feasible and optimal solution for $X$.

As observed before, the number of entries in the dynamic programming table is $m\cdot n^{O(n)}$, and computing each entry takes time $n^{O(n)}$. Therefore, the total running time of the algorithm is bounded by $m\cdot n^{O(n)}$.

\subsection{Leaf Instances: Proof of Theorem \ref{thm: leaf problem2}} 

In this section we provide we prove \ref{thm: leaf problem2}. Recall that the input is a semi-permutation point set $X$ and a partitioning tree $T$ for $X$. Our goal is to compute a $T$-special solution for $X$, of cost at most $2 m h$ where $m = |X|$ and $h = \height(T)$ in time $\poly(m)$.

For every point $p\in X$, we define a new set of points $Z(p)$, as follows. For every vertex $v\in V(T)$ of the tree with $p\in S(v)$, we add two points to $Z(p)$, that have the same $y$-coordinate as $p$: one lying on the left boundary of $S(v)$, and one lying on its right boundary; we call these two points the \emph{projection of $p$ to the boundaries of $S(v)$}. We then set $Y=\bigcup_{p\in X}Z(p)$.

First, it is easy to see that the solution is $T$-special, as $Y$ only contains points on boundaries of strips $S(v)$, for $v\in V(T)$. 
Moreover, for every point $p\in X$, the set of all vertices $v\in V(T)$ with $p\in S(v)$ lie on a single root-to-leaf path of the tree $T$. Therefore, $|Z(p)|\leq 2h$, and  $|Y| \leq 2 m h$. 

Finally, we argue that the set $X\cup Y$ of points is satisfied.
Consider any pair of points $p,q \in X \cup Y$, where $p.x < q.x$. We assume that and $p.y < q.y$ (other case is symmetric).
Assume first that at least one of these two points, say $p$, lies in $X$. Let $v\in V(T)$ be the unique leaf vertex with $p\in S(v)$. 
Then the copy of $p$ lying on the right boundary of $S(v)$ was added to $Z(p)$, and this point lies in $\rect_{p,q}\setminus\set{p,q}$.

Therefore, we can assume that both $p$ and $q$ lie in $Y$. Denote by $p'$ and $q'$ be the points in $X$, such that $p\in Z(p')$ and $q\in Z(q')$. 
 
Let $v$ be the node of the tree for which $S(v)$ contains both $p'$ and $q'$, such that the width of $S(v)$ is the smallest.
Let $u$ and $w$ be children of $v$ such that $S(u)$ contains $p'$ and $S(w)$ contains $q'$, and $L$ be the vertical line in the middle that splits $S(v)$ into $S(u)$ and $S(w)$. 
Assume first that at least one of the two points, $p$ or $q$ (say $p$) lies either on the boundary of $S(v)$, or outside $S(v)$. Since for every ancestor $x$ of vertex $v$, both $p',q'\in S(x)$, both $p'$ and $q'$ were projected to the boundaries of $S(x)$. Therefore, there is a copy of $q'$ on the same column where $p$ lies, satisfying $\rect_{p,q}$. 

We can therefore assume from now on that both $p$ and $q$ lie in the interior of $S(v)$. But then $p$ must lie $S(u)$ and $q$ must lie in $S(w)$. Since both $p',q'$ are projected to the line $L=L(v)$, the pair $(p,q)$ must be satisfied.

\begin{figure}[h]
    \centering
    \includegraphics[width=0.5\textwidth]{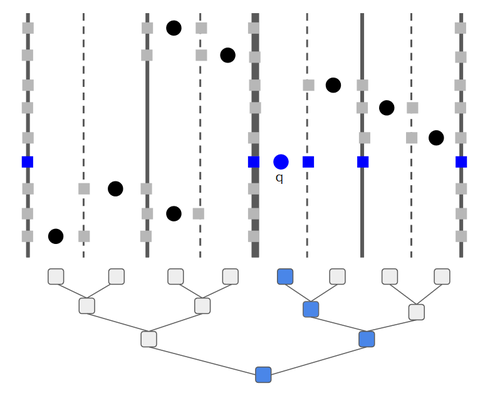}
    \caption{An illustration of the algorithm. Point $q$ is projected to the boundaries of every strip $S(v)$ that contains $q$.}
    \label{fig:static-opt}
\end{figure}

We remark that this approximation algorithm corresponds to an algorithm that uses a static tree $T$ in the classical BST view.

This completes the proof of Theorem \ref{thm:intro_alg}, except for an $O(\log \log n)$-competitive online algorithm. We show that our $O(\log\log n)$-approximation algorithm can be adapted to the online setting in Section \ref{sec: online}. %

\newcommand{\first}{{\sf first}} 

\newcommand{\last}{{\sf last}}

\newcommand{\proj}{{\sf proj}} 

\section{An $O(\log \log n)$-Competitive Online Algorithm}  
\label{sec: online}

In this section we extend the $O(\log\log n)$-approximation algorithm from the previous section to the online setting, obtaining an $O(\log\log n)$-competitive algorithm and completing the proof of Theorem \ref{thm:intro_alg}. The recursive description of the algorithm from the previous section is however not very convenient for the online setting. In subsection~\ref{subsec: unfold}, we present an equivalent iterative description of the algorithm. In subsection~\ref{sec: online-friendly}, we slightly modify the solution $Y$ that the algorithm returns to obtain another solution $\hat{Y}$ that is more friendly for the online setting, before presenting the final online algorithm in subsection \ref{subsec: online final}.

\subsection{Unfolding the recursion} 
\label{subsec: unfold}  
 
Let $X$ be an input set of points that is semi-permutation, with $|X|=m$, and $c(r)=n$. Let $T$ be a balanced partitioning tree of height $H=O(\log n)$ for $X$. %

We now construct another tree $R$, that is called a \emph{recursion tree}, and which is unrelated to the tree $T$. Every vertex $q$ of the tree $R$ is associated with an instance $\iset(q)$ of \minsat that arose during the recursive execution of Algorithm {\sc RecursiveBST}($X,T,\rho$), with $\rho=1$. 
Specifically, for the root $r$ of the tree $R$, we let $\iset(r)=X$. For every vertex $q$ of the tree $R$, if Algorithm  {\sc RecursiveBST}($X,T,\rho$), when called for instance $\iset(q)$, constructed instances $\iset_1,\ldots,\iset_z$ (recall that one of these instances is a compressed instance, and the remaining instances are strip instances), then vertex $q$ has $z$ children in tree $R$, each of which is associated with one of these instances.

For a vertex $q\in V(R)$, let $n(q)$ be the number active columns in the instance $\iset(q)$. 
Recall that instance $\set(q)$ corresponds to some sub-tree of the partitioning tree $T$, that we denote by $T_q$. 
For all $i\geq 0$, let $\Lambda'_i$ be the set of all vertices of the tree $R$ that lie at distance exactly $i$ from the root of $R$. We say that the vertices of $\Lambda'_i$ belong to the $i$th layer of $R$. Notice that, if vertex $q$ lies in the $i$th layer of $R$, then the height of the corresponding tree $T_q$ is bounded by $H\cdot (2/3)^i$ (the constant $2/3$ is somewhat arbitrary and is used because, when a tree $T'$ is split in its middle layer, each of the resulting subtrees has height at most $\floor{\height(T')/2}+1\leq 2H/3$).
Recall that the recursion terminates once we obtain instances whose corresponding partitioning trees have height $1$. It is then easy to verify that the height of the recursion tree $R$ is bounded by $O(\log\log n)$.

Consider now some layer $\Lambda_i$ of the recursion tree $R$. We let $\tset_i$ be the collection of all sub-trees of the partitioning tree $T$ corresponding to the vertices of $\Lambda_i$, so $\tset_i=\set{T_q\mid q\in \Lambda_i}$. Recall that all trees in $\tset_i$ have height at most $H\cdot (2/3)^i$.

Notice that $\tset_0 = \set{T}$. 
Let $U=\set{v_1,\ldots, v_k}$ be the middle layer of $T$. Then $\tset_1 = \set{T_{v_1},\ldots, T_{v_k},T_r}$. Set $\tset_2$ is similarly obtained by subdividing every tree in $\tset_1$, and so on. (See Figure~\ref{fig:Trees} for an illustration. )
The construction of the tree sets $\tset_i$ can be described using the following process:  

\begin{itemize}
\item Start from $\tset_0 = \set{T}$. 

\item For $i = 1,\ldots, D$, if some tree in $\tset_{i-1}$ has height greater than $1$, then construct the set $\tset_i$ of trees as follows: For every tree $T' \in \tset_i$ whose height is greater than $1$, split $T'$ at the middle layer $U_{T'} \subseteq V(T')$. For every vertex $v\in U_{t'}$, we then obtain a sub-tree of $T'$ rooted at $v$, that is added to $\tset_i$ -- this tree corresponds to a strip instance. Additionally, we obtain a sub-tree of $T'$ rooted at the root of $T'$, whose leaves are the vertices of $U_{t'}$, that we also add to $\tset_i$ -- this tree corresponds to the compressed instance.
\end{itemize}

The following observation is immediate.

\begin{observation} 
For each $1\leq i\leq D$, $\bigcup_{T' \in \tset_i} V(T') = V(T)$, and for each pair $T', T'' \in \tset_i$ of trees, either $V(T') \cap V(T'') = \emptyset$ or the root of one of these trees is a leaf of the other tree.
\end{observation}

\begin{figure}
    \centering
    \includegraphics[width=0.9\textwidth]{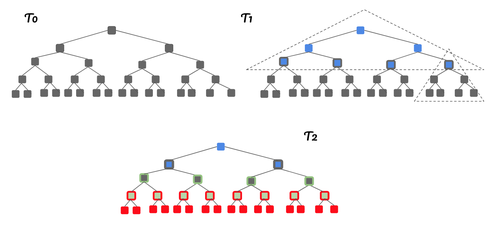}
    \caption{An illustration of the families $\tset_i$. Notice that $\tset_1$ contains $5$ trees, each having height $2$.}
    \label{fig:Trees}
\end{figure}

\paragraph{Boxes:}
Fix an index $1\leq i\leq D$, and consider any tree $T' \in \tset_i$.  
Denote the number of leaves by $k$, and let $v_1,\ldots, v_k$ be the leaves of $T'$. 
If $v$ is the root vertex of the tree $T'$, then we can view $T'$ as defining a hierarchical partitioning scheme of the strip $S(v)$, until we obtain the partition $(S(v_1),S(v_2), \ldots ,S(v_k))$ of $S(v)$.  
We define a collection of $T'$-boxes as follows. Let $X' = X \cap S(v)$ be the set of all points lying in strip $S(v)$. 
Assume that $X' = \set{p_1,p_2, \ldots, p_{m'}}$, where the points are indexed in the increasing order of their $y$-coordinates.

We now iteratively partition the point set $X'$ into boxes, where each box is a consecutive set of points of $X'$. 
Let $v_i$ be the leaf vertex of $T'$, such that $p_1\in S(v_i)$, and let $m_1$ be the largest index, such that all points $p_1,\ldots,p_{m_1}$ lie in $S(v_i)$. We then then define a box $B_1 = \set{p_1, p_2,\ldots, p_{m_1}}$. We discard the points $p_1,\ldots,p_{m_1}$, and continue this process to define $B_2$ (starting from point $p_{m_1}+1$), and so on, until every point of $X'$ belongs to some box.

We let $\bset(T') = \set{B_1,B_2,\ldots, B_{z(T')}}$ be the resulting partition of the points in $X'$, where we refer to each set $B_i$ as a $T'$-box, or just a box.
For each box $B' \in \bset(T')$, the lowest and the highest rows containing points of $B'$ are denoted by $\first(B')$ and $\last(B')$ respectively.

\paragraph{Projections of points:} Recall that the solution that our recursive algorithm returns is $T$-special, that is, the points that participate in the solution lie on the active rows of $X$ and on $T$-auxiliary columns. 
Let $Y$ be a feasible solution obtained by our algorithm {\sc RecursiveBST}($X$, $T$, $\rho$), where $\rho=1$. 
Notice the every point in $Y$ is obtained by ``projecting'' some input point $p \in X$ to the boundary of some strip $S(v)$ for $v \in V(T)$, where strip $S(v)$ contains $p$.  
For each point $p \in X$ and node $v \in V(T)$ of the partitioning tree, such that $p \in S(v)$, we define the set $\proj(p,v)$ to contain two points on the same row as $p$, that lie on the left and right boundaries of $S(v)$, respectively.
We denote by $Y_{p,v}$ the set that contain the two points $\proj(p,v)$ if  our algorithm adds these two points to the solution. Notice that,  since all points of $X$ lie on distinct rows, for any two points $p\neq p'$, for any pair $v,v'\in V(T)$ of vertices, $\proj(p,v)\cap \proj(p',v')=\emptyset$.
Observe that we can now write the solution $Y$ as $Y = \bigcup_{p \in X} \bigcup_{v \in V(T)} Y_{p,v}$.   
The following lemma characterizes the points of $Y$ in terms of boxes.  

\begin{lemma} 
For every point $p \in X$ and every node $v \in V(T)$ of the partitioning tree, $|Y_{p,v}| =2$ if and only if there is an index $1\leq i\leq D$, and a subtree $T' \in \tset_i$, such that (i) $p$ is the first or the last point of its $T'$-box; (ii) $v$ lies in the middle layer of $T'$; and (iii) $S(v)$ contains $p$.
\end{lemma} 

\begin{proof}
	Let $1\leq i\leq D$ be an index, and let $T'\in \tset_i$ be the corresponding partitioning tree. We denote the corresponding point set by $X'$. Let $v_1,\ldots,v_k$ be the leaf vertices of $T'$, and let $u_1,\ldots,u_r$ be the vertices lying in the middle layer of $T'$. The only points added to the solution when instance $X'$ is processed are the following: for every $1\leq j\leq r$, for every point $p\in X'\cap S(u_j)$, we add the two copies of $p$ to the boundaries of $S(u_j)$. In subsequent, or in previous iterations, we may add points to boundaries of $S(u_j)$, but we will never add two copies of the same input point to both boundaries of $S(u_j)$. Therefore, if $|Y_{p,v}|=2$, then there must be an index $i$, and a tree $T'\in \tset_i$, such that $v$ lies in the middle layer of $T'$, and $S(v)$ contains $p$. Observe that for every $T'$-box $B$, instance $X'$ only contains the first and the last point of $B$; all remaining points are redundant for $X'$, and such points are not projected to the boundaries of their strips. Therefore, $p$ must be the first or the last point of its $T'$-box.
	\end{proof}

\paragraph{An equivalent view of our algorithm:}
We can think of our algorithm as follows.
First, we compute all families $\tset_0,\tset_1, \ldots, \tset_D$ of trees, and for each tree $T'$ in each family $\tset_i$, all $T'$-boxes.
For each point $p \in X$, for each vertex $v \in V(T)$ of the partitioning tree $T$, such that $p\in S(v)$, we add the projection points $\proj(p,v)$ to the solution, depending on whether there exists an index $i \in \set{0,1,\ldots, D}$ and a tree $T' \in \tset_i$, such that $v$ lies in the middle layer $U_{T'}$ of $T'$,  and whether $p$ is the first or last point of its $T'$-box. Notice that in the online setting, when the point $p$ arrives, we need to be able to immediately decide which copies of $p$ to add to the solution. Since the trees in the families $\tset_0,\tset_1,\ldots,\tset_D$ are known in advance, and we discover, for every vertex $v\in V(T)$ whether $p\in S(v)$ immediately when $p$ arrives, the only missing information, for every relevant tree $T'$, is whether vertex $p$ is the first or the last vertex in its $T'$-box. In fact it is easy to check whether it is the first vertex in its $T'$-box, but we will not discover whether it is the last vertex in its $T'$-box until the next iteration, at which point it is too late to add points to the solution that lie in the row of $p$.

In order to avoid this difficulty, we slightly modify the instance and the solution.

\subsection{Online-friendly solutions}
\label{sec: online-friendly}

In this section we slightly modify both the input set of points and the solutions produced by our algorithm, in order to adapt them to the online setting.

\paragraph{Modifying the Instance.}
Let $X=\set{p_1,\ldots,p_m}$ be the input set of points, where for all $1\leq i\leq m$, the $y$-coordinate of $p_i$ is $i$. We produce a new instance $X'=\set{p_i',p_i''\mid 1\leq i\leq m}$, as follows: for all $1\leq i\leq m$, we let $p'_i$ and $p''_i$ be points whose $y$-coordinates are $(2i-1)$ and $2i$, respectively, and whose $x$-coordinate is the same as that of $p_i$. We refer to $p'_i$ and to $p''_i$ as \emph{copies of $p_i$}.
	
	Clearly, $|X'|=2|X|$, and it is easy to verify that $\opt(X')\leq 2\opt(X)$. Indeed, if we let $Y$ be a solution for $X$, we can construct a solution $Y'$ for $X'$, by creating, for every point $q\in Y$, two copies $q'$ and $q''$, that are added to rows with $y$-coordinates $2q.y-1$ and $2q.y$ respectively.

For convenience, we denote $\tset=\bigcup_{i=0}^D\tset_i$. Notice that, for every tree $T'\in \tset$, for every point $p_i$ of the original input $X$, copy $p'_i$ of $p_i$ may never serve as the last point of a $T'$-box, and copy $p''_i$ of $p_i$ may never serve as the first point of a $T'$-box.

\paragraph{Modifying the Solution.}
Let $Y$ be the solution that our $O(\log\log n)$-approximation algorithm produces for the new instance $X'$. For convenience, for all $1\leq i\leq 2m$, we denote by $R_i$ the row with $y$-coordinate $i$. Notice that all points of $Y$ lying on a row $R_{2i-1}$ are projections of the point $p'_i$. Since this point may only serve as the first point of a $T'$-box for every tree $T'\in \tset$, when point $p'_i$ arrives online, we can immediately compute all projections of $p'_i$ that need to be added to the solution.
All points of $Y$ lying on row $R_{2i}$ are projections of the point $p''_i$. This point may only serve as the last point of a $T'$-box in a tree $T'\in \tset$. But we cannot know whether $p''_i$ is the last point in its $T'$-box until we see the next input point. Motivated by these observations, we now modify the solution $Y$ as follows.

We perform $m$ iterations. In iteration $i$, we consider the row $R_{2i}$. If no point of $Y$ lies on row $R_{2i}$, then we continue to the next iteration. Otherwise, we move every point of $Y$ that lies on row $R_{2i}$ to row $R_{2i+1}$ (that is, one row up). Additionally, we add another copy $p'''_i$ of point $p_i$ to row $R_{2i}$, while preserving its $x$-coordinate.

In order to show that the resulting solution is a feasible solution to instance $X'$, it is sufficient to show that the solution remains feasible after every iteration. Let $Y_i$ be the solution $Y$ obtained before the $i$th iteration, and let $S_i=X'\cup Y_i$. We can obtain the new solution $Y_{i+1}$ equivalently as follows. First, we collapse the rows $R_{2i+1}$ and $R_{2i}$ for the set $S_i$ of points into the row $R_{2i+1}$, obtaining a new set $S'_i$ of points that is guaranteed to be satisfied. Notice that now both row $R_{2i+1}$ and row $R_{2i-1}$ contain a point at $x$-coordinate $(p_i).x$, while row $R_{2i}$ contains no points. Therefore, if we add to $S'_i$ a point with $x$-coordinate $(p_i).x$, that lies at row $R_{2i}$, then the resulting set of points, that we denote by $S_{i+1}$ remains satisfied. But it is easy to verify that $S_{i+1}=X'\cup Y_{i+1}$, where $Y_{i+1}$ is the solution obtained after iteration $i$.
We denote by $Y'$ the final solution obtained by this transformation of $Y$. It is easy to see that $|Y'|\leq 2|Y|\leq O(\log\log n)\opt(X)$.

\subsection{The Final Online Algorithm} \label{subsec: online final}
We now provide the final online algorithm, that is described in two steps.

First, we assume that in the online sequence, all input points are doubled, that is, for $1\leq i\leq m$, both in iterations $(2i-1)$ and $2i$ the algorithm receives copies of the input point $p_i$; the point received at iteration $(2i-1)$ is denoted by $p'_i$, and the point received at iteration $2i$ is denoted by $p''_i$. The algorithm computes the collection $\tset$ of trees at the beginning. When a point $p'_i$ arrives, we can identify, for every tree $T'\in \tset$, whether $p'_i$ is the first point in any $T'$-box. We then add projections of the point $p'_i$ to the corresponding strips $S(v)$ as needed, all on row $R_{2i-1}$. When a point $p''_i$ arrives, we only discover whether it is the last point of any $T'$-box, for all $T'\in \tset$, in the following iteration. Therefore, in iteration $2i+1$, we add both the copy $p'''_i$ of $p''_i$, and all its relevant projections, to row $R_{2i+1}$. Note that at the same time, we may also add projections of point $p'_{i+1}$ to row $R_{2i+1}$. It is easy to verify that the resulting solution is precisely $Y'$, and therefore it is a feasible solution. The cost of the solution is bounded by $O(\log\log n) \cdot \opt(X)$.

Lastly, for all $1\leq i\leq m$, we unify the iterations $(2i-1)$ and $(2i)$ into a single iteration in a straightforward way.

\appendix

\section{Iacono's instance} 
\label{sec: Iacono} 
Iacono \cite{in_pursuit} showed that there is an access sequence $X$, for which $\opt(X) = \Omega(|X| \log |X|)$ but $\WB_T(X) = O(|X|)$, where $T$ is a balanced tree (used by Tango's tree to obtain $O(\log \log n)$ approximation.)

We describe his proof in  the geometric view in more detail here. Let $k$ be an integer that is an integral power of $2$, $n=2^k$, and $\ell=\log k$. Let $\cset$ be a set of columns with integral $x$-coordinates from $1$ to $n$, and let $\lset$ be the set of vertical lines with half-integral coordinates between $1/2$ and $n-1/2$. 
We consider a balanced partitioning tree $T$, where for every inner vertex $v\in V(T)$, the line $L(v)$ that $v$ owns crosses the strip $L(v)$ exactly in the middle. Let $P= (v_1, v_2, \ldots, v_{k})$ be the unique path in tree $T$, obtained by starting from the root $r=v_1$, and then repeatedly following the left child of each current vertex, until a leaf $v_k$ is reached, so $k=\log n$. 

For each $2\leq i\leq k$, let $v'_i$ be the sibling of $v_i$; observe that the strips $\set{S(v'_i)}_{i=2}^k$ are internally disjoint. 
For each such strip $S(v'_i)$, let $C_i$ be any column with an integral $x$-coordinate in that strip. 
This set $\cset = \set{C_2,\ldots, C_{k}}$ of  $\Theta(\log n)$ columns will serve as the set of active columns. 
Let $\rset$ be any collection of $k$ rows, and let $X=\BRS(\ell,\rset,\cset)$ be the Bit Reversal Sequence defined in Section \ref{sec: sequences}. From Claim \ref{claim:BRS complexity}, $\opt(X)\geq \Omega(|X|\log |X|)=\Omega(k\log k)$.

Next, we show that $\WB_T(X)=O(k)$. 
Indeed,  consider any vertex $v\in V(T)$. Notice that $\cost(v)$ is positive if and only if $v \in V(P)\setminus\set{v_k}$. Therefore, 
using Claim \ref{claim: bounds on WB for path and subtree}, we get that $\WB_T(X)\leq \sum_{v \in P} \cost(v) \leq O(k)$.

\section{Binary Search Trees and Equivalence with the Geometric Formulation}
\label{sec: equivalence}

We start by formally defining the BST model.
An \emph{access sequence} $X=\{x_{1},\dots,x_{m}\}$ is a sequence
of keys from the universe $\{1,\dots,n\}$.
We consider algorithms that maintain a Binary Search Tree (BST) $T$ over the set 
 $\{1,\dots,n\}$ of keys, and support only \emph{access} operations.
After each access the algorithm can reconfigure the tree $T$, as follows. (The definitions below are adapted from \cite{DHIKP09}.)

\begin{definition}
	Given two BST's $T_{1}$ and $T_2$, a subtree $\tau$ of $T_{1}$ containing the
	root of $T_1$, and a tree $\tau'$ with $V(\tau)=V(\tau')$, we say that $T_{2}$ is a \emph{reconfiguration of $T_1$ via the operation $\tau\rightarrow\tau'$}, if $T_{2}$ is identical to $T_{1}$ except
	that $\tau$ is replaced by $\tau'$. The \emph{cost} of the reconfiguration
	is  $|V(\tau)|$. 
\end{definition}

\begin{definition}
	Given an access sequence $X=(x_{1},x_{2},\dots,x_{m})$, an \emph{execution sequence} of a BST algorithm on $X$ is $E=(T_{0},\tau_{1}\rightarrow\tau'_{1},\dots,\tau_{m}\rightarrow\tau'_{m})$, where $T_0$ is the initial tree, and for each $1\leq i\leq m$, $\tau_i\rightarrow \tau_{i+1}$ is a valid reconfiguration operation to be executed at step $i$, such that $x_i\in \tau_i$. The cost of execution $E$ is $\sum_{i=1}^{m}|V(\tau_{i})|$. 
\end{definition}

As stated in \cite{DHIKP09}, the cost of an algorithm's execution in this model is identical to within a constant-factor
to other standard models. This includes, for example, the model that we informally defined in the introduction, where, after each access, the tree is modified by a series of rotations, and the cost consists of the number of such rotations and the distance of the accessed key from the tree's root.

The BST problem can now be defined as follows:
\begin{problem}[The BST problem]
	\label{prob:tree view}Given an access sequence $X$, compute an execution sequence for $X$ that has minimum cost.
	In the online version, the $t$-th element of the sequence $X$ is revealed at time $t$, and we need to output a valid reconfiguration operation $\tau_t\rightarrow \tau'_t$ with $x_t\in \tau_t$ before the $(t+1)$-th element is revealed.
\end{problem}

Observe that any access sequence $X$ uniquely corresponds to a set $P$ of points in the plane that is a semi-permutation: if we denote $X=(x_{1},x_{2},\dots,x_{m})$, then for each $1\leq i\leq m$, we add the point $(x_{i},i)$ to $P$.  For convenience, we do not distinguish between the input sequence $X$ and the corresponding set $P$ of points.

\begin{theorem}
	[Implied by Lemmas 2.1, 2.2 and 2.3 from \cite{DHIKP09}]
	
There is an (online) algorithm for BST, that, for any access sequence $X$ computes an execution sequence of cost at most $O(|X|+c(X))$, if and only if there is an (online) algorithm for the \minsat problem, that, given any input $X$, computes a solution $Y$ of cost $|Y|\leq c(X)$. The running time of the two algorithms are the same up to a constant factor.
\end{theorem}

We then obtain the following corollary. 
\begin{corollary}
	There is a $c$-approximation polynomial-time algorithm for the offline version of \Cref{prob:tree view}
	iff there is an $O(c)$-approximation polynomial-time algorithm for the offline \minsat problem. Similarly, there is a $c$-competitive algorithm for the online version of \Cref{prob:tree view}
	iff there is a $O(c)$-competitive algorithm for the online \minsat problem.
\end{corollary}

\section{Proofs Omitted from Preliminaries}

\subsection{Proof of Observation~\ref{obs: number of crossings at most min of both sides}}
 
For simplicity, we denote $X' =  X \cap S(v_1)$ and $X'' = X \cap S(v_2)$. 
Assume w.l.o.g. that $|X'|\leq |X''|$. Notice that, if the pair $(p_i,p_{i+1})$ of points in $S(v)$ define a crossing of $L(v)$, then one of $p_i,p_{i+1}$ must lie in $X'$. 
Every point $p_j\in X'$ may participate in at most two pairs of points that define crossings: the pairs $(p_{j-1},p_j)$ and $(p_j,p_{j+1})$. 
Therefore, the total number of crossings of $L(v)$ is at most $2|X'|$.

\subsection{Proof of Observation~\ref{obs: aligned point}} \label{subsec: proof of Obs: alligned point} 
Since the set $Z$ is satisfied, rectangle $\Box_{p,q} $ must contain at least one point of $Z$ that is distinct from $p$ and $q$. Among all such points, let $r$ be the one with smallest $\ell_1$-distance to $p$. We claim that either $p.x=r.x$, or $p.y=r.y$. Indeed, assume otherwise. Then $p$ and $r$ are non-collinear, but no point of $Z$ lies in $\rect_{p,r}\setminus\set{p,r}$, contradicting the fact that $Z$ is a satisfied point set.

\subsection{Proof of Observation \ref{obs: collapsing columns}} \label{subsec: proof of obs: collapsing col}
It is sufficient to prove the observation for the case where $\cset$ contains two consecutive active columns, that we denote by $C$ and $C'$, which are collapsed into the column $C$. We can then apply this argument iteratively to collapse any number of columns.

Assume for contradiction that the set $S_{|\cset}$ of points is not satisfied, and let $p,q\in S_{|\cset}$ be a pair of points that are not satisfied. Note that, if $p$ and $q$ cannot both lie on the column $C$ in $S_{|\cset}$. Moreover, if both $p$ and $q$ lie to the right, or to the left of the column $C$, then they continue to be satisfied by the same point $r\in S$ that satisfied them in set $S$. We now consider two cases.

Assume first that $p$ lies to the left of the column $C$, and $q$ lies to the right of the column $C$ in point set $S_{|\cset}$. Let $r$ be the point that satisfied the pair $(p,q)$ in point set $S$. If $r$ lied on column $C$ in $S$, then it remains on column $C$ in $S_{|\cset}$. If $r$ lied on column $C'$ in $S$, then a copy of $r$ lies on column $C$ in $S_{|\cset}$, and this copy continues to satisfy the pair $(p,q)$. Otherwise, point $r$ belongs to set $S_{|\cset}$, and it continues to satisfy the pair $(p,q)$.

It now remains to consider the case when exactly one of the two points (say $p$) lies on the column $C$ in $S_{|\cset}$. Assume w.l.o.g. that $q$ lies to the right of $p$ and below it in $S_{|\cset}$. Then either $p$ belongs to $S$ (in which case we denote $p'=p$), or $p$ is a copy of some point $p'$ that lies on column $C'$ in $S$. Let $r$ be the point that satisfies the pair $(p',q)$ of points in $S$. Using the same reasoning as before, it is easy to see that either $r$ belongs to $S_{|\cset}$, where it continues to satisfy the pair $(p,q)$ of points, or  a copy of $r$ belongs to $S_{|\cset}$, and it also continues to satisfy the pair $(p,q)$.

It is easy to verify that an analogue of Observation \ref{obs: collapsing columns} holds for collapsing rows as well.

\subsection{Proof of Observation~\ref{obs: canonical solutions}} 
 \label{subsec: proof of obs: canonical solutions}
Let $C$ and $C'$ be any pair of consecutive active columns for $X$, such that some point of $Y$ lies strictly between $C$ and $C'$. Let $\cset$ be the set of all columns that lie between $C$ and $C'$, including $C$ but excluding $C'$, that contain points of $X\cup Y$. We collapse the columns in $\cset$ into the column $C$, obtaining a new feasible solution for instance $X$ (we use Observation \ref{obs: collapsing columns}). We continue this process until every point of the resulting solution $Y$ lies on an active column, and we perform the same procedure for the rows.

\subsection{Proof of Claim~\ref{claim: upper bound on OPT static}}
Let $X$ be a given semi-permutation.
We will construct a solution $Y$ for $X$, with $|Y| = O(r(X) \log c(X))$. In order to do so, we
consider a \emph{balanced} partitioning tree $T$ of $X$ with depth $O(\log c(X) )$: in other words, if $v$ is an inner vertex of $T$, and $u,u'$ are its two children, then $\width(u),\width(u')\geq 3\width(v)/4$.
Let $v$ be the root vertex of $T$ and $L(v)\in \lset$ be vertical line that $v$ owns. Let $T_1$ and $T_2$ be the subtrees of $T$ rooted at the two children of $v$.
Notice that $L(v)$ partitions the input set $X$ of points into two subsets, $X_1$ and $X_2$.

We start by adding, for every point $p \in X$, the point $p'$, whose $y$-coordinate is $p.y$, and that lies on $L(v)$, to the sollution $Y$. We then recursively compute a solution $Y_1$ for $X_1$, using the partitioning tree $T_1$, and similarly a solution $Y_2$ for $X_2$, using the partitioning tree $T_2$. The final solution is obtained by adding the points of $Y_1\cup Y_2$ to the set $Y$. It is easy to verify that $Y$ is indeed a feasible solution. Moreover, each active row of $X$ may contain at most $O(\log c(X) )$ points of $Y$, and so $|Y|\leq O(r(X) \log c(X))$.

\subsection{Proof of \Cref{claim: bounding WB by OPT}} 
\label{sec: WB proof}

It is sufficient to prove that, if $X$ is a semi-permutation, and $T$ is any partitioning tree for $X$, then $\WB_T(X) \leq 2\opt(X)$. %

We prove this claim by induction on the height of $T$. 
The base case, when the height of $T$ is $0$, is obvious: there is only one active column, and so $\WB_T(X) = \opt(X) = 0$.

We now consider the inductive step.
Let $X$ be any point set that is a semi-permutation, and let $T$ be any partitioning tree for $X$, such that the height of $T$ is at least $1$. Let $v$ be the root vertex of $T$, $L=L(v)$ the line that $v$ owns, and let $v_L,v_R$ be the two children of $v$. We assume w.l.o.g. that 
the strip $S(v_L)$ lies to the left of $S(v_R)$. We denote $S(v)=B$ -- the bounding box of the instance, $S(v_L)=S^L,S(v_R)=S^R$, and we also denote $X^L=X\cap S^L,X^R=X\cap S^R$. Lastly, we let $T^L$ and $T^R$ be the sub-trees of $T$ rooted at $v_L$ and $v_R$, respectively.
We prove the following claim.

\begin{claim}\label{claim: ind step}
	\[\opt(X)\geq \opt(X^L)+\opt(X^R)+\cost(v)/2.\]
	\end{claim}

Notice that, if the claim is correct, then we can use the induction hypothesis on $X^L$ and $X^R$ with the trees $T^L$ and $T^R$ respectively, to conclude that:

\[\opt(X)\geq \half \left(\WB_{T^L}(X_L)+\WB_{T^R}(X_R)+\cost(v)\right )=\half \WB_T(X).\]

Therefore, in order to complete the proof of \Cref{claim: bounding WB by OPT}, it is enough to prove Claim \ref{claim: ind step}.

\begin{proofof}{Claim \ref{claim: ind step}}
	Let $Y$ be an optimal solution to instance $X$. We can assume w.l.o.g. that $Y$ is a canonical solution, so no point of $Y$ lies on the line $L$. Let $Y^L,Y^R$ be the subsets of points of $Y$ that lie to the left and to the right of the line $L$, respectively.
	
Let $\rset^L$ be the set of all rows $R$, such that (i) no point of $X^L$ lies on $R$, and (ii) some point of $Y^L$ lies on $R$. We define a set $\rset^R$ of rows similarly for instance $X^R$. The crux of the proof is the following observation.

\begin{observation}\label{obs number of inactive rows with points}
\[	|\rset^L|+|\rset^R|\geq \cost(v)/2.\]
\end{observation}

Before we	prove Observation \ref{obs number of inactive rows with points}, we show that Claim \ref{claim: ind step} follows from it. In order to do so, we will define a new feasible solution $\hat Y^L$ for instance $X^L$, containing at most $|Y^L|-|\rset^L|$ points, and similarly, we will define a new feasible solution $\hat Y^R$ for instance $X^R$, containing at most $|Y^R|-|\rset^R|$ points. This will prove that $\opt(X^L)\leq |Y^L|-|\rset^L|$ and $\opt(X^R)\leq |Y^R|-|\rset^L|$, so altogether, $\opt(X^L)+\opt(X^R) \leq |Y|-|\rset^L|-|\rset^R|\leq \opt(X)-\cost(v)/2$, thus proving Claim \ref{claim: ind step}.

We now show how to construct the solution $\hat Y^L$ for instance $X^L$. The solution $\hat Y^R$ for instance $X^R$ is constructed similarly.

In order to construct the solution $\hat Y^L$, we start with the solution $Y^L$, and then gradually modify it over the course of $|\rset^L|$ iterations, where in each iteration we reduce the number of points in the solution $Y^L$ by at least $1$, and we eliminate at most one row from $\rset^L$. In order to execute an iteration, we select two rows $R,R'$, with the following properties:

\begin{itemize}
	\item Row $R$ contains a point of $X^L$;
	\item Row $R'$ contains a point of $Y^L$ and it contains no points of $X^L$; and
	\item No point of $X^L\cup Y^L$ lies strictly between  rows $R$ and $R'$.
\end{itemize}

Note that, if $\rset^L\neq\emptyset$, such a pair of rows must exist. We then collapse the row $R'$ into the row $R$, obtaining a new modified solution to instance $X^L$ (we use Observation \ref{obs: collapsing columns}). We claim that the number of points in the new solution decreases by at least $1$. In order to show this, it is sufficient to show that there must be two points $p\in R$, $p'\in R'$ with the same $x$-coordinates; after the two rows are collapsed, these two points are mapped to the same point. Assume for contradiction that no such two points exist. Let $p\in R$, $p'\in R'$ be two points with smallest horizontal distance. Then it is easy to see that no point of $X^L\cup Y^L$ lies in the rectangle $\rect_{p,p'}$, contradicting the fact that $Y^L$ is a feasible solution for $X^L$.

In order to complete the proof of Claim \ref{claim: ind step}, it is now enough to prove Observation \ref{obs number of inactive rows with points}.

\begin{proofof}{Observation \ref{obs number of inactive rows with points}}
We denote $X=\set{p_1,\ldots,p_m}$, where the points are indexed in the increasing order of their $y$-coordinates. Recall that a pair $(p_i,p_{i+1})$ of points is a crossing, if the two points lie on opposite sides of the line $L$. We say that it is a \emph{left-to-right} crossing if $p_i$ lies to the left of $L$, and we say that it is a \emph{right-to-left} crossing otherwise. Clearly, either at least half the crossings of $L$ are left-to-right crossings, or at least half the crossings of $L$ are right-to-left crossings. We assume w.l.o.g. that it is the former. Let $\Pi$ denote the set of all left-to-right crossings of $L$, so $|\Pi|\geq \cost(v)/2$. Notice that every point of $X$ participates in at most one crossing in $\Pi$. We will associate, to each crossing $(p_i,p_{i+1})\in \Pi$, a unique row in $\rset^L\cup \rset^R$. This will prove that $|\rset^L|+|\rset^R|\geq |\Pi|\geq \cost(v)/2$.

Consider now some crossing $(p_i,p_{i+1})$. Assume that $p_i$ lies in row $R$, and that $p_{i+1}$ lies in row $R'$. Let $\rset_i$ be a set of all rows lying between $R$ and $R'$, including these two rows. We will show that at least one row of $\rset_i$ lies in $\rset^L\cup \rset^R$. In order to do so, let $H$ be the closed horizontal strip whose bottom and top boundaries are $R$ and $R'$, respectively. Let $H^L$ be the area of $H$ that lies to the left of the line $L$, and that excludes the row $R$ -- the row containing the point $p_i$, that also lies to the left of $L$. Similarly, let $H^R$ be the area of $H$ that lies to the right of the line $L$, and that excludes the row $R'$. Notice that, if any point $y\in Y^L$ lies in $H^L$, that the row containing $y$  must belong to $\rset^L$. Similarly, if any point $y'\in Y^R$ lies in $H^R$, then the row containing $y'$ belongs to $\rset^R$. Therefore, it is now sufficient to show that either $H^L$ contains a point of $Y^L$, or $H^R$ contains a point of $Y^R$. Assume for contradiction that this is false. Let $p\in X^L\cup Y^L$ be the point lying on the row $R$ furthest to the right (such a point must exist because we can choose $p=p_i$). Similarly, let $p'\in X^R\cup Y^R$ be the point lying on the row $R'$ furthest to the left (again, such a point must exist because we can choose $p'=p_{i+1}$.) But if $H^L$ contains no points of $Y^L$, and $H^R$ contains no points of $Y^R$, then no points of $X\cup Y$ lie in the rectangle $\rect_{p,p'}$, and so the pair $(p,p')$ of points is not satisfied in $X\cup Y$, a contradiction.
\end{proofof}	
\end{proofof}

\section{Proof of Lemma \ref{lem: GB is lb for opt}} \label{sec: bound for GB}

In this section we prove Lemma \ref{lem: GB is lb for opt}, by showing that for any point set $X$ that is a permutation, $\GB(X)\leq 2\opt(X)$. In order to do so, it is enough to prove that, for any point set $X$ that is a permutation, for any partitioning tree $T$ for $X$, $\GB_T(X) \leq 2\opt(X)$. 
The proof is by induction on the height of $T$, and it is almost identical to the proof of Claim \ref{claim: bounding WB by OPT} for the standard Wilber Bound. When the height of the tree $T$ is $1$, then $|X|=1$, so $\GB(X)=0$ and $\opt(X)=0$.

Consider now a partitioning tree $T$ whose height is greater than $1$. Let $T_1,T_2$ be the two sub-trees of $T$, obtained by deleting the root vertex $r$ from $T$. Let $(X_1,X_2)$ be the partition of $X$ into two subsets given by the line $L(r)$, such that $T_1$ is a partitioning tree for $X_1$ and $T_2$ is a partitioning tree for $X_2$. Notice that, from the definition of the $\GB$ bound:

\[\GB_T(X)=\GB_{T_1}(X_1)+\GB_{T_2}(X_2)+\cost(r).\]

Moreover, from the induction hypothesis, $\GB_{T_1}(X_1)\leq 2\opt(X_1)$ and $\GB_{T_2}(X_2)\leq 2\opt(X_2)$.
Using Claim \ref{claim: ind step} (that can be easily adapted to horizontal partitioning lines), we get that:

\[\opt(X)\geq \opt(X_1)+\opt(X_2)+\cost(r)/2.\]

Therefore, altogether we get that:

\[\GB_T(X)\leq 2\opt(X_1)+2\opt(X_2)+\cost(r)\leq 2\opt(X).  \]

\bibliographystyle{alpha}
\bibliography{ref}

\end{document}